\documentclass[11pt]{article}
\usepackage{waingarten}
\usepackage{thmtools}
\usepackage{thm-restate}
\usepackage{environ}

\providecommand{\email}[1]{\href{mailto:#1}{\nolinkurl{#1}}}

\usepackage{caption}
\title{Finding monotone patterns in sublinear time}
\author{
	Omri Ben-Eliezer\thanks{Tel-Aviv University, email: \email{omrib@mail.tau.ac.il}} 
	\and Cl\'{e}ment L. Canonne\thanks{Stanford University, email: \email{ccanonne@cs.stanford.edu}} 
	\and Shoham Letzter\thanks{ETH Institute for Theoretical Studies, ETH Zurich, email: \email{shoham.letzter@eth-its.ethz.ch}} 
	\and Erik Waingarten\thanks{Columbia University, email: \email{eaw@cs.columbia.edu}}
}
\begin{document}
\maketitle
\begin{abstract}

We study the problem of finding monotone subsequences in an array from the viewpoint of
sublinear algorithms. For fixed $k \in \mathbb{N}$ and $\eps > 0$, we show that the non-adaptive query
complexity of finding a length-$k$ monotone subsequence of $f \colon [n] \to \mathbb{R}$, assuming that $f$ is $\eps$-far from free of such subsequences, is $\Theta((\log n)^{\lfloor \log_2 k \rfloor})$. Prior to our work, the best algorithm for this problem, due to Newman, Rabinovich, Rajendraprasad, and Sohler (2017), made $(\log n)^{O(k^2)}$ non-adaptive queries; and the only 
lower bound known, of $\Omega(\log n)$ queries for the case $k = 2$,
followed from that on testing monotonicity due to Erg\"un, Kannan, Kumar, Rubinfeld, and Viswanathan (2000) and Fischer (2004). 

 \end{abstract}
\thispagestyle{empty}

\newpage
\thispagestyle{empty}
\tableofcontents
\thispagestyle{empty}
\newpage

\newcommand{\dprof}{\textsf{dist-prof}}
\newcommand{\bprof}{\textsf{bin-prof}}
\newcommand{\gap}{\textsf{gap}}
\newcommand{\med}{\mathrm{med}}
\newcommand{\median}{\textsf{median}}

\newcommand{\GreedyDisjointTuples}{\texttt{GreedyDisjointTuples}}
\newcommand{\sfI}{\mathsf{I}}
\newcommand{\Event}{\ensuremath{\mathcal{E}}}
\newcommand{\boldF}{\mathbf{F}}
\newcommand{\val}{\texttt{val}}
\newcommand{\first}{\texttt{first}}
\newcommand{\seg}{\texttt{seg}}
\newcommand{\len}{\texttt{len}}

\setcounter{page}{1}

\section{Introduction}

For a fixed integer $k \in \N$ and a function (or sequence) $f \colon [n] \to \R$, a \emph{length-$k$ monotone subsequence of $f$} is a tuple of $k$ indices, $(i_1, \dots, i_k) \in [n]^k$, such that $i_1 < \dots < i_k$ and $f(i_1) < \dots < f(i_k)$. More generally, for a permutation $\pi \colon [k] \to [k]$, a \emph{$\pi$-pattern of $f$} is given by a tuple of $k$ indices $i_1 < \dots < i_k$ such that $f(i_{j_1}) < f(i_{j_2})$ whenever $j_1, j_2 \in [k]$ satisfy $\pi(j_1) < \pi(j_2)$. A sequence $f$ is $\pi$-free if there are no subsequences of $f$ with order pattern $\pi$. Recently, Newman, Rabinovich, Rajendraprasad, and Sohler~\cite{NRRS17} initiated the study of property testing for forbidden order patterns in a sequence. Their paper was the first to analyze algorithms for finding $\pi$-patterns in sublinear time (for various classes of the permutation $\pi$); additional algorithms and lower bounds for several classes of permutations have later been obtained by Ben-Eliezer and Canonne~\cite{BC18}. 

Of particular interest of $\pi$-freeness testing is the case where $\pi = (12\dots k)$, i.e., $\pi$ is a monotone permutation. In this case, avoiding length-$k$ monotone subsequence may be equivalently rephrased as being decomposable into $k-1$ monotone non-increasing subsequences. Specifically, a function $f \colon [n] \to \R$ is $(12\dots k)$-free if and only if $[n]$ can be partitioned into $k-1$ disjoint sets $A_1, \ldots, A_{k-1}$ such that, for each $i\in[k-1]$, the restriction $f|_{A_i}$ is non-increasing. When interested in algorithms for testing $(12\dots k)$-freeness that have a \emph{one-sided error},\footnote{An algorithm for testing property $\calP$ is said to have \emph{one-sided error} if the algorithm always outputs ``yes'' if $f \in \calP$, i.e., has perfect completeness.} the algorithmic task becomes the following:
\begin{quote}\itshape
	For $k \in \N$ and $\eps > 0$, design a randomized algorithm that, given query access to a function $f \colon [n] \to \R$ guaranteed to be $\eps$-far from being $(12\dots k)$-free,\footnote{A function $f \colon [n] \to \R$ is \emph{$\eps$-far} from $\pi$-free if any $\pi$-free function $g \colon [n] \to \R$ satisfies $\Pr_{\bi \sim [n]}[f(\bi) \neq g(\bi)]\geq \eps$.} outputs a length-$k$ monotone subsequence of $f$ with probability at least $9/10$. 
\end{quote}
The task above is a natural generalization of monotonicity testing of a function $f \colon [n] \to \R$ with algorithms that make a one-sided error, 
 a question which dates back to the early works in property testing, and has received significant attention since in various settings (see, e.g.,~\cite{DGLRRS99,GGLRS00,FLNRRS02,AMW13,BRY14b,Bel18,PRV18, BE19}, and the recent textbook \cite{G17}). For the problem of testing monotonicity, Erg\"{u}n, Kannan, Kumar, Rubinfeld, and Viswanathan \cite{EKKRV00} were the first to give a non-adaptive algorithm which tests monotonicity of functions $f\colon [n] \to \R$ with one-sided error making $O(\log(n) / \eps)$ queries. (Recall that an algorithm is \emph{non-adaptive} if its queries do not depend on the answers to previous queries, or, equivalently, if all queries to the function can be made in parallel.) Furthermore, they showed that $\Omega(\log n)$ queries are necessary for non-adaptive algorithms. Subsequently, Fischer \cite{F04} showed that $\Omega(\log n)$ queries are necessary even for adaptive algorithms. Generalizing from monotonicity testing (when $k=2$), 
 Newman et al.\ gave in~\cite{NRRS17} the first sublinear-time algorithm for $(12\dots k)$-freeness testing, whose query complexity is $(\log(n)/\eps)^{O(k^2)}$. Their algorithm is non-adaptive and has one-sided error; as such, it outputs a length-$k$ monotone subsequence with probability at least $9/10$ assuming the function $f$ is $\eps$-far from $(12\dots k)$-free. However, other than the aforementioned lower bound of $\Omega(\log n)$ which follows from the case $k=2$, no lower bounds were known for larger~$k$. 

The main contribution of this work is to settle the dependence on $n$ in the query complexity of testing for $(12\dots k)$-freeness with non-adaptive algorithms making one-sided error. Equivalently, we settle the complexity of non-adaptively finding a length-$k$ monotone subsequence under the promise that the function $f \colon [n] \to \R$ is $\eps$-far from $(12\dots k)$-free.
\begin{theorem}\label{thm:intro-ub}
	Let $k \in \N$ be a fixed parameter. For any $\eps > 0$, there exists an algorithm that, given query access to a function $f \colon [n] \to \R$ which is $\eps$-far from $(12\dots k)$-free, outputs a length-$k$ monotone subsequence of $f$ with probability at least $9/10$. The algorithm is non-adaptive and makes $(\log n)^{\lfloor \log_2 k \rfloor} \cdot \poly(1/\eps)$ queries to $f$.
\end{theorem}
Our algorithm thus significantly improves on the $(\log(n)/\eps)^{O(k^2)}$-query non-adaptive algorithm of \cite{NRRS17}. Furthermore, its dependence on $n$ is optimal; indeed, in the next theorem we prove a matching lower bound for all fixed $k \in \N$.
\begin{theorem}
	\label{thm:intro-lb}
	Let $k \in \N$ be a fixed parameter. There exists a constant $\eps_0 > 0$ such that any non-adaptive algorithm which, given query access to a function $f\colon [n] \to \R$ that is $\eps_0$-far from $(12\dots k)$-free, outputs a length-$k$ monotone subsequence with probability $9/10$, must make $\Omega((\log n)^{\lfloor \log_2 k \rfloor})$ queries. Moreover, one can take $\eps_0 = 1/(4k)$.
\end{theorem}
\noindent We further note that the lower bound holds even for the more restricted case where $f \colon [n] \to [n]$ is a permutation.

\subsection{Related work}
Testing monotonicity of a function over a partially ordered set $\mathcal{X}$ is a well-studied and fruitful question, with works spanning the past two decades. 
Particular cases include when $\calX$ is the line $[n]$ \cite{EKKRV00,F04,Bel18,PRV18,BE19}, the Boolean hypercube $\{0,1\}^d$~\cite{DGLRRS99,BBM12,BCGM12,CS13,CST14,CDST15,KMS15,BB15,CS16,CWX17, CS18}, and the hypergrid $[n]^d$~\cite{BRY14a,CS14, BCS18}. 
We refer the reader to~\cite[Chapter~4]{G17} for more on monotonicity testing, or for an overview of the field of property testing (as introduced in~\cite{RS96,GGR98}) in general.

This paper is concerned with the related line of work on finding order patterns in sequences and permutations. For the exact case, Guillemot and Marx \cite{GM14} showed that an order pattern $\pi$ of length $k$ can be found in a sequence $f$ of length $n$ in time $2^{O(k^2 \log k)} n$; in particular, the problem of finding order patterns is fixed-parameter tractable (in the parameter $k$). Fox \cite{Fox13} later improved the running time to $2^{O(k^2)} n$. A very recent work of Kozma \cite{Koz19} provides the state-of-the-art for the case where $k = \Omega(\log n)$. 
In the sublinear regime, the most relevant works are the aforementioned papers of Newman et al.~\cite{NRRS17} and Ben-Eliezer and Canonne~\cite{BC18}. In particular, \cite{NRRS17} shows an interesting dichotomy for testing $\pi$-freeness: when $\pi$ is monotone, the non-adaptive query complexity 
is polylogarithmic in $n$ for fixed $k$ and $\eps$, whereas for non-monotone $\pi$, the query complexity is $\Omega(\sqrt{n})$. 

Two related questions are that of estimating the \emph{distance to monotonicity} and the length of the \emph{longest increasing subsequence} (LIS), 
which have also received significant attention from both the sublinear algorithms perspective~\cite{PRR06, ACCL07, SS17}, as well as the streaming perspective~\cite{GJKK07,GG10,SS13,EJ15,NSaks15}.
In particular, Saks and Seshadhri gave in~\cite{SS17} a randomized algorithm which, on input $f \colon [n] \to \R$, makes $\poly(\log n, 1/\delta)$ queries and outputs $\hat{m}$ approximating up to additive error $\delta n$ the length of the longest increasing subsequence of $f$. This paper also studies monotone subsequences of the input function, albeit from a different (and incomparable) end of the problem. Loosely speaking, in \cite{SS17} the main object of interest is a very \emph{long} monotone subsequence (of length linear in $n$), and the task at hand is to get an estimate for its total length, whereas in our setting, there are $\Omega(n)$ disjoint copies of \emph{short} monotone subsequences (of length $k$, which is a constant parameter), and these short subsequences may not necessarily combine to give one long monotone subsequence.
 
\subsection{Our techniques: Upper bound}

We now give a detailed overview of the techniques underlying our upper bound, Theorem~\ref{thm:intro-ub}, and provide some intuition behind the algorithms and notions we introduce. The starting point of our discussion will be the algorithm of Newman et al.~\cite{NRRS17}, which we re-interpret in terms of the language used throughout this paper; this will set up some of the main ideas behind our structural result (stated in Section~\ref{sec:structural}), which will be crucial in the analysis of the algorithm.

For simplicity, let $\eps > 0$ be a small constant and let $k \in \N$ be fixed. Consider a function $f \colon [n] \to \R$ which is $\eps$-far from $(12\dots k)$-free. This implies that there is a set $T \subseteq [n]^k$ of $\eps n/k$ \emph{disjoint} $(12\dots k)$-patterns. Specifically, the set $T$ is comprised of $k$-tuples $(i_1,\dots, i_k) \in [n]$ where $i_1 < \dots < i_k$ and $f(i_1) < \dots < f(i_k)$ and each $i \in [n]$ appears in at most one $k$-tuple in $T$.\footnote{To see why such $T$ exists, take $T$ to be a maximal set of disjoint $(12\dots k)$-patterns. Suppose $|T| < \eps n/k$ and consider the function $g$ given by greedily eliminating all $(12\dots k)$-patterns in $f$, and note that $g$ is $(12\dots k)$-free and differs on $f$ in less than $\eps n$ indices.}
A key observation made in \cite{NRRS17} is that if, for some $c \in [k-1]$, $(i_1, \dots, i_c, i_{c+1}, \dots, i_k)$ and  $(j_1, \dots, j_c, j_{c+1}, \dots, j_k)$ are two $k$-tuples in $T$ which satisfy $i_c < j_{c+1}$ and $f(i_c) < f(j_{c+1})$, then their combination 
\[
    (i_1,\dots, i_c, j_{c+1}, \dots, j_k)
\]
is itself a length-$k$ monotone subsequence of $f$. Therefore, in order to design efficient sampling algorithms, one should analyze to what extent parts of different $(12\dots k)$-tuples from $T$ may be combined to form length-$k$ monotone subsequences of $f$. 

Towards this goal, assign to each $k$-tuple $(i_1, \dots, i_k)$ in $T$ a \emph{distance profile} $\dprof(i_1, \dots, i_k) = (d_1, \dots, d_{k-1}) \in [\eta]^{k-1}$, where $\eta = O(\log n)$.\footnote{We remark that the notion of a distance profile is solely used for the introduction and for explaining \cite{NRRS17}, and thus does not explicitly appear in subsequent sections.} This distance profile is a $(k-1)$-tuple of non-negative integers satisfying
\[ 2^{d_j} \leq i_{j+1} - i_j < 2^{d_j + 1} \qquad\qquad j \in [k-1]\,; \]
and let $\gap(i_1,\dots, i_k) = c \in [k-1]$ be the smallest integer where $d_c \geq d_{j}$ for all $j \in [k-1]$ (i.e., $d_c$ denotes an (approximately) maximum length between two adjacent indices in the $k$-tuple). Suppose, furthermore, that for a particular $c \in [k-1]$, the subset $T_c \subseteq T$ of $k$-tuples whose gap is at $c$ satisfies $|T_c| \geq \eps n / k^2$ (such a $c \in [k-1]$ must exist since the $T_c$'s partition $T$). If $(i_1, \dots, i_k) \in T_c$ and $\dprof(i_1, \dots, i_k) = (d_1, \dots, d_k)$, then the probability that a uniformly random element $\bell$ of $[n]$ ``falls'' into that gap is
\begin{align} 
\Prx_{\bell \sim [n]}\left[ i_{c} \leq \bell \leq i_{c+1}\right] &\geq \frac{2^{d_c}}{n}. \label{eq:prob}
\end{align}
Whenever this occurs for a particular $k$-tuple $(i_1, \dots, i_k)$ and $\bell \in [n]$, we say that $\bell$ \emph{cuts} the tuple $(i_1,\dots, i_k)$. Note that the indices $i_{c+1}, \dots, i_{k}$ are contained within the interval $[\bell, \bell + k \cdot 2^{d_c + 1}]$ and the indices $i_1, \dots, i_c$ are contained within the interval $[\bell - k \cdot 2^{d_c + 1}, \bell]$. As a result, if we denote by $\delta_d(\bell) \in [0,1]$, for each $d \in [\eta]$, the \emph{density} of $k$-tuples from $T_c$ lying inside $[\bell - k \cdot 2^{d + 1}, \bell + k \cdot 2^{d + 1}]$ (i.e., the fraction of this interval comprised of elements of $T_c$), we have
\begin{align}
 \Ex_{\bell \sim [n]} \left[ \sum_{d \in [\eta]} \delta_{d}(\bell) \right] &= \sum_{d \in [\eta]} \sum_{\substack{(i_1,\dots, i_k) \in T_c \\ \dprof(i_1,\dots, i_k)_c =d}} \Prx_{\bell \sim [n]}\left[ i_c \leq \bell \leq i_{c+1}\right] \cdot \frac{1}{2 \cdot k \cdot 2^{d+1}} \gsim \frac{|T_c|}{n} \gsim \eps. \label{eq:expect}
 \end{align}
For any $\ell$ achieving the above inequality, since $\eta = O(\log n)$, there  exists some $d^* \in [\eta]$ such that $\delta_{d^*}(\ell) \gsim \eps / \log n$. Consider now the set of $k$-tuples $T_{c, d^*}(\ell)\subseteq T_c$ contributing to $\delta_{d^*}(\ell)$, i.e., those $k$-tuples in $T_c$ which are cut by $\ell$ and lie in $[\ell - k \cdot 2^{d^*+1}, \ell + k \cdot 2^{d^*+1}]$. Denote $r_{\med} = \median\{ f(i_c) : (i_1, \dots, i_k) \in T_{c, d^*}(\ell)\}$, and let
\begin{align*}
T_{L} &= \left\{ (i_1,\dots, i_c) : (i_1, \dots, i_k) \in T_{c, d^*}(\ell) \text{ and } f(i_c) \leq r_{\med} \right\}, \qquad \text{and}\\
T_{R} &= \left\{ (i_{c+1},\dots, i_k) : (i_1, \dots, i_k) \in T_{c, d^*}(\ell) \text{ and } f(i_c) \geq r_{\med} \right\},
\end{align*}
where we note that $T_L$ and $T_R$ both have size at least $|T_{c, d^*}(\ell)| / 2$. If the algorithm finds a $c$-tuple in $T_{L}$ and a $(k-c)$-tuple in $T_R$, by the observation made in \cite{NRRS17} that was mentioned above, the algorithm could combine the tuples to form a length-$k$ monotone subsequence of $f$. At a high level, one may then recursively apply these considerations on $[\ell - k \cdot 2^{d^*+1}, \ell]$ with $T_L$ and $[\ell, \ell + k \cdot 2^{d^*+1}]$ with $T_R$. A natural algorithm then mimics the above reasoning algorithmically, i.e., samples a parameter $\bell \sim [n]$, and tries to find the unknown parameter $d^* \in [\eta]$ in order to recurse on both the left and right sides; once the tuples have length $1$, the algorithm samples within the interval to find an element of $T_L$ or $T_R$. This is, in essence, what the algorithm from~\cite{NRRS17} does, and this approach leads to a query complexity of $(\log n)^{O(k^2)}$.
In particular, suppose that at each (recursive) iteration, the parameter $c$, corresponding to the gap of tuples in $T$, always equals $1$. Note that this occurs when all $(12\dots k)$-patterns $(i_1,\dots, i_k)$ in $T$ have $\dprof(i_1, \dots, i_k) = (d_1,\dots, d_{k-1})$ with
\begin{align} 
d_1 \geq d_2 \geq \dots \geq d_{k-1}. \label{eq:decreasing}
\end{align}
Then, if $k$ is at $k_0$, a recursive call leads to a set $T_L$ containing $1$-tuples, and $T_{R}$ containing $(k_0 - 1)$-tuples. This only decreases the length of the subsequences needed to be found by $1$ (so there will be $k-1$ recursive calls), while the algorithm pays for guessing the correct value of $d^*$ out of $\Omega(\log n)$ choices, which may decrease the density of monotone $k_0$-subsequences within the interval of the recursive call by a factor as big as $\Omega(\log n)$.\footnote{Initially, the density of $T$ within $[n]$ is $\eps$, and the density of $T_L$ or $T_R$ in $[\ell - k\cdot 2^{d^*+1}, \ell]$ and $[\ell, \ell + k \cdot 2^{d^*+1}]$ is $\eps / \log n$.} As a result, the density of the length-$k_0$ monotone subsequence in the relevant interval could be as low as $\eps / (\log n)^{k_0}$, which means that $(\log n)^{\Omega(k_0)}$ samples will be needed for the $k_0$-th round according to the above analysis, giving a total of $(\log n)^{\Omega(k^2)}$ samples
(as opposed to $O((\log n)^{\lfloor \log_2 k\rfloor})$, which is the correct number, as we prove).\smallskip

In order to overcome the above difficulty, we consider a particular choice of a family $T$ of length-$k$ monotone subsequences given by the ``greedy'' procedure (see Figure~\ref{fig:greedy}). Loosely speaking, this procedure begins with $T = \emptyset$ and iterates through each index $i_1 \in [n] \setminus T$. 
Each time, if $(i_1)$ can be extended to a length-$k$ monotone subsequence (otherwise it continues to the next available index), the procedure sets $i_2$ to be the first index, after $i_1$ and not already in $T$, such that $(i_1, i_2)$ can be extended to a length-$k$ monotone subsequence; then, it finds an index $i_3$ which is the next first index after $i_2$ and not in $T$ such that $(i_1, i_2, i_3)$ can be extended; and so on, until it has obtained a length-$k$ monotone subsequence starting at $i_1$. It then adds the subsequence as a tuple to $T$, and repeats. This procedure eventually outputs a set $T$ of disjoint, length-$k$ monotone subsequences of $f$ which has size at least $\eps n/ k^2$, and satisfies another crucial ``interleaving'' property (see Lemma~\ref{lem:rematching}): 
\begin{quote}
	($\star$) If $(i_1, \dots, i_k)$ and $(j_1, \dots, j_k)$ are $k$-patterns from $T$ and $c \in [k-1]$ satisfy $j_1 < i_1$, $j_c < i_c$, and $i_{c+1} < j_{c+1}$, then $f(i_{c+1}) < f(j_{c+1})$.
\end{quote}
Moreover, a slight variant of~\eqref{eq:prob} guarantees that for any $(i_1, \dots, i_k) \in T_c$ with $\dprof(i_1, \dots, i_k) = (d_1, \dots, d_{k-1})$,
\begin{align*}
\Prx_{\bell \sim [n]}\left[ i_c + 2^{d_c}/3 \leq \bell \leq i_{c+1} - 2^{d_c} / 3\right] \gsim \frac{2^{d_c}}{n}.
\end{align*}
Whenever the above event occurs, we say $\bell \sim [n]$ \emph{cuts} $(i_1,\dots, i_k)$ at $c$ \emph{with slack}, and note that $i_1, \dots, i_c$ lie in $[\bell - k \cdot 2^{d_c+1}, \bell]$ and $i_{c+1}, \dots, i_k$ in $[\bell, \bell + k \cdot 2^{d_c+1}]$. We denote, similarly to the above, $\delta_{d}(\bell) \in [0, 1]$ to be the density of $k$-tuples from $T_c$ which are cut with slack by $\bell$, and conclude~\eqref{eq:expect}. We then utilize ($\star$) to make the following claim: suppose two $k$-tuples $(i_1, \dots, i_k), (j_1, \dots, j_k) \in T_c$ satisfy $\dprof(i_1, \dots, i_k) = (d_1, \dots, d_{k-1})$, and $\dprof(j_1,\dots, j_k) = (d_1', \dots, d_{k-1}')$, where $d_{c} \leq d_{c}' - a\log k$, for some constant $a$ which is not too small. If $(i_1, \dots, i_k)$ and $(j_1, \dots, j_k)$ are cut at $c$ with slack, this means that $\ell$ lies roughly in the middle of $i_{c}$ and $i_{c+1}$ and of $j_c$ and $j_{c+1}$, and since the distance between $i_c$ and $i_{c+1}$ is much smaller than that between $j_c$ and $j_{c+1}$, the index $j_1$ will come before $i_1$, the index $j_c$ will come before $i_{c}$, but the index $i_{c+1}$ will come before $j_{c+1}$. By ($\star$), $f(i_{c+1}) < f(j_{c+1})$ (cf. Lemma~\ref{lem:At:properties}). In other words, the value, under the function $f$, of $(c+1)$-th indices from tuples in $T_{c, d}(\ell)$ increases as $d$ increases. 

As a result, if $\ell \in [n]$ satisfies $\sum_{d \in [\eta]} \delta_{d}(\ell) \gsim \eps$, and $\delta_{d}(\ell) \ll \eps$ for all $d \in [\eta]$, that is, if the summands in~\eqref{eq:expect} are \emph{spread out}, an algorithm could find a length-$k$ monotone subsequence by sampling, for many values of $d \in [\eta]$, indices which appear as the $(c+1)$-th index of tuples in $T_{c, d}(\ell)$. We call such values of $\ell$ the starts of \emph{growing suffixes} (as illustrated in Figure~\ref{fig:growing-suffix}). In Section~\ref{sec:case1}, we describe an algorithm that makes $\tilde{O}(\log n / \eps)$ queries and finds, with high probability, a length-$k$ monotone subsequence if there are many such growing suffixes (see Lemma~\ref{lem:case1}). The algorithm works by randomly sampling $\bell \sim [n]$ and hoping that $\bell$ is the start of a growing suffix; if it is, the algorithm  samples enough indices from the segments $[\ell + 2^{d}, \ell + 2^{d+1}]$ to find a $(c+1)$-th index of some tuple in $T_{c, d}(\ell)$, which gives a length-$k$ monotone subsequence.

The other case corresponds to the scenario where $\ell \in [s]$ satisfies $\sum_{d \in [\eta]} \delta_d(\ell) \gsim \eps$, but the summands are \emph{concentrated} on few values of $d \in [\eta]$. In this case, we may consider a value of $d^* \in [\eta]$ which has $\delta_{d^*}(\ell) \gsim \eps$, and then look at the intervals $[\ell - k \cdot 2^{d^*+1}, \ell]$ and $[\ell, \ell + k \cdot 2^{d^* + 1}]$. We can still define $T_{L}$ and $T_R$, both of which have size at least $|T_{c, d^*}| / 2$ and have the property that any $c$-tuple from $T_L$ can be combined with any $(k-c)$-tuple from $T_R$. Additionally, since $\delta_{d^*}(\ell) \gsim \eps$, we crucially do \emph{not} suffer a loss in the density of $T_L$ and $T_R$ in their corresponding intervals~--~a key improvement over the $\Omega(\log n)$ loss in density incurred by the original approach we first discussed. We refer to these intervals as \emph{splittable intervals} (cf. Figure~\ref{fig:splittable}), and observe that they lead to a natural recursive application of these insights to the intervals $[\ell - k \cdot 2^{d^*+1}, \ell]$ and $[\ell, \ell + k \cdot 2^{d^*+1}]$. The main structural result, given in Theorem~\ref{thm:tree}, does exactly this, and encodes the outcomes of the splittable intervals in an object we term a \emph{$k$-tree descriptor} (see Section~\ref{ssec:tree:descriptors}) whenever there are not too many growing suffixes. Intuitively, a $k$-tree descriptor consists of a rooted binary tree $G$ on $k$ leaves, as well as some additional information, which corresponds to a function $f \colon [n] \to \R$ without many growing suffixes. Each internal node $v$ in $G$ corresponds to a recursive application of the above insights, i.e., $v$ has $k_0$ leaves in its subtree, a parameter $c_v \in [k_0 - 1]$ encoding the gap of sufficiently many $k_0$-tuples, and a collection of disjoint intervals of the form $[\ell - k \cdot 2^{d^*}, \ell + k \cdot 2^{d^*}]$ where $\ell$ cuts $(12\dots k_0)$-patterns with slack at $c_v$ and satisfies~\eqref{eq:expect}; the left child of $v$ has $c$ leaves and contains the $(12\dots c)$-patterns in $T_L$ and intervals $[\ell - k \cdot 2^{d^*}, \ell]$; the right child of $v$ has $k_0 -c$ leaves and contains the $(12\dots (k_0 - c))$-patterns in $T_R$ and intervals $[\ell, \ell + k\cdot 2^{d^*}]$ (see Figure~\ref{fig:tree-descriptor}).

The algorithm for this case is more involved than the previous, and leads to the $O((\log n)^{\lfloor \log_2 k \rfloor})$-query complexity stated in Theorem~\ref{thm:intro-ub}. The algorithm proceeds in $r_0 = 1 + \lfloor \log_2 k \rfloor$ rounds, maintaining a set $\bA \subseteq [n]$, initially empty:
\begin{itemize}
\item \textit{Round 1}: For each $i \in [n]$, include $i$ in $\bA$ independently with probability $\Theta(1/(\eps n))$.
\item \textit{Round $r$, $2 \leq r \leq r_0$}: For each $i \in \bA$ from the previous round, and each $j = 1, \dots, O(\log n)$, consider the interval $B_{i,j} = [i - 2^j, i + 2^j]$. For each $i' \in B_{i,j}$, include $i'$ in $\bA$ independently with probability $\Theta(1/(\eps 2^j))$.
\end{itemize} 
At the end of all rounds, the algorithm queries $f$ at all indices in $\bA$, and outputs a $(12\dots k)$-pattern from $\bA$, if one exists. 

Recall the case considered in the sketch of the algorithm of~\cite{NRRS17}, when the function $f$ has all $(12\dots k)$-patterns $(i_1, \dots, i_{k})$ in $T$ satisfying $\dprof(i_1, \dots, i_k) = (d_1, \dots, d_{k-1})$ with $d_1 \ge d_2 \ge \ldots \ge d_{k-1}$. In this case, the $k$-tree descriptor $G$ consists of a rooted binary tree of depth $k$. The root has a left child which is a leaf (corresponding to $1$-tuples of first indices of some tuples in $T$, stored in $T_L$) and a right child (corresponding to suffixes of length $(k-1)$ of some tuples in $T$, stored in $T_R$) is an internal node. The root node corresponds to one application of the structural result, and the right child corresponds to a $(k-1)$-tree descriptor for the tuples in $T_R$. Loosely speaking, as $d_2 \ge \ldots \ge d_{k-1}$ the same reasoning repeats $k-1$ times, and leads to a path of length $k-1$ down the right children of the tree, the right child of the $(k-1)$-th internal node corresponding to a $1$-tuple (i.e., a leaf).\footnote{This is somewhat inaccurate, as in each step, after forming $T_L$ and $T_R$, we apply the greedy algorithm again and obtain new sets $T_L'$ and $T_R'$, which may violate the assumption $d_1 \ge d_2 \ge \ldots \ge d_k$. We ignore this detail at the moment to simplify the explanation.}

To gain some intuition, we analyze how the algorithm behaves on these instances. Suppose that in round 1, the algorithm samples an element $i \in [n]$ which is the $k$-th index of a $1$-tuple stored in the right-most leaf of $G$. In particular, this index belongs to the set $T_R$ of the $(k-1)$-th internal node, as a second index of a cut $(12)$-pattern in the $(k-1)$-th recursive call of the structural result. Similarly, $i$ also belongs to that set $T_R$ of the $(k-2)$-th internal node, as a part the third index of a cut $(123)$-pattern in the $(k-2)$-th recursive call. We may continue with all these inclusions to the root, i.e., $i$ is the $k$-th element of some $(12\dots k)$-pattern in $T$, which is cut in the first call to the structural result. Round 2 of the algorithm will consider the $k-1$ intervals $B_{i, d_{k-1}'}, B_{i, d_{k-2}'}, \dots, B_{i, d_{1}'}$, where $d_j'= d_j + \Theta(\log k)$, since it iterates through all $O(\log n)$ intervals of geometrically increasing lengths.\footnote{Note that the intervals $B_{i, d_{j}'}$ and $B_{i,d_{j+1}'}$ may be the same, for instance when $d_j = d_{j+1}$.} One can check that for each $j \in [k-1]$, the interval $B_{i, d_{j}'}$ contains $[\ell_j - k \cdot 2^{d_j}, \ell_j]$, where $\ell_j$ is some index which cut the $(k-j+1)$-tuple $(i_j,\dots, i_k)$ with slack in the $j$-th recursive call of the structural result. Recall that the set $T_L$ of $1$-tuples has density $\Omega(\eps)$ inside $[\ell_j - k \cdot 2^{d_j}, \ell_j]$ and may be combined with any $(k-j)$-tuple from $T_R$. Following this argument, in the \emph{second} round of the algorithm, $\bA$ will include some index of $T_L$ (for each $j \in [k-1]$), and these indices combine to form a $(12\dots k)$-pattern~--~that is, with high probability, after two rounds, the algorithm succeeds in finding a monotone subsequence of length $k$. 

Generalizing the above intuition for all possible distance profiles necessitates the use of $1 + \lfloor \log_2 k \rfloor$ rounds, and requires extra care. At a high level, consider an arbitrary $k$-tree descriptor $G$ for $\Omega(\eps n)$ many $(12\dots k)$-patterns in $f$. Denote the root $u$, and consider the unique leaf $w$ of $G$ where the root-to-$w$ path $(u_1, \dots, u_h)$ with $u_1 = u$ and $u_h = w$, satisfies that at each internal node $u_{l}$, the next node $u_{l+1}$ is the child with larger number of leaves in its subtree.\footnote{Ties are broken by picking the left child.} We call such a leaf a \emph{primary index} of $G$. The crucial property of the primary index is that the root-to-leaf path of $w$, $(u_1, u_2, \dots u_h)$, is such that the siblings of the nodes on this  path\footnote{For example, if $(u_1, \dots, u_h)$ is the root-to-$w$ path where $u_1$ is the root and $u_h = w$, the sibling nodes along the path are given by $u_{2}', u_{3}', \dots, u_{h}'$, where $u_{l}'$ is the sibling of $u_l$. Namely, if the $l$-th node on the root-to-$w$ path is a left child of the $(l-1)$th node, then $u_{l}'$ is the right child of the $(l-1)$-th node. Analogously, if the $l$-th node is a right child of the $(l-1)$-th node, then $u_l'$ is the left child of the $(l-1)$-th node.}  have strictly fewer than $k / 2$ leaves in their subtrees. 

 The relevant event in the first round of the algorithm is that of sampling an index $i \in [n]$ which belongs to a $1$-tuple of the primary index $w$ of $G$. This occurs with probability at least $1 - 1/(100 k)$, since we sample each element of $[n]$ with probability $\Theta(1/(\eps n))$ while there are at least $\Omega(\eps n)$ many $(12\dots k)$-patterns. Now, roughly speaking, letting $(u_1, \dots, u_h)$ be the root-to-$w$ path in $G$, and $(u_2', \dots, u_h')$ be the sibling nodes, the subtrees of $G$ rooted at $u_2', \dots, u_h'$ will be tree descriptors for the function $f$ restricted to $B_{i, j}$'s and within these interval, the density of tuples is at least $\Omega(\eps)$. As a result, the second round of the algorithm, recursively handles each subtree rooted at $u_2', \dots, u_h'$ with one fewer round. Since the subtrees have strictly fewer than $k/2$ leaves, $\lfloor \log_2 k \rfloor - 1$ rounds are enough for an inductive argument. Moreover, since the total number of nodes is at most $2k$ and each recursive call succeeds with probability at least $1-1/(100 k)$, by a union bound we may assume that all recursive calls succeed. 
 
Unrolling the recursion, the query complexity $\Theta((\log n)^{\lfloor \log_2 k \rfloor})$ can be explained with a simple combinatorial game. We start with a rooted binary tree $G$ on $k$ leaves. In one round, whenever $G$ is not simply a leaf, we pick the leaf $w$ which is the primary index of $G$, and replace $G$ with a collection of subtrees obtained by cutting out the root-to-$w$ path in $G$. These rounds ``pay'' a factor of $\Theta(\log n)$, since the algorithm must find intervals on which the collection of subtrees form tree descriptors of $f$ (restricted to these intervals). In the subsequent rounds, we recurse on each subtree simultaneously, picking the leaf of the primary index in each, and so on. After $\lfloor \log_2 k \rfloor$ many rounds, the trees are merely leaves, and the algorithm does not need to pay the factor $\Theta(\log n)$ to find good intervals, as it may simply sample from these intervals.
 
The execution of the above high-level plan is done in Section~\ref{sec:case2}, where Lemma~\ref{lem:case2} is the main inductive lemma containing the analysis of the main algorithm (shown in Figure~\ref{fig:sample-splittable} and Figure~\ref{fig:sample-splittable-2}).

\subsection{Our techniques: Lower bound}

In order to highlight the main ideas behind the proof of Theorem~\ref{thm:intro-lb} (the lower bound on the query complexity), we first cover the simpler case of $k=2$. This case corresponds to a lower bound of $\Omega(\log n)$ on the number of queries needed for non-adaptive and one-sided algorithms for monotonicity testing. Such a lower bound is known, even for adaptive algorithms with two-sided error \cite{EKKRV00, F04}. We rederive and present the well-known non-adaptive one-sided lower bound in our language; after that, we generalize it to the significantly more involved case $k > 2$. For the purpose of this introduction, we give an overview assuming that both $n$ and $k$ are powers of $2$; as described in Section \ref{sec:lowerbounds}, a simple ``padding'' argument generalizes the result to all $n$ and $k$.

For any $n \in \N$ which is a power of $2$ and $t \in [n]$, consider the \emph{binary representation} $B_n(t) = (b^t_1, b^t_2, \ldots, b^t_{\log_2 n}) \in \{0,1\}^{\log_2 n}$ of $t$, where $t = b^t_1 \cdot 2^{0} + b^t_2 \cdot 2^{1} + \dots + b^t_{\log_2 n} \cdot 2^{\log_2 n - 1}$.  
For $i \in [\log_2 n]$, the \emph{bit-flip operator}, $F_i \colon [n] \to [n]$, takes an input $t \in [n]$ with binary representation $B_n(t)$ and outputs the number $F_i(t) = t' \in [n]$ with binary representation obtained by flipping the $i$-th bit of $B_n(t)$. Finally, for any two distinct elements $x,y \in [n]$, let $M(x,y) \in [\log_2 n]$ be the index of the most significant bit in which they differ, i.e., the largest $i$ where $b_i^x \neq b_i^y$. 

As usual for lower bounds on randomized algorithms, we rely on Yao's minimax principle \cite{Y77}. In particular, our lower bounds proceed by defining, for each $n$ and $k$ (which are powers of $2$), a distribution $\calD_{n,k}$ supported on functions $f \colon [n] \to \R$ which are all $\eps$-far from $(12\dots k)$-free. We show that any deterministic and non-adaptive algorithm which makes fewer than $q$ queries, where $q = c_k (\log_2 n)^{\log_2 k}$ and $c_k > 0$ depends only on $k$, fails to find a $(12\ldots k)$-pattern in a random $\boldf \sim \calD_{n,k}$, with probability at least $1/10$. Note that any deterministic, non-adaptive algorithm which makes fewer than $q$ queries is equivalently specified by a set $Q \subseteq [n]$ with $|Q| < q$. Thus, the task of the lower bound is to design a distribution $\calD_{n,k}$ supported on functions $f \colon [n] \to \R$, each of which is $\eps$-far from $(12\dots k)$-free, such that for any $Q \subseteq [n]$ with $|Q| < c_k (\log_2 n)^{\log_2 k}$ the following holds
\begin{align}
\Prx_{\boldf \sim \calD_{n,k}}\left[ \exists i_1, \dots, i_k \in Q: i_1 < \dots < i_k \text{ and } \boldf(i_1) < \dots < \boldf(i_k) \right] \leq \frac{9}{10}. \label{eq:yao}
\end{align}
 
\medskip\noindent\textsc{Lower bound for $k=2$ (monotonicity).} 
The case of $k=2$ relies on the following idea: for any $i \in [\log_2 n]$, one can construct a function (in fact, a permutation) $f_i \colon [n] \to [n]$ which is $1/2$-far from $(12)$-free, and furthermore, all pairs of distinct elements $(x, y) \in [n]^2$ where $x < y$ and $f_i(x) < f_i(y)$ satisfy $M(x, y) = i$.
One can construct such a function $f_i \colon [n] \to [n]$, for any $i \in[\log_2 n]$ in the following way.
First, let $f^{\downarrow} \colon [n] \to [n]$ be the decreasing permutation,   $f^{\downarrow}(x) = n+1-x$ for any $x \in [n]$.
Now take $f_i$ to be $f^{\downarrow} \circ F_i$, where $\circ$ denotes function composition, that is, $f_i(x) = f^{\downarrow}(F_i(x))$ for any $x \in [n]$.
Finally, set $\calD_{n,2}$ to be the uniform distribution over the functions $f_1, f_2, \ldots, f_{\log n}$ (see Figure~\ref{fig:mon-construct}).

\begin{figure}
\centering
\begin{picture}(340, 200)
\put(0,0){\includegraphics[width=0.7\linewidth]{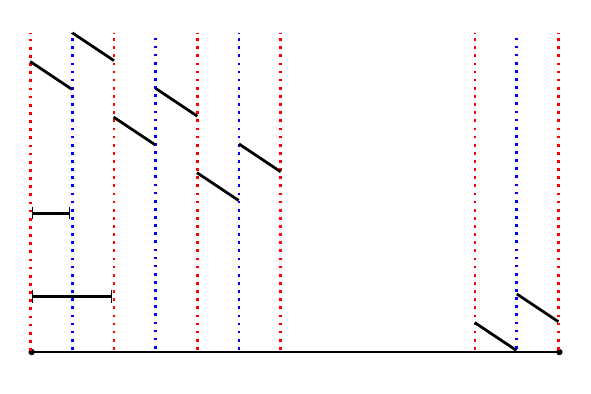}}
\put(45, 60){$2^{i}$}
\put(19, 108){$2^{i-1}$}
\end{picture}
\caption{Example of a function $f_i$ lying in the support of $\calD_{n,2}$. One may view the domain as being divided into intervals of length $2^{i}$ (displayed as intervals lying between dotted red lines) and the permutation $f^{\downarrow}$ flipped across adjacent intervals of length $2^{i-1}$ inside a segment of length $2^i$ (displayed as intervals lying between dotted blue lines). Note that all $(12)$-patterns in $f_i$ above have the $i$th bit flipped.}
\label{fig:mon-construct}
\end{figure}

Towards proving \eqref{eq:yao} for the distribution $\calD_{n, 2}$, we introduce the notion of \emph{binary profiles} captured by a set of queries. For any fixed $Q \subseteq [n]$, the binary profiles captured in $Q$ are given by the set
\[
\bprof(Q) = \{ i \in[\log_2 n] \  :\  \text{there exist $x,y \in Q$ such that } M(x,y)=i\}.
\]
Since all $(12)$-patterns $(x, y)$ of $f_i$ have $M(x, y) = i$, the probability over $\boldf \sim \calD_{n,2}$ that an algorithm whose set of queries is $Q$, finds a $(12)$-pattern in $\boldf$ is at most $|\bprof(Q)| / \log_2 n$. 
We show that for any set $Q \subseteq [n]$, $|\bprof(Q)| \leq |Q| - 1$. This completes \eqref{eq:yao}, and proves the lower bound of $\frac{9}{10}\log_2 n$ for $k=2$. 

The proof that $|\bprof(Q)| \leq |Q|-1$ for any set $Q \subset [n]$, i.e., the number of \emph{captured profiles} is bounded by the number of queries, follows by induction on $|Q|$. 
The base case $|Q| \leq 2$ is trivial. When $|Q| > 2$, 
let $i_\text{max} = \max \bprof(Q)$. Consider the partition of $Q$ into $Q_0$ and $Q_1$, where
\[ 
	Q_0 = \{x \in Q: b_{i_\text{max}}^x = 0\} \qquad \text{ and }\qquad Q_1 = \{ y \in Q : b_{i_{\text{max}}}^y = 1\}. 
\]
Since $\bprof(Q) = \bprof(Q_0) \cup \bprof(Q_1) \cup \{i_\text{max}\}$, we conclude that 
\[
	|\bprof(Q)| \leq |\bprof(Q_0)| + |\bprof(Q_1)| + 1 \leq |Q_0| - 1 + |Q_1| - 1 + 1 = |Q|-1,
\]
where the second inequality follows from the inductive hypothesis.

\medskip\noindent\textsc{Generalization to $k>2$: Proof of Theorem \ref{thm:intro-lb}.} 
We now provide a detailed sketch of the proof of Theorem~\ref{thm:intro-lb}. The main objects and notions used are defined, while leaving technical details to Section~\ref{sec:lowerbounds}. Let $k = 2^h$ for $h \in\N$; the case $h=1$ corresponds to the previous discussion. 

We first define the distributions $\calD_{n,k}$ supported on permutations $f \colon [n] \to [n]$ which are $\Omega(1/k)$-far from $(12\dots k)$-free. Recall that the function $f_i$ in the case $k=2$ was constructed by ``flipping'' bit $i$ in the representation of $f^{\downarrow}$, that is, $f_i = f^{\downarrow} \circ F_i$. Generalizing this construction, for any $i_1 < i_2 < \dots < i_h \in [\log_2 n]$ we let $f_{i_1, \ldots, i_{h}} \colon [n] \to [n]$ denote the result of flipping bits $i_1, i_2, \ldots, i_h$ in the representation of $f^{\downarrow}$: 
\[ f_{i_1, \ldots, i_{h}} \eqdef f^{\downarrow} \circ F_{i_h} \circ \ldots \circ F_{i_1}.
\]
It can be shown that $f_{i_1, \ldots, i_h}$ is $(1/k)$-far from $(12\dots k)$-free (see Figure~\ref{fig:construct-recurse}).
We take $\calD_{n,k}$ as the uniform distribution over all functions of the form $f_{i_1,\ldots, i_{h}}$, where $i_1 < \dots < i_h \in [\log_2 n]$.

\begin{figure}
	\centering
	\begin{picture}(370, 200)
	\put(0,0){\includegraphics[width=0.8\linewidth]{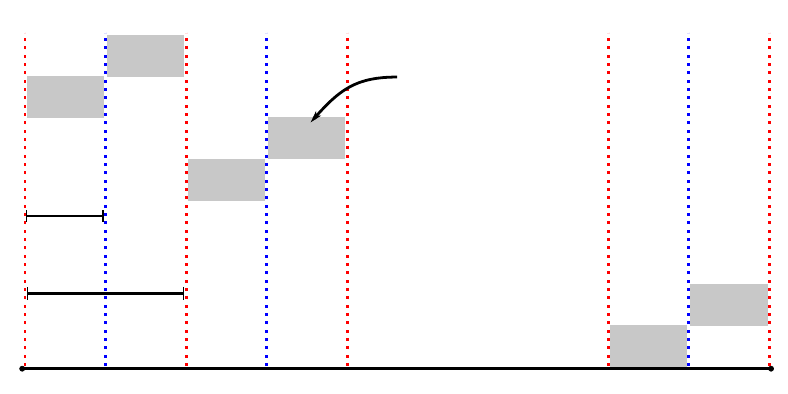}}
	\put(65, 52){$2^{i_h}$}
	\put(19, 90){$2^{i_h-1}$}
	\put(195, 150){$f_{i_1, \dots, i_{h-1}}$}
	\end{picture}
	\caption{Example of a function $f_{i_1,\dots, i_h}$ lying in the support of $\calD_{n,k}$, where $k = 2^h$. Similarly to the case of $k=2$ (shown in Figure~\ref{fig:mon-construct}), the domain is divided into intervals of length $2^{i_h}$ (shown between red dotted lines), and functions are flipped across adjacent intervals of length $2^{i_h-1}$ within an interval of length $2^{i_h}$ (shown between blue dotted lines). Inside each grey region is a recursive application of the construction, $f_{i_1, \dots, i_{h-1}}$, after shifting the range.}
	\label{fig:construct-recurse}
\end{figure}
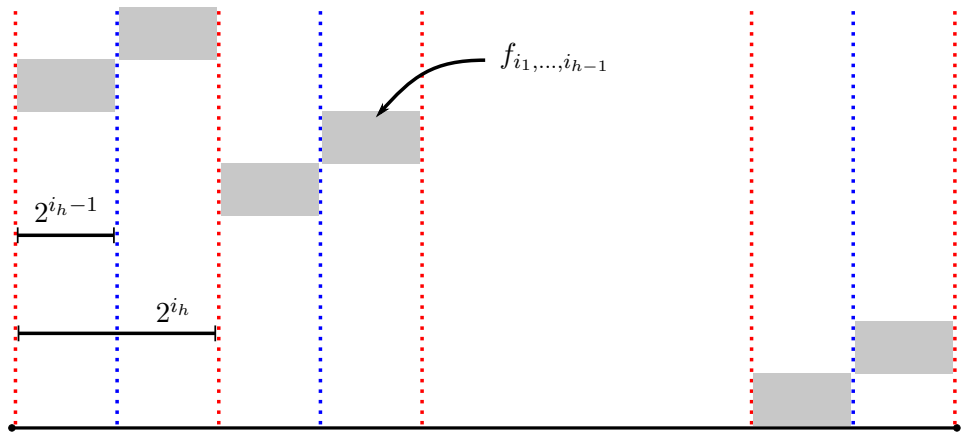

Towards the proof of \eqref{eq:yao} for the distribution $\calD_{n,k}$, we 
generalize the notion of a binary profile. 
Consider any $k$-tuple of indices $(x_1, \dots, x_k) \in [n]^{k}$ satisfying $x_1 < \dots < x_k$.  
We say that $(x_1, x_2, \ldots, x_k)$ has \emph{$h$-profile of type $(i_1, \ldots, i_h)$} if, 
\[ 
	\text{for every $j \in [k-1]$}, \qquad M(x_{j}, x_{j+1}) = i_{M(j-1, j)}. 
\]
For instance, when $h=3$ (i.e., $k=8$) the tuple $(x_1, \dots, x_k)$ has $h$-profile of type $(i_1, i_2, i_3)$ if the sequence $(M(x_j, x_{j+1}))_{j = 1}^{7}$ is $(i_1, i_2, i_1, i_3, i_1, i_2, i_1)$. See Figure~\ref{fig:profiles} for a visual demonstration of a $3$-profile.\footnote{Unlike the case $k=2$, not all tuples $(x_1, \ldots, x_k)$ with $x_1 < \ldots < x_k$ have an $h$-profile. For what follows we will only be interested in tuples that \emph{do} have a profile.}

\begin{figure}
	\centering
	\begin{picture}(600, 200)
	\put(0,0){\includegraphics[width=\linewidth]{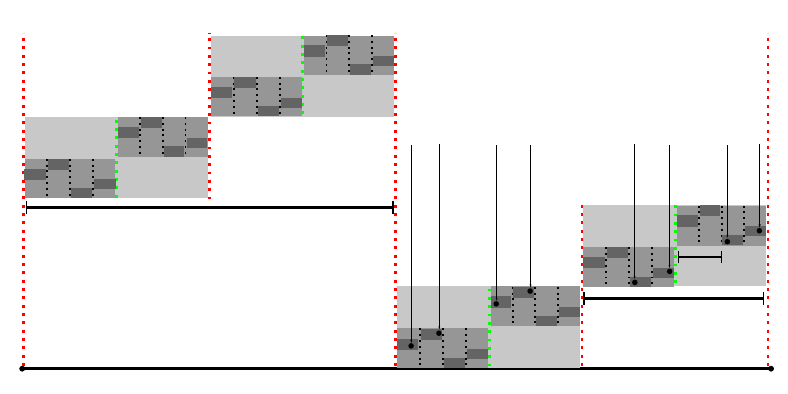}}
	\put(120, 95){$2^{i_3}$}
	\put(405, 42){$2^{i_2}$}
	\put(415, 70){$2^{i_1}$}
	\put(240, 155){$x_1$}
	\put(258, 155){$x_2$}
	\put(295, 155){$x_3$}
	\put(312, 155){$x_4$}
	\put(375, 155){$x_5$}
	\put(395, 155){$x_6$}
	\put(430, 155){$x_7$}
	\put(448, 155){$x_8$}
	\end{picture}
	\caption{Example of a function $f_{i_1, i_2, i_3}$ in the support of $\calD_{n, 8}$, with a $k$-tuple $(x_1, \dots, x_8)$ whose $h$-profile has type $(i_1, i_2, i_3)$.}
	\label{fig:profiles}
\end{figure}
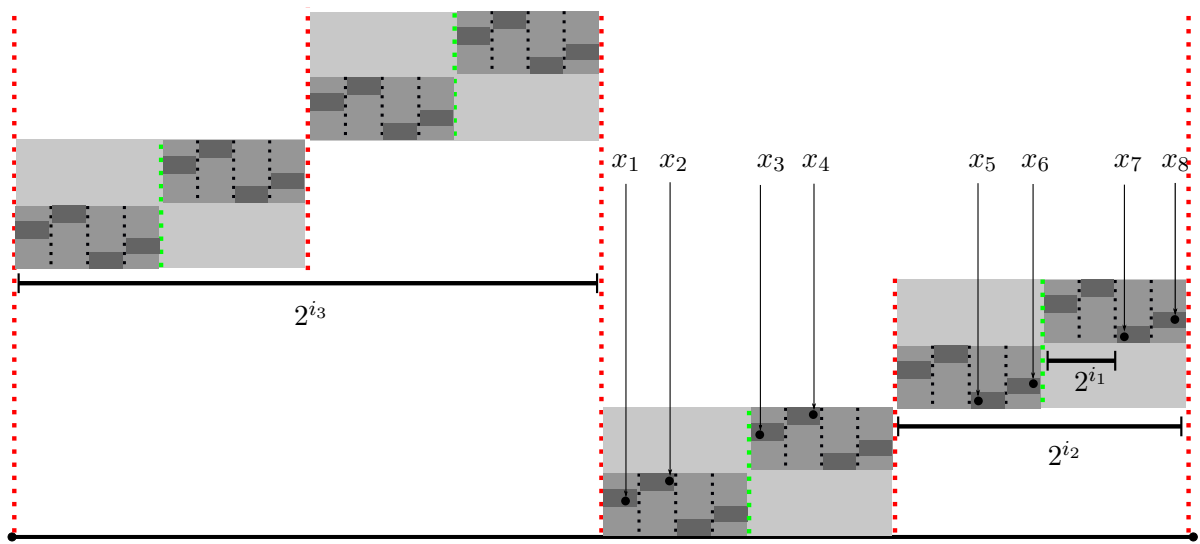

It can be shown that for any $i_1 < \dots < i_h \in [\log_2 n]$, the function $f = f_{i_1, \ldots, i_{h}}$ has the following property. If $x_1 < \dots < x_k \in [n]$ satisfy $f(x_1) < \dots < f(x_k)$, i.e., the $k$-tuple $(x_1, \dots, x_k)$ is a $(12\dots k)$-pattern of $f_{i_1, \dots, i_k}$, then $(x_1, \dots, x_k)$ has an $h$-profile of type $(i_1, \dots, i_h)$. 
We thus proceed similarly to the case $k=2$. For any $Q \subseteq [n]$, we define the set of all $h$-profiles captured by $Q$ as follows
\[
	\bprof_h(Q) = \left\{(i_1, \ldots, i_{h})  :\  \begin{array}{l} \text{there exist $x_1, \dots,x_k \in Q$ where }  x_1 < \dots < x_k \\
	\text{and $(x_1, \dots, x_k)$ has $h$-profile of type $(i_1, \ldots, i_{h})$} \end{array} \right\}.
\]

The proof that $|\bprof_h(Q)| \leq |Q| - 1$ for any $Q \subseteq [n]$ follows by induction on $h$. The base case $h=1$ was covered in the discussion on $k=2$. For $h > 1$, we define subsets 
\[ 
	\emptyset = B_{\log_2 n+1} \subseteq B_{\log_2 n} \subseteq \ldots \subseteq B_{1} = Q, 
\]
where, given $B_{i+1}$, the set $B_i \supseteq B_{i+1}$ is an arbitrary maximal subset of $Q$ containing $B_{i+1}$, so that no two elements $x \neq y \in B_i$ satisfy $M(x,y) < i$. Additionally, for each $j \in [\log_2 n]$ we let
\[ 
	N_j = \left\{ (i_2, \dots, i_h) : 1 \leq j < i_2 \dots < i_h \leq \log_2 n \text{ and } (j, i_2, \dots, i_h) \in \bprof_h(Q) \right\}. 
\]

The key observation is that $N_j \subseteq \bprof_{h-1}(B_j \setminus B_{j+1})$. To see this, note first that any $(j, i_2, \ldots, i_{h}) \in \bprof_h(Q)$ also satisfies $(j, i_2, \ldots, i_{h}) \in \bprof_h(B_j)$. Indeed, suppose that a tuple $(x_1, \ldots, x_k)$ with $x_1 < \dots < x_k \in Q$ has $h$-profile $(j, i_2, \ldots, i_{h})$. By the maximality of $B_j$, we know that for every $1 \leq \ell \leq k$ there exists $y_\ell \in B_j$ such that either $x_{\ell} = y_{\ell}$ or $M(x_{\ell}, y_{\ell}) < j$. This implies that $\{y_1, \ldots, y_{k}\} \subseteq B_j$ has $h$-profile $(j, i_2, \ldots, i_{h})$.

Now, suppose that $y_1, \dots, y_k \in B_j$ satisfies $y_1< \ldots< y_{k}$ and has $h$-profile of type $(j, i_2, \ldots, i_{h})$ in $B_j$. For any $1 \leq t \leq k/2$ we have $M(y_{2t-1}, y_{2t}) = j$. Therefore, at most one of $y_{2t-1}, y_{2t}$ is in $B_{j+1}$, and hence, for any such $t$ there exists $z_t \in \{y_{2t-1}, y_{2t}\} \setminus B_{j+1} \subseteq B_{j} \setminus B_{j+1}$. It follows that $(z_1, \ldots, z_{k/2}) \in B_j \setminus B_{j+1}$ has $(h-1)$-profile $(i_2, \ldots, i_{h})$. This concludes the proof that $N_j \subseteq \bprof_{h-1}(B_j \setminus B_{j+1})$.

We now use the last observation to prove that $|\bprof_h(Q)| \leq |Q|-1$. Note that 
$$
	\bprof_h(Q) = \bigcup_{j=1}^{\log_2 n} \{(j, i_2, \ldots, i_h) : (i_2, \ldots, i_h) \in N_j \}
	\qquad \text{and}  \qquad 
	Q = \bigcup_{j=1}^{\log_2 n} (B_j \setminus B_{j+1}),
$$
where both unions are disjoint unions. By the induction assumption, $|N_j| \leq |\bprof_{h-1}(B_j \setminus B_{j+1})| < |B_j \setminus B_{j+1}|$ for any $j$ if $N_j$ is non-empty; If $N_j$ is empty, then $|N_j| \leq |\bprof_{h-1}(B_j \setminus B_{j+1})| \le |B_j \setminus B_{j+1}|$ trivially holds. Hence
\[
	|Q| = \sum_{j=1}^{\log_2 n} |B_j \setminus B_{j+1}| >\sum_{j=1}^{\log_2 n} |N_j| = |\bprof_h(Q)|, 
\]
where the strict inequality follows because if $\bprof_h(Q)$ is non-empty then $N_j$ is non-empty for some $j$. This completes the proof.

\subsection{Organization}
	We start by introducing the notation that we shall use throughout the paper in Section~\ref{sec:preliminaries}. In Section~\ref{sec:structural} we prove our main structural result, and formally define the notions that underlie it: namely, Theorem~\ref{thm:tree}, along with the definitions of growing suffixes and representation by tree descriptors (Definitions~\ref{def:growing-suffixes} and~\ref{def:tree-rep}). Section~\ref{sec:algorithm} then leverages this dichotomy to describe and analyze our testing algorithm, thus establishing the upper bound of Theorem~\ref{thm:intro-ub} (see Theorem~\ref{thm:ub} for a formal statement). Finally, we complement this algorithm with a matching lower bound in Section~\ref{sec:lowerbounds}, where we prove Theorem~\ref{thm:intro-lb}.

	While Section~\ref{sec:algorithm} crucially relies on Section~\ref{sec:structural}, these two sections are independent of~Section~\ref{sec:lowerbounds}, which is mostly self-contained.
 
\subsection{Notation and Preliminaries}\label{sec:preliminaries}

We write $a \lsim b$ if there exists a universal positive constant $C > 0$ such that $a \leq C b$, and $a \asymp b$ if we have both $a \lsim b$ and $b \lsim a$. At times, we write $\poly(k)$ to stand for $O(k^{C})$, where $C > 0$ is a large enough universal constant. Unless otherwise stated, all logarithms will be in base 2. We frequently denote $\calI$ as a collection of disjoint intervals, $I_1, \dots, I_s$, and then write $\calS(\calI)$ for the set of all sub-intervals which lie within some interval in $\calI$. For two collections of disjoint intervals $\calI_0$ and $\calI_1$, we say that $\calI_1$ is a \emph{refinement} of $\calI_0$ if every interval in $\calI_1$ is contained within an interval in $\calI_0$. (We remark that it is not the case that intervals in $\calI_1$ must form a partition of intervals in $\calI_0$.) For a particular set $A \subseteq [n]$ and an interval $I \subseteq [n]$, we define the \emph{density} of $A$ in $I$ as the ratio $|A \cap I| / |I|$. Given a set $S$, we write $\bx \sim S$ to indicate that $\bx$ is a random variable given by a sample drawn uniformly at random from $S$, and $\calP(S)$ for the power set of $S$. Given a sequence $f$ of length $n$, we shall interchangeably use the notions \emph{$(12\dots k)$-copy}, \emph{$(12 \dots k)$-pattern}, and \emph{length-$k$ increasing subsequence}, to refer to a tuple $(x_1, \ldots, x_k) \in [n]^k$ such that $x_1 < \ldots < x_k$ and $f(x_1) < \ldots < f(x_k)$.
 
\section{Structural Result}\label{sec:structural}

\subsection{Rematching procedure}

	Let $f \colon [n] \to \R$ be a function which is $\eps$-far from $(12\dots k)$-free. 
	Let $T$ be a set of $k$-tuples representing monotone subsequences of length $k$ within $f$, i.e.,
	\[
		T \subseteq \left\{ (i_1, \dots, i_k) \in [n]^k : i_1 < \dots < i_k \text{ and } f(i_1) < \dots < f(i_k) \right\},
	\]
	and for such $T$ let $E(T)$ be the set of indices of subsequences in $T$, so
	\[
		E(T) = \bigcup_{(i_1, \ldots, i_k) \in T} \{i_1, \ldots, i_k\}.
	\]

	\begin{observation} \label{obs:disjoint-families}
		If $f \colon [n] \to \R$ is $\eps$-far from $(12 \dots k)$-free, then there exists a set $T \subseteq [n]^k$ of disjoint length-$k$ increasing subsequences of $f$ such that $|T| \ge \eps n / k$.
	\end{observation}

	To see why the observation holds, take $T$ to be a maximal disjoint set of such $k$-tuples. Then we can obtain a $(12 \dots k)$-free sequence from $f$ by changing only the entries of $E(T)$ (e.g.\ for every $i \in E(T)$ define $f(i) = f(j)$ where $j$ is the largest $[n] \setminus E(T)$ which is smaller than $i$. If there is no $j \in [n] \setminus E(T)$ where $j < i$, let $f(i) = \max_{\ell \in [n]} f(\ell)$). Since $f$ is $\eps$-far from being $(1 2 \dots k)$-free, we have $|E(T)| \ge \eps n$, thus $|T| \ge \eps n / k$.

	In this section, we show that from a function $f \colon [n] \to \R$ which is $\eps$-far from $(12\dots k)$-free and a set $T_0$ of disjoint, length-$k$ monotone subsequences of $f$, a greedy rematching algorithm finds a set $T$ of disjoint, length-$k$ monotone subsequences of $f$ where $E(T) \subseteq E(T_0)$ with some additional structure, which will later be exploited in the structural lemma and the algorithm. The greedy rematching algorithm, $\GreedyDisjointTuples$, is specified in Figure~\ref{fig:greedy}; \red{for convenience, in view of its later use in the algorithm, we phrase it in terms of an arbitrary parameter $k_0$, not necessarily the (fixed) parameter $k$ itself.}

	\begin{lemma}\label{lem:rematching}
		Let $k_0 \in \N$, $f \colon [n] \to \R$, and let $T_0 \subseteq [n]^{k_0}$ be a set of disjoint monotone subsequences of $f$ of length $k_0$. Then there exists a set $T \subseteq [n]^{k_0}$ of disjoint $k_0$-tuples with $E(T) \subseteq E(T_0)$ such that the following holds.
		\begin{enumerate}
		\item\label{en:cond-1} 
			The set $T$ holds disjoint monotone subsequences of length $k_0$.
		\item\label{en:cond-2} 
			The size of $T$ satisfies $|T| \geq |T_0| / k_0$. 
		\item\label{en:cond-3} 
			For any two $(i_1, \dots, i_{k_0}), (j_1, \dots, j_{k_0}) \in T$ and any $\ell \in [k_0 - 1]$, if $i_1 < j_1$, $i_{\ell} < j_{\ell}$ and $i_{\ell+1} > j_{\ell+1}$ then $f(i_{\ell+1}) > f(j_{\ell+1})$.
		\end{enumerate}
	\end{lemma}

	\begin{figure}[ht!]
		\begin{framed}
			\centering 
			
			\begin{minipage}{.98\textwidth}

				\noindent Subroutine $\GreedyDisjointTuples\hspace{0.05cm}(f, k_0, T_0)$
				\vspace{0.3cm}

				\noindent {\bf Input:} A function $f \colon [n] \to \R$, integer $k_0\in\N$, and a set $T_0$ of disjoint monotone subsequences of $f$ of length $k_0$.\\
				{\bf Output:} a set $T \subseteq [n]^{k_0}$ of disjoint monotone subsequences of $f$ of length $k_0$. 
				
				\begin{enumerate}
					\item Let $T = \emptyset$ and $i$ be the minimum element in $E(T_0)$. Repeat the following.					
					\begin{itemize}
						\item[i.] 
							Let $i_1 \gets i$. If there exists $j_2, \ldots, j_{k_0} \in E(T_0) \setminus E(T)$ such that $(i_1, j_2, \dots, j_{k_0})$ is an increasing subsequence of $f$, pick $i_2, \ldots, i_{k_0} \in E(T_0) \setminus E(T)$ recursively as follows: for $\ell = 2, \ldots, k_0$, let $i_{\ell}$ be the smallest element in $E(T_0) \setminus E(T)$ for which there exist $j_{\ell+1}, \dots, j_{k_0} \in E(T_0) \setminus E(T)$ such that $(i_1, \dots, i_{\ell}, j_{\ell+1}, \dots, j_{k_0})$ is an increasing subsequence of $f$.
						\item[ii.] 
							If $(i_1, \dots, i_{k_0})$ is a monotone subsequence found by (i), set $T \gets T \cup \{ (i_1, \dots, i_{k_0})\}$.
						\item[iii.] 
							Let $i$ be the next element of $E(T_0)\setminus E(T)$, if such an element exists; otherwise, proceed to \ref{itm:alg-rematch-end}.
					\end{itemize}
					\item \label{itm:alg-rematch-end}
						Output $T$.
				\end{enumerate}
			\end{minipage}
		\end{framed}\vspace{-0.2cm}
		\caption{Description of the $\GreedyDisjointTuples$ subroutine.} \label{fig:greedy}
	\end{figure}

	\begin{proofof}{Lemma~\ref{lem:rematching}}
		We show that the subroutine $\GreedyDisjointTuples(f, k_0, T_0)$, described in Figure~\ref{fig:greedy}, finds a set $T$ with $E(T) \subseteq E(T_0)$ satisfying properties \ref{en:cond-1}, \ref{en:cond-2}, and \ref{en:cond-3}. Property \ref{en:cond-1} is clear from the description of $\GreedyDisjointTuples(f,k_0, T_0)$. For \ref{en:cond-2}, suppose $|T| < |T_0|/k_0$, then, there exists a tuple $(i_1, \dots, i_{k_0}) \in T_0$ with $\{ i_1, \dots, i_{k_0}\} \cap E(T) = \emptyset$. Since $\GreedyDisjointTuples(f, k_0, T_0)$ increases the size of $T$ throughout the execution, $\{ i_1, \dots, i_{k_0} \} \cap T = \emptyset$ at every point in the execution of the algorithm. This is a contradiction; when $i = i_1$, a monotone subsequence disjoint from $T$ would have been found, and $i_1$ included in $T$. Finally, for \ref{en:cond-3}, consider the iteration when $i = i_1$, and note that at this moment, $T \cap \{ i_1, \dots, i_{k_0}, j_1, \dots, j_{k_0}\} = \emptyset$. Suppose that $i_{\ell} < j_{\ell}$, $j_{\ell+1} < i_{\ell+1}$; if $f(j_{\ell+1}) \geq f(i_{\ell+1})$, then $(i_1, \dots, i_{\ell}, j_{\ell+1}, \dots, j_{k_0})$ is an increasing subsequence in $E(T_0) \setminus E(T)$, which means that $j_{\ell+1}$ would have been preferred over $i_{\ell+1}$, a contradiction.
	\end{proofof}

	\begin{definition}[$c$-gap]
		Let $(i_1, \dots, i_{k_0})$ be a monotone subsequence of $f$ and let $c \in [k_0-1]$. We say that $(i_1, \dots, i_{k_0})$ is a \emph{$c$-gap subsequence} if $c$ is the smallest integer such that $i_{c+1} - i_c \geq i_{b+1} - i_b$ for all $b \in [{k_0}-1]$. 
	\end{definition}

	Note that for a set $T$ of disjoint length-$k_0$ monotone subsequences of $f$, we may partition the $k_0$-tuples of $T$ into $(T_1, \dots, T_{k_0-1})$ where for each $c \in [k_0-1]$, $T_c$ holds the $c$-gap monotone subsequences of $T$. As these sets form a partition of $T$, the following lemma is immediate from Lemma~\ref{lem:rematching}.  

	\begin{lemma}\label{lem:big-split}
		Let $f \colon [n] \to \R$, and let $T_0$ be a set of disjoint length-$k_0$ monotone subsequences of $f$. Then there exist $c \in [k_0-1]$ and a family $T \subseteq [n]^{k_0}$ of disjoint monotone subsequences of $f$, with $E(T) \subseteq E(T_0)$ such that the following holds.
		\begin{enumerate}
			\item
				The subsequences in $T$ are all $c$-gap subsequences.
			\item
				$|T| \ge |T_0| / k_0^2$.
			\item\label{itm:property-greedy}
				For any two $(i_1, \dots, i_{k_0}), (j_1, \dots, j_{k_0}) \in T$ and any $\ell \in [k_0 - 1]$, if $i_1 < j_1$, $i_{\ell} < j_{\ell}$ and $i_{\ell+1} > j_{\ell+1}$ then $f(i_{\ell+1}) > f(j_{\ell+1})$.
		\end{enumerate}
	\end{lemma}
 
\subsection{Growing suffixes and splittable intervals}

	We now proceed to set up notation and prepare for the main structural theorem for sequences $f \colon [n] \to \R$ which are $\eps$-far from $(12\dots k)$-free. In order to simplify the presentation of the subsequent discussion, consider fixed $k \in \N$ and $\eps \in (0, 1)$, as well as a fixed sequence $f \colon [n] \to \R$ which is $\eps$-far from $(12\dots k)$-free. We will, at times, suppress polynomial factors in $k$ by writing $\poly(k)$ to refer to a large enough polynomial in $k$, whose degree is a large enough universal constant. 
	By Observation~\ref{obs:disjoint-families} and Lemma~\ref{lem:big-split}, there exists an integer $c \in [k-1]$ and a set $T$ of disjoint monotone subsequences of $f$ which have a $c$-gap, satisfying $|T| \geq \eps n / \poly(k)$ and property \ref{itm:property-greedy} from Lemma~\ref{lem:big-split}. For the rest of the subsection, we consider a fixed setting of such $c \in [k-1]$ and set $T$. 

	We will show (in Theorem~\ref{thm:main-structure}) that one of the following two possibilities holds.
	Either there is a large set of what we call \emph{growing suffixes} (see Definition~\ref{def:growing-suffixes} for a formal definition), or there are disjoint intervals which we call \emph{splittable} (see Definition~\ref{def:splittable} for a formal definition). Intuitively, a growing suffix will be given by the suffix $(a, n]$ and will have the property that by dividing $(a, n]$ into $\Theta(\log_2(n - a))$ segments of geometrically increasing lengths, there are many monotone subsequences $(i_1, \dots, i_k)$ of $f$ lying inside $(a, n]$ where each $i_t$ belongs to a different segment. In the other case, an interval $[a,b]$ is called splittable if it can be divided into three sub-intervals of roughly equal size, 
	which we refer to as the left, middle, and right intervals, with the following property: the left interval contains a large set $T_{L}$ of $(12\dots c)$-patterns, the right interval contains a large set $T_R$ of $(12\dots (k-c))$-patterns, and combining any $(12\dots c)$-pattern in $T_L$ with any $(12\dots (k-c))$-pattern in $T_R$ yields a $(12\dots k)$-pattern.

	For each index $a \in [n]$, let $\eta_a = \lceil \log_2(n-a) \rceil$. Let $S_1(a), \dots, S_{\eta_a}(a) \subseteq [n]$ be disjoint intervals given by $S_t(a) = [a + 2^{t-1}, a + 2^t) \cap [n]$. The collection of intervals $S(a) = (S_t(a) : t \in [\eta_a])$ partitions the suffix $(a, n]$ into intervals of geometrically increasing lengths (except possibly the last interval, which may be shorter), and we refer to the collection $S(a)$ as the \emph{growing suffix} at $a$. 

	\begin{definition}\label{def:growing-suffixes}
		Let $\alpha, \beta \in [0,1]$.
		We say that an index $a \in [n]$ starts an \emph{$(\alpha, \beta)$-growing suffix} if, when considering the collection of intervals $S(a) = \{ S_t(a) : t \in [\eta_a]\}$, for each $t \in [\eta_a]$ there is a subset $D_t(a) \subseteq S_t(a)$ of indices 
		such that the following properties hold. 
		\begin{enumerate} 
			\item 
				We have $|D_t(a)|/|S_t(a)| \le \alpha$ for all $t \in [\eta_a]$, and
				$\sum_{t=1}^{\eta_a} |D_t(a)|/|S_t(a)| \geq \beta$.
				\label{en:grow-cond-3}
			\item 
				For every $t, t' \in [\eta_a]$ where $t < t'$, if $b \in D_t(a)$ and $b' \in D_{t'}(a)$, then $f(b) < f(b')$. 
				\label{en:grow-cond-2}
		\end{enumerate}
	\end{definition}

	Intuitively, our parameter regime will correspond to the case when $\alpha$ is much smaller than $\beta$, specifically, $\alpha \leq \beta / \poly(k)$, for a sufficiently large-degree polynomial in $k$. If $a \in [n]$ starts an $(\alpha, \beta)$-growing suffix with these parameters, then the $\eta_a$ segments, $S_1(a), \dots, S_{\eta_a}(a)$, contain many monotone subsequences of length $k$ which are algorithmically easy to find (given access to the start $a$).
	Indeed, by (\ref{en:grow-cond-2}), it suffices to find a $k$-tuple $(i_1, \ldots, i_k)$ such that $i_1 \in D_{t_1}, \ldots, i_k \in D_{t_k}$, for some $t_1, \ldots, t_k \in [\eta_a]$ with $t_1 < \ldots < t_k$ (see Figure~\ref{fig:growing-suffix}).
	By (\ref{en:grow-cond-3}), the sum of densities is at least $\beta$, yet each density is less than $\alpha \le \beta / \poly(k)$. In other words, the densities of $D_{1}(a), \dots, D_{\eta_a}(a)$ within $S_{1}(a), \dots, S_{\eta_a}(a)$, respectively, must be spread out, which implies, intuitively, that there are many ways to pick suitable $i_1, \ldots, i_k$.

	\begin{figure}
		\begin{picture}(300, 150)
		\put(0,0){\includegraphics[width=\linewidth]{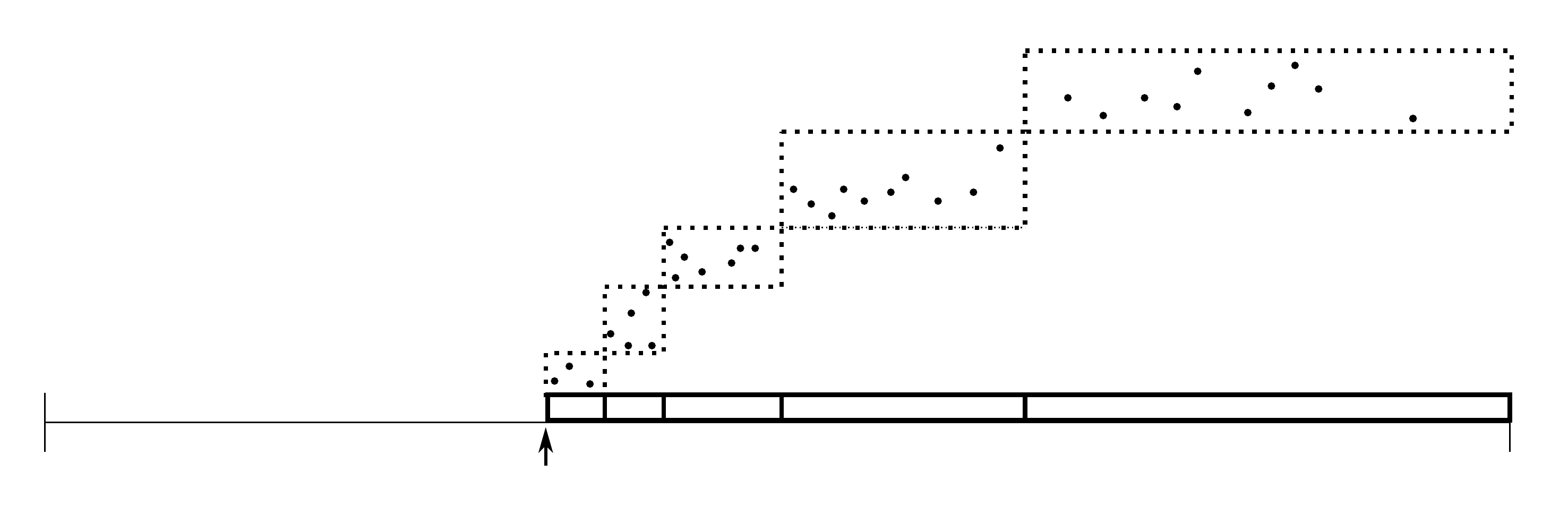}}
		\put(162, 5){$a$}
		\end{picture}
		\caption{Depiction of an $(\alpha, \beta)$-growing suffix at index $a \in [n]$ (see Definition~\ref{def:growing-suffixes}). The labeled segments $S_t(a)$ are shown, as well as the subsets $D_t(a)$. Notice that for all $j$, all the elements in $D_t(a)$ lie below those in $D_{t+1}(a)$. In Section~\ref{sec:case1}, we show that if an algorithm knows that $a$ starts an $(\alpha, \beta)$-growing suffix, for $\alpha \leq \beta / \poly(k)$, then sampling $\poly(k) / \beta$ many random indices from each $S_t(a)$ finds a monotone pattern with probability at least $0.9$.}
		\label{fig:growing-suffix}
	\end{figure}

	\begin{definition}\label{def:splittable}
		Let $\alpha,\beta \in (0,1]$ and $c \in [k_0-1]$. Let $I \subseteq [n]$ be an interval, let $T \subseteq I^{k_0}$ be a set of disjoint, length-$k_0$ monotone subsequences of $f$ lying in $I$, and define
		\begin{align*} 
			T^{(L)} &= \{ (i_1, \dots, i_c) \in I^c : (i_1, \dots, i_c) \text{ is a prefix of a $k_0$-tuple in $T$}\}, \text{ and }\\
			T^{(R)} &= \{ (j_1, \dots, j_{k_0-c}) \in I^{k_0-c} : (j_1, \dots, j_{k_0-c}) \text{ is a suffix of a $k_0$-tuple in $T$}\}.
		\end{align*}
		We say that the pair $(I, T)$ is \emph{$(c, \alpha,\beta)$-splittable} if $|T|/|I| \geq \beta$; $f(i_c) < f(j_1)$ for every $(i_1, \dots, i_c) \in T^{(L)}$ and $(j_1, \dots, j_{k_0-c}) \in T^{(R)}$; and there is a partition of $I$ into three adjacent intervals $L, M, R \subseteq I$ (that appear in this order, from left to right) of size at least $\alpha |I|$, satisfying $T^{(L)} \subseteq L^c$ and $T^{(R)} \subseteq R^{k_0-c}$. 
		
		A collection of disjoint interval-tuple pairs $(I_1, T_1), \dots, (I_s, T_s)$ is called a \emph{$(c, \alpha,\beta)$-splittable collection of $T$} if each $(I_j, T_j)$ is $(c, \alpha, \beta)$-splittable and the sets $(T_j : j \in [s])$ partition $T$.
	\end{definition}
	
	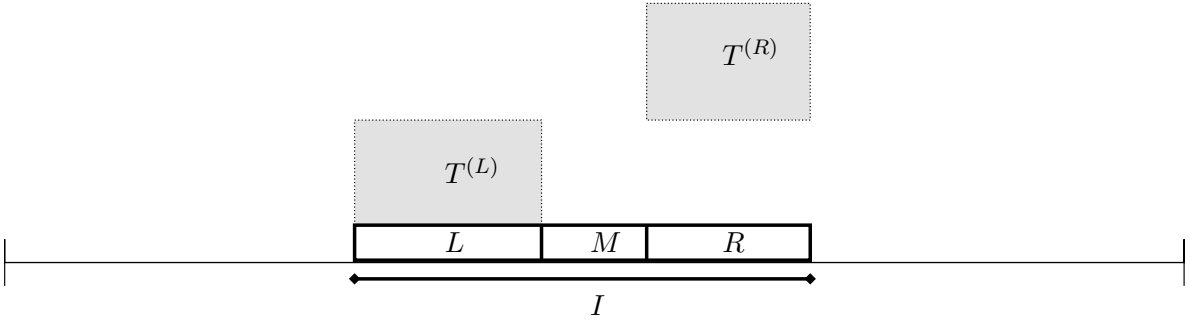
\begin{figure}
		\begin{picture}(300, 120)
		\put(0,0){\includegraphics[width=\linewidth]{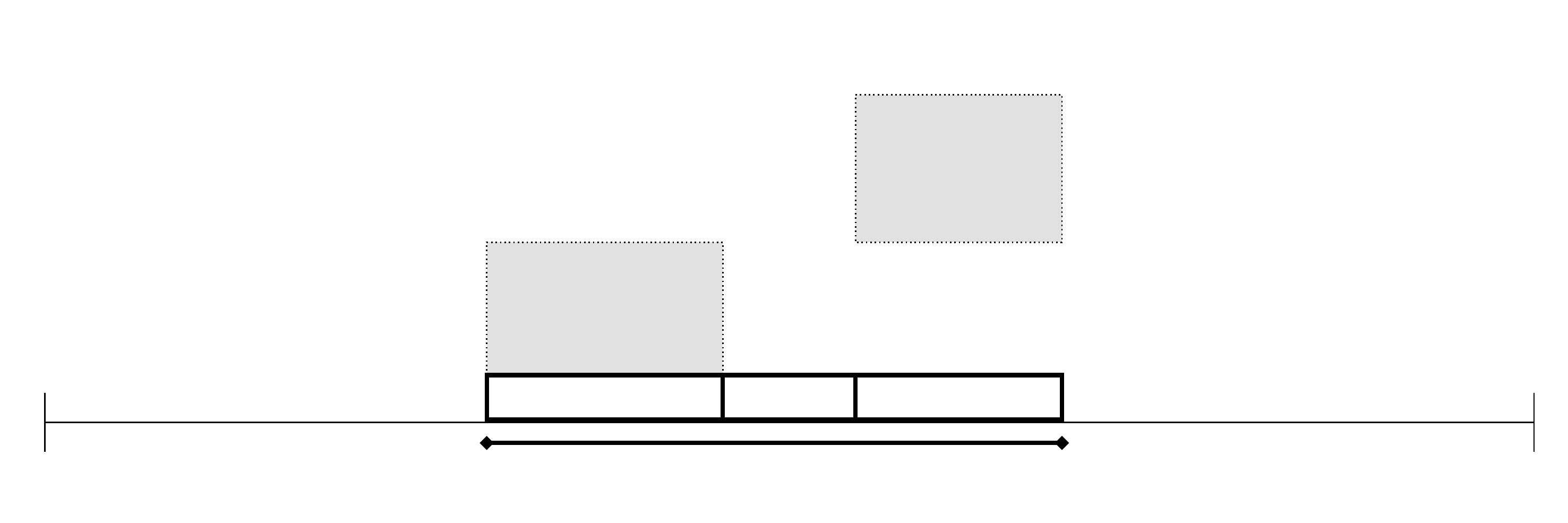}}
		\put(235, 10){$I$}
		\put(180, 34){$L$}
		\put(285, 34){$R$}
		\put(235, 34){$M$}
		\put(180, 60){$T^{(L)}$}
		\put(285, 105){$T^{(R)}$}
		\end{picture}
		\caption{Depiction of a $(c, \alpha,\beta)$-splittable interval, as defined in Definition~\ref{def:splittable}. The interval $I$ is divided into three adjacent intervals, $L, M$, and $R$, and the disjoint monotone sequences are divided so that $T^{(L)}$ contains the indices $(i_1, \dots, i_c)$ and $T^{(R)}$ contains the indices $(i_{c+1}, \dots, i_k)$. Furthermore, we have that every $(i_1, \dots, i_c) \in T^{(L)}$ and $(j_{c+1}, \dots, j_k) \in T^{(R)}$ have $f(i_c) < f(j_{c+1})$, so that any monotone pattern of length $c$ in $E(T^{(L)})$ may be combined with any monotone pattern of length $k-c$ in $E(T^{(R)})$ to obtain a monotone pattern of length $k$ within $I$.}
		\label{fig:splittable}
	\end{figure}

	We now state the main theorem of this section, whose proof will be given in Section~\ref{sec:proof1}. 
	\begin{theorem}\label{thm:main-structure}
		Let $k, k_0 \in \N$ be positive integers satisfying $1\leq k_0\leq k$, and let $\delta\in (0,1)$ and let $C > 0$. Let $f \colon [n] \to \R$ be a function and let $T_0 \subseteq [n]^{k_0}$ be a set of $\delta n$ disjoint monotone subsequences of $f$ of length $k_0$. 
		Then there exists an $\alpha \geq \Omega(\delta/k^5)$ such that at least one of the following conditions holds.
		\begin{enumerate}
			\item \label{en:suffix} 
				Either there exists a set $H \subseteq [n]$, of indices that start an $(\alpha, C k \alpha)$-growing suffix, satisfying $\alpha |H| \geq \delta n / \poly(k, \log(1/\delta))$; or
			\item \label{en:split} 
				There exists an integer $c$ with $1 \le c < k_0$, a set $T$, with $E(T) \subseteq E(T_0)$, of disjoint length-$k_0$ monotone subsequences, and a $(c, 1/(6k),\alpha)$-splittable collection of $T$, of disjoint interval-tuple pairs $(I_1, T_1), \dots, (I_s, T_s)$, such that 
				\[
					\alpha \sum_{h=1}^s |I_h| \geq \frac{|T_0|}{\poly(k, \log(1/\delta))}.
				\]
		\end{enumerate}
	\end{theorem}

	We remark that the above theorem is stated with respect to the two parameters, $k_0$ and $k$, for ease of applicability. In particular, in the next section, we will apply Theorem~\ref{thm:main-structure} multiple times, and it will be convenient to have $k$ be fixed and $k_0$ be a varying parameter. In that sense, even though the monotone subsequences in question have length $k_0$, the relevant parameters that Theorem~\ref{thm:main-structure} lower bounds only depend on $k$.

	Consider the following scenario: $f \colon [n] \to \R$ is a sequence which is $\eps$-far from $(12\dots k)$-free, so by Observation~\ref{obs:disjoint-families}, there exists a set $T_0$ of disjoint, length-$k$ monotone subsequences of $f$ of size at least $\eps n / k$. Suppose that upon applying Theorem~\ref{thm:main-structure} with $k_0 = k$ and $\delta = \eps / k$, (\ref{en:split}) holds. Then, there exists a $(c, 1/(6k),\alpha)$-splittable collection of a large subset of disjoint, length-$k$ monotone subsequences $T$ into disjoint interval-tuple pairs $(I_1, T_1), \dots, (I_s, T_s)$. For each $h \in [s]$, the pair $(I_h, T_h)$ is $(c, 1/(6k),\alpha)$-splittable, so let $I_h = L_h \cup M_h \cup R_h$ be the left, middle, and right intervals of $I_h$; furthermore, let $T_{h}^{(L)}$ be the $(12\dots c)$-patterns in $L_h$ which appear as prefixes of $T_h$, and $T_h^{(R)}$ be the $(12\dots (k-c))$-patterns in $R_h$ which appear as suffixes of $T_h$ in $R_h$. Thus, the restricted function $f_{|L_h} \colon L_h \to \R$ contains $|T_h|$ disjoint $(12\dots c)$-patterns, and $f_{|R_h} \colon R_h \to \R$ contains $|T_h|$ disjoint $(12\dots (k-c))$-patterns. This naturally leads to a recursive application of Theorem~\ref{thm:main-structure} to the function $f_{|L_h}$ with $k_0 = c$, and to the function $f_{|R_h}$ with $k_0 = k-c$, for all $h \in [s]$. 

\subsection{Tree descriptors}\label{ssec:tree:descriptors}

	We now introduce the notion of \emph{tree descriptors}, which will summarize information about a function $f$ after applying Theorem~\ref{thm:main-structure} recursively. Then, we state the main structural result for functions that are $\eps$-far from $(12\dots k)$-free. The goal is to say that every function which is $\eps$-far from $(12\dots k)$-free either has many growing suffixes, or there exists a tree descriptor which describes the behavior of many disjoint, length-$k$ monotone subsequences in the function. The following two definitions make up the notion of a tree descriptor representing a function. Figure~\ref{fig:tree-descriptor} shows an example of Definitions~\ref{def:tree-descriptor} and~\ref{def:tree-rep}.

	\begin{definition}\label{def:tree-descriptor}
		Let $k_0 \in \N$ and $\delta \in (0,1)$.
		A \emph{$(k_0, \delta)$-weighted-tree} is a pair $(G, \varrho)$, where
		\begin{itemize}
			\item 
				$G = (V, E, w)$ is a rooted binary tree with edges labeled by a function $w \colon E \to \{0,1\}$. Every non-leaf node has two outgoing edges, $e_0, e_1$ with $w(e_0) = 0$ and $w(e_1) = 1$. The set of leaves $V_{\ell} \subseteq V$ satisfies $|V_{\ell}| = k_0$, and $\leq_{G}$ is the total order defined on the leaves by the values of $w$ on a root-to-leaf path.\footnote{Specifically, for $l_1, l_2 \in V_{\ell}$ at depths $d_1$ and $d_2$, with root to leaf paths $(r, u^{(1)}, \dots, u^{(d_1-1)}, l_1)$ and $(r, v^{(1)}, \dots, v^{(d_2-1)}, l_2)$, then $l_1 \leq_{G} l_2$ if and only if $(w(r, u^{(1)}), w(u^{(1)}, u^{(2)}), \dots, w(u^{(d_1-1)}, l_1)) \leq (w(r, v^{(1)}), w(v^{(1)}, v^{(2)}), \dots, w(v^{(d_2-1)}, l_2))$ in the natural partial order on $\{0,1\}^*$.} 
			\item 
				$\varrho \colon V \to [\lceil \log(1 / \delta) \rceil]$ is a function that assigns a positive integer to each node of $G$.
		\end{itemize}
	\end{definition}

	In the next definition, we show how we use weighted trees to represent a function $f$ and a set of disjoint, length-$k_0$ monotone subsequences. 
	\begin{definition}\label{def:tree-rep}
		Let $k, k_0 \in \N$ be such that $1 \leq k_0 \leq k$, let $\alpha \in (0,1)$, let $I \subseteq \N$ be an interval, and let $f \colon I \to \R$ be a function. Let $T \subseteq I^{k_0}$ be a set of disjoint monotone subsequences of $f$.
		A triple $(G, \varrho, \sfI)$ is called a $(k, k_0, \delta)$-\emph{tree descriptor}\footnote{We shall sometimes refer to this as a $k_0$-tree descriptor, in particular when $k, \delta$ are not crucial to the discussion.} of $(f,T,I)$, if $(G, \varrho)$ is a $(k_0, \delta)$-weighted tree, $\sfI$ is a function $\sfI \colon V \rightarrow \calP(\calI)$ (where $V = V(G)$), and the following recursive definition holds. 
		\begin{enumerate}
		\item If $k_0 = 1$ (so $T \subseteq I$),
		\begin{itemize}
			\item 
				The graph $G= (V, E, w)$ is the rooted tree with one node, $r$, and no edges.
			\item 
				The function $\varrho \colon V \to [\lceil \log(1 / \delta) \rceil]$ (simply mapping one node) satisfies $ 2^{-\varrho(r)} \le |T|/|I| \le 2^{-\varrho(r) + 1}$.
			\item 
				The map $\sfI \colon V \to \calS(I)$ is given by $\sfI(r) = \{ \{ t \} : t \in T\}$.
		\end{itemize}
		\item If $k_0 > 1$, 
		\begin{itemize}
			\item 
				The graph $G = (V, E, w)$ is a rooted binary tree with $k_0$ leaves. We refer to the root by $r$, the left child of the root (namely, the child incident with the edge given $0$ by $w$) by $v_{\sf left}$, and the right child of the root (the child incident with the edge given $1$) by $v_{\sf right}$. Let $c$ be the number of leaves in the subtree of $v_{\sf left}$, so $v_{\sf right}$ has $k_0 - c$ leaves in its subtree.
			\item 
				Write $\sfI(r) = \{I_1, \dots, I_s\}$. Then $I_1, \dots, I_s$ are disjoint sub-intervals of $I$, and, setting $T_i = (I_i)^{k_0} \cap T$, the pairs $(I_1, T_1), \dots, (I_s, T_s)$ form a $(c, 1/(6k), 2^{-\varrho(r)})$-splittable collection of $T$, and
				\[
					2^{-\varrho(r)} \sum_{h = 1}^s |I_h| \ge \frac{|T|}{\poly(k, \log(1/\delta))^k}.
				\]
			\item 
				For each $h\in [s]$ there exists a partition $(L_h, M_h, R_h)$ of $I_h$ that satisfies Definition~\ref{def:splittable}, such that the sets $T_h^{(L)}$, of prefixes of length $c$ of subsequences in $T_h$, and $T_h^{(R)}$, of suffixes of length $k_0 - c$ of subsequences in $T_h$, satisfy $T_h^{(L)} \subseteq (L_h)^c$ and $T_h^{(R)} \subseteq (R_h)^{k_0 - c}$. Moreover, the following holds.
				
				The tuple $(G_{\sf left}, \varrho_{\sf left}, \sfI_{h, \, \sf left})$ is a $(k, c, \delta)$-tree descriptor of $f$, $T_h^{(L)}$, and $L_h$, where $G_{\sf left}$ is the subtree rooted at $v_{\sf left}$, $\varrho_{\sf left}$ is the restriction of $\varrho$ to the subtree $G_{\sf left}$, and $\sfI_{h,\, \sf left}$ is defined by $\sfI_{h,\, \sf left}(v) \eqdef \{J \in \sfI(v) \colon J \subseteq L_h\}$ for all $v \in G_{\sf left}$. 

				Analogously, the tuple $(G_{\sf right}, \varrho_{\sf right}, \sfI_{h,\, \sf right})$ is a $(k,k_0 - c, \delta)$-tree descriptor of $f$, $T_h^{(R)}$, and $R_h$, where $G_{\sf right}, \varrho_{\sf right}, \sfI_{h,\, \sf right}$ are defined analogously.
		\end{itemize}
	\end{enumerate}
	\end{definition}

	We remark that it is \emph{not} the case that for every function $f \colon I \to \R$ defined on an interval $I$, and for every $T \subseteq I^{k_0}$ which is a set of disjoint, length-$k_0$ monotone subsequences of $f$, there must exist a $k_0$-tree descriptor which represents $(f, T,I)$. 
	The goal will be to apply Theorem~\ref{thm:main-structure} recursively whenever we are in (\ref{en:split}), and to find a sufficiently large set $T$ of disjoint length-$k$ monotone subsequences, as well as a $k$-tree descriptor which represents $(f, T,I)$.  

	\begin{figure}
		\begin{picture}(300, 280)
		\put(0,0){\includegraphics[width=\linewidth]{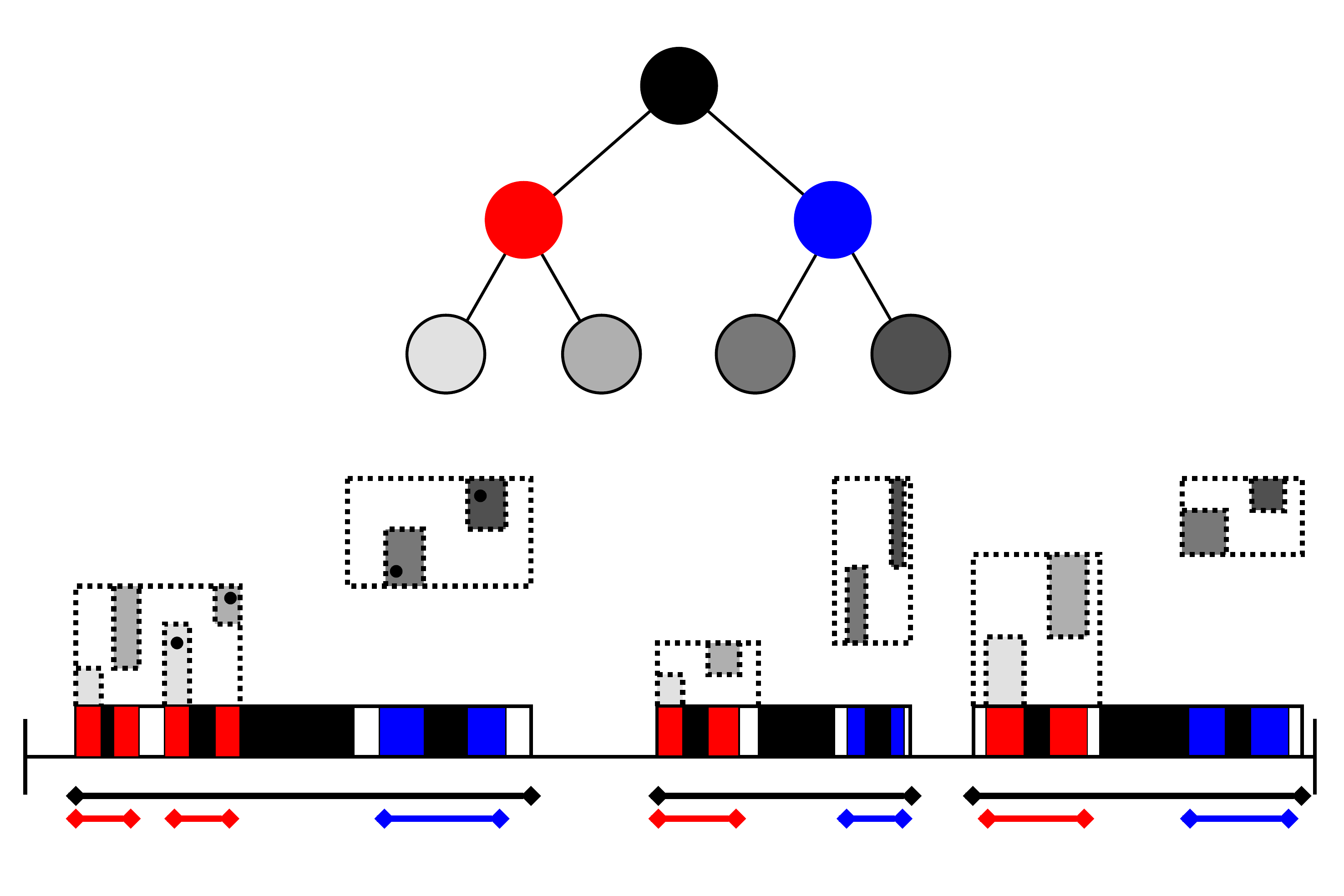}}
		\put(154, 187){$1$}
		\put(209, 187){$2$}
		\put(262, 187){$3$}
		\put(317, 187){$4$}
		\put(236, 282){\textcolor{white}{$r$}}
		\put(180, 235){$v_0$}
		\put(288, 235){\textcolor{white}{$v_1$}}
		\put(60, 39){$i_1$}
		\put(78, 39){$j_2$}
		\put(138, 39){$l_3$}
		\put(168, 39){$h_4$}
		\end{picture}
			\caption{Depiction of a tree descriptor $(G, \varrho, \sfI)$ representing $(f, T, I)$, as defined in Definitions~\ref{def:tree-descriptor} and~\ref{def:tree-rep}. The graph $G$ displayed above is a rooted tree with four leaves, which are ordered and labeled left-to-right. The root node $r$, filled in black, has its corresponding intervals from $\sfI(r)$ shown below the sequence as three black intervals. Each of the black intervals in $\sfI(r)$ is a $(2, \alpha,\beta)$-splittable interval, for $\alpha \approx 1/3$ and $\beta \geq 1/6$. Then, the root has the left child $v_0$, filled in red, and the right child $v_1$, filled in blue. The red intervals are those belonging to $\sfI(v_0)$, and the blue intervals are those belonging to $\sfI(v_1)$. Each black interval in $\sfI(r)$ has a left part, which contains intervals in $\sfI(v_0)$, and a right part, which contains intervals in $\sfI(v_1)$. The red and blue intervals in $\sfI(v_0)$ and $\sfI(v_1)$ are also $(1, \alpha,\beta)$-splittable, and the left part of the red intervals contains indices which will form the 1 in the monotone pattern of length $4$, and the right part of the red intervals contains indices which will form the 2. Likewise, the left part of blue intervals will contain the indices corresponding to 3, and the right part of the blue intervals will contain indices corresponding to 4. The regions where the indices from $T$ lie are shown above the sequence, where the indices 1--4 of some monotone pattern in $T$ lie in regions which are progressively darker. In order to see how a monotone subsequence may be sampled given that $(G, \ell, \sfI)$ is a tree descriptor for $(f, T, I)$ with sufficiently large $T$, consider indices $i_1$ and $j_2$ that belong to some subsequences from $T$, and lie in different shaded regions of the same red interval, within a black interval; and furthermore, $l_3$ and $h_4$ belong to some subsequence from $T$, and lie in different shaded regions of the same blue interval, within the same black interval as $i_1$ and $j_2$; then, the subsequence $(i_1, j_2, l_3, h_4)$ is a monotone subsequence even though $(i_1, j_2, l_3, h_4) \notin T$.
		}
		\label{fig:tree-descriptor}
	\end{figure}
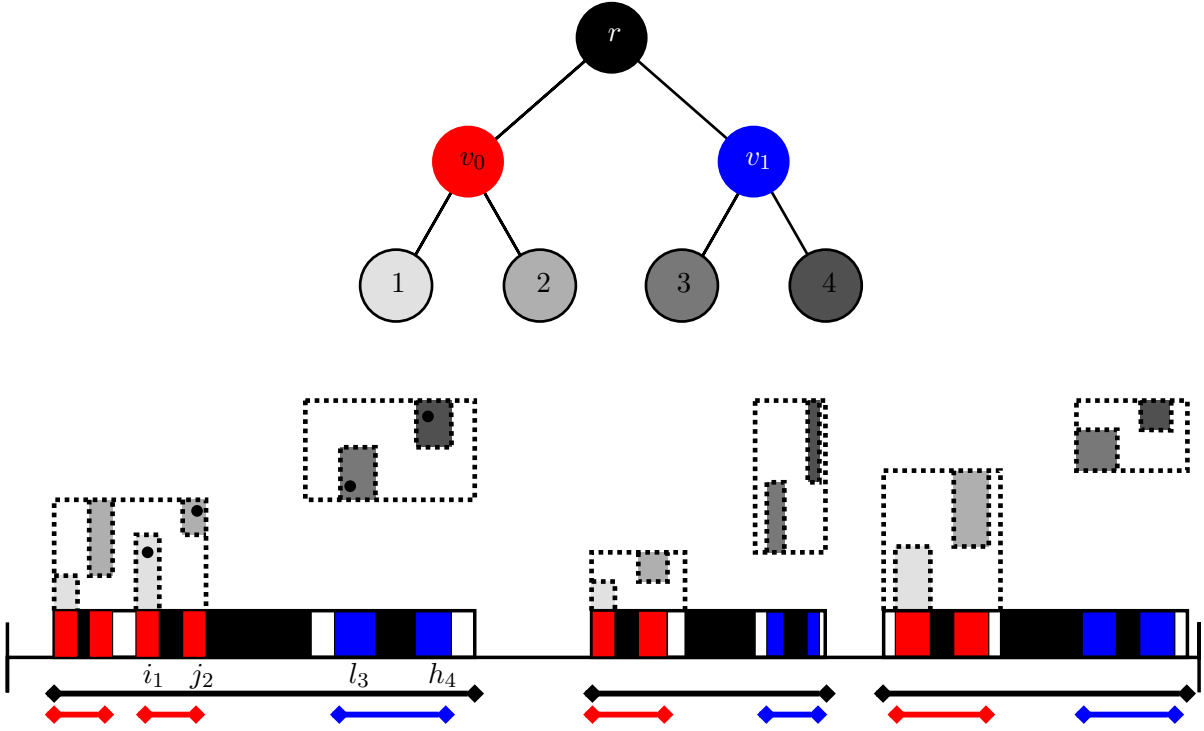

\subsection{The structural dichotomy theorem}

	We are now in a position to state the main structural theorem of far-from-$(12\dots k)$-free sequences, which guarantees that every far-from-$(12\dots k)$-free sequence either has many growing suffixes, or can be represented by a tree descriptor. The algorithm for finding a $(12\dots k)$-pattern will proceed by considering the two cases independently. The first case, when a sequence has many growing suffixes, is easy for algorithms; we will give a straight-forward sampling algorithm making roughly $O_k(\log n/\eps)$ queries. The second case, when a sequence is represented by a tree descriptor is the ``hard'' case for the algorithm.

	\begin{theorem}[Main structural result]
	  \label{thm:tree}
		Let $k \in \N$, $\eps > 0$, and let $f\colon [n] \to \R$ be a function which is $\eps$-far from $(12 \dots k)$-free. Then one of the following holds, where $C > 0$ is a large constant.
		\begin{itemize}
			\item
				There exists a parameter $\alpha \ge \eps / \poly(k, \log(1/\eps))^k$, and a set $H \subseteq [n]$ of indices which start an $(\alpha, C k \alpha)$-growing suffix, with
				\[ 
					\alpha |H| \geq \frac{\eps n}{\poly(k, \log(1/\eps))^{k}}, 
				\]
			\item
				or there exists a set $T \subseteq [n]^k$ of disjoint monotone subsequences of $f$ satisfying
				\[ 
					|T| \geq \frac{\eps n}{\poly(k, \log(1/\eps))^{k^2}}
				\]
				and a $(k, k, \beta)$-tree descriptor $(G, \varrho, \sfI)$ which represents $(f, T, [n])$, where $\beta \ge \eps / \poly(k, \log(1/\eps))^{k^2}$.
		\end{itemize}
	\end{theorem}

	\begin{proof}
		We shall prove the following claim, by induction, for all $k_0 \in [k]$. Here $C > 0$ is a large constant, and $C' > 0$ is a large enough constant such that $\alpha \ge \delta / (C' k^{5})$ in the statement of Theorem~\ref{thm:main-structure}, applied with the constant $C$.
		\begin{description}
			\item[Claim.] 
				Let $K = C'  k^5$ and let $P(\cdot, \cdot)$ be the function from the statement of Theorem~\ref{thm:main-structure}; so $P(x,y) = \poly(x, \log y)$, and we may assume that $P$ is increasing in both variables. 
				Let $A(\cdot, \cdot)$ and $B(\cdot, \cdot)$ be increasing functions, such that
				\begin{align} \label{eqn:ABC}
					\begin{split}
					& A(k_0, 1/\delta) \ge 12k \lceil \log(K^{k_0}/\delta) \rceil \cdot P( k, 1/\delta) \cdot A(k_0-1, K/\delta) \\
					& A(1, 1/\delta) = 1/\delta \\
					& B(k_0, 1/\delta) \ge 2 \cdot P(k, K/\delta) \cdot \left(2k \lceil \log(K B(k_0-1, K/\delta)/\delta) \rceil\right)^{2k_0} \cdot B(k_0-1, K/\delta) \\
					& B(1, 1/\delta) = 1/\delta \\					
					\end{split}
				\end{align}
				Note that there exists such $A(\cdot, \cdot)$ and $B(\cdot, \cdot)$ with $A(k, 1/\delta) = (\poly(k, \log(1/\delta)))^k$ and $B(k, 1/\delta) = (\poly(k, \log(1/\delta)))^{k^2}$.

				Let $I \subseteq \N$ be an interval, let $g$ be a sequence $g \colon I \to \R$, let $T_0 \subseteq I^{k_0}$ be a set of disjoint length-$k_0$ monotone subsequences, and define $\delta \eqdef |T_0|/|I|$. Then 			
				\begin{enumerate}
					\item \label{case:H}
						Either there exists $\alpha \geq \delta / K^{k_0}$, which is an integer power of $1/2$, along with a set $H \subseteq I$ of $(\alpha, C k \alpha)$-growing suffix start points such that 
						\[ 
							\alpha |H| \geq \frac{\delta|I|}{A(k_0, 1/\delta)}, 
						\]
					\item \label{case:split}
						Or there exists a set $T \subseteq I^{k_0}$ of disjoint $k_0$-tuples satisfying $E(T) \subseteq E(T_0)$ and
						\[ 
							|T| \geq \frac{|T_0|}{B(k_0, 1/\delta)} 
						\]
						and a ($k, k_0, \alpha$)-tree descriptor $(G, \varrho, \sfI)$ for $(g, T, I)$, where $\alpha \ge \delta / B(k_0, 1/\delta)$.
				\end{enumerate}
		\end{description}
		Note that since $f$ is $\eps$-far from $(12 \dots k)$-free, there is a set $T_0 \subseteq [n]^k$ of at least $\eps n / k$ disjoint length-$k$ monotone  subsequences. By applying the above claim for $k_0 = k$, $T_0$, $[n]$ and $f$, the theorem follows. Thus, it remains to prove the claim; we proceed by induction.
		\begin{description}
			\item[if $k_0=1$:] 
				Note that here $T_0$ is a subset of $I$.
				We define the ($k,1, \delta$)-tree descriptor $(G, \varrho, \sfI)$ which represents $f, T = T_0, I$ in the natural way:
				\begin{itemize}
					\item 
						$G= (V, E)$ is a rooted tree with one node: $V = \{r\}$ and $E = \emptyset$. 
					\item 
						$\varrho\colon V\to\N$ is given by $\varrho(r) = \lceil \log(1/\delta) \rceil$, so $2^{-\varrho(r)} \leq |I\cap T|/|I| \leq 2^{-\varrho(r)+1}$.
					\item 
						$\sfI \colon V \to \calS(I)$ is given by $\sfI(r) = \{ \{ t \} : t \in T\}$.
				\end{itemize}
		 
			\item[if $2\leq k_0\leq k$:] 
				By Theorem~\ref{thm:main-structure}, there exists $\alpha \ge \delta/ K$ such that one of~\eqref{en:suffix} and~\eqref{en:split}, from the statement of the theorem, holds.
				\begin{itemize}
					\item 
						If~\eqref{en:suffix} holds, there is a set $H \subseteq I$ of $(\alpha, C k \alpha)$-growing suffix start points with 
						\[
							\alpha |H| \geq \frac{\delta |I|}{P(k, 1/\delta)};
						\]
						note that we may assume that $\alpha$ is an integer power of $1/2$.\footnote{to be precise and to ensure that we can take $\alpha$ to be an integer power of $2$, it might be better to apply Theorem~\ref{thm:main-structure} with constant $2C$, to allow for some slack; this does not change the argument.}
					\item 
						Otherwise,~\eqref{en:split} holds, and we are given an integer $c\in[k_0-1]$, a set $T$ of disjoint length-$k_0$ monotone subsequences, with $E(T) \subseteq E(T_0)$, and a $(c, 1/(6k),\alpha)$-splittable collection of $T$ into disjoint interval-tuple pairs $(I_1, T_1), \dots, (I_s, T_s)$, such that
						\[
							\alpha \sum_{h=1}^s |I_h| 
							\geq \frac{|T_0|}{P(k, 1/\delta)} 
							= \frac{\delta |I|}{P(k, 1/\delta)}.
						\]
						Recall that by definition of splittability, $|T_h| / |I_h| \ge \alpha$ for every $h \in [s]$.
				\end{itemize}
				If \eqref{en:suffix} holds, we are done; so we assume that \eqref{en:split} holds.

				For each $h \in [s]$, since $(I_h, T_h)$ is a $(c, 1/(6k), \alpha)$-splittable pair, there exists a partition $(L_h, M_h, R_h)$ that satisfies the conditions stated in Definition~\ref{def:splittable}. Let $T_h^{(L)}$ be the collection of prefixes of length $c$ of subsequences in $T_h$, and let $T_h^{(R)}$ be the collection of suffixes of length $k_0 - c$ of subsequences in $T_h$.
				 
				We apply the induction hypothesis to each of the pairs $(L_h, T_h^{(L)})$ and $(R_h, T_h^{(R)})$.
				We consider two cases for each $h \in [s]$.
				\begin{enumerate}
					\item
						\eqref{case:H} holds for either $(L_h, T_h^{(L)})$ or $(R_h, T_h^{(R)})$. This means that there exists $\beta_h$, which is an integer power of $1/2$, and which satisfies $\beta_h \ge \alpha / K^{\max\{c, k_0 - c\}} \ge \alpha / K^{k_0 - 1} \ge \delta / K^{k_0}$, and a set $H_h \subseteq I_h$ of start points of $(\beta_h, Ck \beta_h)$-growing subsequences, such that (using $|R_h|, |L_h| \ge |I_h| / (6k)$)
						\[
							\beta_h |H_h| \ge 
							\frac{\alpha |I_h|}{6k \cdot A(k_0-1, 1/\alpha)} 
						\]
					\item
						Otherwise, \eqref{case:split} holds for both $(L_h, T_h^{(L)})$ and $(R_h, T_h^{(R)})$. Setting $\beta = \alpha / B(k_0-1, 1/\alpha)$, this means that there exists a $(k, c, \beta)$-tree descriptor $(G_h^{(L)}, \varrho_h^{(L)}, \sfI_h^{(L)})$, for $(g, \calL_h, L_h)$ where $\calL_h \subseteq (L_h)^c$ is a set of length-$c$ monotone subsequences, such that $E(\calL_h) \subseteq E(T_h^{(L)})$ and
						\begin{equation} \label{eqn:calL-h}
							|\calL_h| \ge 
							\frac{|T_h^{(L)}|}{B(k_0-1, 1/\alpha)}, 
						\end{equation}
						and, similarly, there exists a $(k, k_0 - c, \beta)$-tree descriptor $(G_h^{(R)}, \varrho_h^{(R)}, \sfI_h^{(R)})$ for $(g, \calR_h, L_h)$, where $\calR_h \subseteq (R_h)^{k_0-c}$ is a set of length-$(k_0 - c)$ monotone subsequences, such that $E(\calR_h) \subseteq E(T_h^{(R)})$ and
						\begin{equation} \label{eqn:calR-h}
							|\calR_h| \ge 
							\frac{|T_h^{(R)}|}{B(k_0-1, 1/\alpha)}.
						\end{equation}
						For convenience, we shall assume that $|\calL_h| = |\calR_h|$, by possibly removing some elements of the largest of the two (and reflecting this in the corresponding tree descriptor).
				\end{enumerate}
				Suppose first that 
				\[
					\sum_{h \colon \, \text{first case holds for $h$}} |I_h| \ge \frac{1}{2} \cdot \sum_{h = 1}^s |I_h|.
				\]
				Since each $\beta_h$ is an integer power of $1/2$, there are at most $\lceil \log(K^{k_0} / \delta) \rceil$ possible values for $\beta_h$. Hence, there exists some $\beta$ (with $\beta \ge \delta / K^{k_0}$) such that the collection $S$, of indices $h \in [s]$ for which the first case holds for $h$ and $\beta_h = \beta$, satisfies 
				\[
					\sum_{h \in S} |I_h| \ge \frac{1}{2 \lceil \log(K^{k_0}/\delta) \rceil} \cdot \sum_{h = 1}^s |I_h|.
				\]
				Let $H = \bigcup_{h \in S} H_h$. Then $H$ is a set of start points of $(\beta, Ck \beta)$-growing suffixes, with
				\begin{align*}
					\beta |H| 
					&\ge \frac{\alpha}{6k \cdot A(k_0-1, 1/\alpha)} \cdot \sum_{h \in S} |I_h| 
					\ge \frac{\alpha}{12 k \lceil \log(K^{k_0}/\delta) \rceil \cdot A(k_0-1, 1/\alpha)} \cdot \sum_{h = 1}^s |I_h| \\
					&\ge \frac{\delta |I|}{12 k \lceil \log(K^{k_0}/\delta) \rceil \cdot P(k, 1/\delta) \cdot A(k_0-1, 1/\alpha) } 
					\ge \frac{\delta|I|}{A(k_0, 1/\delta)},
				\end{align*}
				where the last inequality follows from \eqref{eqn:ABC}.
				This proves the claim in this case.

				Next, we may assume that
				\[
					\sum_{h \colon \, \text{second case holds for $h$}} |I_h| \ge \frac{1}{2} \cdot \sum_{h = 1}^s |I_h|.
				\]
				Note that the number of quadruples $(G_h^{(L)}, \varrho_h^{(L)}, G_h^{(R)}, \varrho_h^{(R)})$ (whose elements are as above) is at most $(2c)^{2c} (2(k_0-c))^{2(k_0-c)} (\lceil \log(1/\beta) \rceil)^{2k_0} \le (2k \lceil \log(1/\beta) \rceil)^{2k_0}$, since the number of trees on $l$ vertices is at most $l^l$, and we have at most $\lceil \log(1/\beta) \rceil$ possible weights to assign to each of the vertices. It follows that there exists such a quadruple $(G_L^\ast, \varrho_L^\ast, G_R^\ast, \varrho_R^\ast)$ such that if $S$ is the set of indices $h$ that were assigned this quadruple, then
				\begin{align} \label{eqn:intervals-dense}
					\begin{split}					
					\alpha \cdot \sum_{h \in S} |I_h| 
					&\ge \frac{\alpha}{(2k\lceil \log(1/\beta) \rceil)^{2k_0}} \cdot \sum_{\text{second case holds for $h$}}|I_h| \\
					&\ge \frac{\alpha}{2\cdot(2k\lceil \log(1/\beta) \rceil)^{2k_0}} \cdot \sum_{h = 1}^s|I_h| 
					\ge \frac{|T_0|}{2 \cdot P(k, 1/\delta) \cdot (2 k \lceil \log(1/\beta) \rceil)^{2k_0}}.
					\end{split}
				\end{align}
				We form a set $\calT_h$ of monotone length-$k_0$ subsequences by matching elements from $\calL_h$ with elements from $\calR_h$ for each $h \in S$; that they can be matched follows from the assumption that $|\calL_h| = |\calR_h|$, and that these form monotone subsequences follows from the assumptions on $\calL_h, \calR_h$. Set $\calT := \cup_{h \in S} \calT_h$.
				Note that $(I_h, \calT_h)$ is $(k_0, c, \beta)$-splittable by \eqref{eqn:calL-h} and \eqref{eqn:calR-h} (using $\beta = \alpha / B(k_0-1, 1/\alpha)$). 
				Let $(G, \varrho)$ be the $(k, k_0, \beta)$-weighted-tree obtained by taking a root $r$, with weight $\varrho(r) = \lceil \log(1 / \beta) \rceil$, adding the tree $(G_L^\ast, \varrho^\ast)$ as a subtree to its left (i.e., the root of this tree is joined to $r$ by an edge with value $0$) and adding the tree $(G_R^\ast, \varrho^\ast)$ as a subtree to its right.
				Now, we form a $(G, \varrho, \sfI)$-tree descriptor by setting
				\[
					\sfI(v) = 
					\left\{
						\begin{array}{ll}
							\{I_h : h \in S\} & v = r \\
							\bigcup_{h \in S} \sfI_h^{(L)}(v) & v \in G_L^{\ast} \\
							\bigcup_{h \in S} \sfI_h^{(R)}(v) & v \in G_R^{\ast}.
						\end{array}
					\right.
				\]
				We claim that $(G, \varrho, \sfI)$ is a $(k, k_0, \beta)$-tree descriptor for $(g, \calT, I)$. 
				Indeed, $((I_h, \calT_h))_{h \in S}$ is a $(c, 1/(6k), 2^{-\varrho(r)})$-splittable collection of $\calT$, and, by \eqref{eqn:intervals-dense} and because $|T_0| \ge |\calT|$ 
				\begin{align*}
					2^{-\varrho(r)} \sum_{h \in S} |I_h| 
					& \ge 
					\frac{\alpha}{2} \cdot \sum_{h \in S} |I_h| \ge
					\frac{|\calT|}{4 \cdot P(k, 1/\delta) \cdot (2k\lceil \log(1/\beta) \rceil)^{2k_0}}  =
					\frac{|\calT|}{\poly(k, \log(1/\delta))^k}.
				\end{align*}
				The remaining requirements in the recursive defnition of a tree descriptor (see Definition~\ref{def:tree-rep}) follow as $(G_L^{\ast}, \varrho^{\ast}, \sfI_h^{(L)})$ is a $(k, c, \beta)$-tree descriptor for $(g, \calL_h, L_h)$ and $(G_L^{\ast}, \varrho^{\ast}, \sfI_h^{(R)})$ is a $(k, k_0-c, \beta)$-tree descriptor for $(g, \calR_h, R_h)$ for every $h \in S$. Since $\beta = \alpha / B(k_0-1, 1/\alpha) \ge  \delta / B(k_0, 1/\delta)$, it follows that $(G, \varrho, \sfI)$ is a $(k, k_0, \delta / B(k_0, 1/\delta))$-tree descriptor for $(g, \calT, I)$.
				
				It remains to lower-bound the size of $\calT$. Using \eqref{eqn:calR-h} and \eqref{eqn:intervals-dense}, we have
				\begin{align*}
					|\calT| 
					& = \sum_{h \in S} |\calR_h|
					\ge \frac{1}{B(k_0-1, 1/\alpha)}\cdot \sum_{h \in S} |T_h|
					\ge \frac{\alpha}{B(k_0-1, 1/\alpha)} \cdot \sum_{h \in S} |I_h|  \\
					& \ge \frac{|T_0|}{2 \cdot P(k, 1/\delta) \cdot (2k \lceil \log(1/\beta) \rceil)^{2k_0} \cdot B(k_0-1, 1/\alpha)} \ge \frac{|T_0|}{B(k_0, 1/\delta)}. 
				\end{align*}
				This completes the proof of the inductive claim in this case.
				\qedhere
		\end{description}
	\end{proof}
 
\subsection{Proof of Theorem~\ref{thm:main-structure}}\label{sec:proof1}

	We now prove Theorem~\ref{thm:main-structure}. For the rest of this section, let $k, k_0 \in \N$, with $1 \leq k_0 \leq k$, be fixed, and let $f \colon [n] \to \R$ be a fixed function. Let $T_0$ be a set of $\delta n$ disjoint monotone subsequences of $f$ of length $k_0$. We apply Lemma~\ref{lem:big-split} to the set $T_0$; this specifies an integer $c \in [k_0-1]$ and a subset $T$ of at least $\delta n / k^2$ disjoint monotone subsequences of length $k_0$ satisfying the conclusion of Lemma~\ref{lem:big-split}. 

	\begin{definition}
		Let $(i_1, \dots, i_{k_0}) \in [n]^{k_0}$ be a monotone subsequence with a $c$-gap. We say that $(i_1, \dots, i_{k_0})$ is \emph{at scale $t$} if $2^{t} \leq i_{c+1} - i_c \leq 2^{t+1}$, where $t \in \{0, \dots, \lfloor \log n \rfloor\}$.
	\end{definition}

	\begin{definition}
		Let $(i_1, \dots, i_{k_0}) \in [n]^{k_0}$ be a monotone subsequence with a $c$-gap. For $\gamma \in (0,1)$, we say that $\ell \in [n]$ \emph{$\gamma$-cuts $(i_1, \dots, i_{k_{0}})$ at $c$ with slack} if
		\begin{equation} \label{eq:slack}
			i_c + \gamma(i_{c+1} - i_c) \leq \ell \leq i_{c+1} - \gamma(i_{c+1} - i_c). 
		\end{equation}
	\end{definition}

	We hereafter consider the parameter setting of $\gamma \eqdef 1/3$. For $\ell \in [n]$, $t \in \{0, \dots, \lfloor \log n \rfloor\}$, and any subset $U \subset T$ of disjoint $(12\dots k_0)$-patterns in $f$ let 
	\begin{align} 
		A_t(\ell, U) &= \{ (i_1, \dots, i_{k_0}) \in U : (i_1, \dots, i_{k_0}) \text{ is at scale $t$ and is $\gamma$-cut at $c$ with slack by $\ell$} \}. \label{eq:cut-sets}
	\end{align}
	We note that for each $(i_1, \dots, i_{k_0}) \in A_t(\ell, U)$, the index $i_{c+1}$ is in $[\ell, \ell + 2^{t+1}]$, and since $A_t(\ell, U)$ is made of disjoint monotone sequences, $|A_t(\ell, U)| \leq 2^{t+1}$.

	\begin{lemma}
		\label{lem:At:properties}
		For every $\ell \in [n]$, $t \in \{0, \dots, \lfloor \log n \rfloor \}$, and $U \subset T$,
		\begin{itemize}
			\item 
				Every $(i_1, \dots, i_{k_0}) \in A_t(\ell, U)$ satisfies 
				\[
					\ell - (k-1)2^{t+1} \leq i_1, \dots, i_c \le \ell - \gamma 2^t \qquad \qquad 
					\ell + \gamma 2^t \le i_{c+1}, \dots, i_{k_0} \leq \ell + (k-1) 2^{t+1}.
				\]
			\item 
				Let $t_1 \geq t_2 + 1 + \log(1/\gamma) + \log(c+1)$, $(i_1, \dots, i_{k_0}) \in A_{t_1}(\ell, U)$ and $(j_1, \dots, j_{k_0}) \in A_{t_2}(\ell, U)$. Then $f(j_{c+1}) < f(i_{c+1})$.
		\end{itemize}
	\end{lemma}
	\begin{proof}
		Fix any $\ell \in [n]$, $t\in \{0, \dots, \lfloor \log n \rfloor\}$ and $U \subset T$. To establish the first bullet, consider any $(i_1, \dots, i_{k_0}) \in A_t(\ell, U)$. 
		By definition of a $c$-gap sequence, we have
		\[
			i_1 \ge i_{c+1} - c (i_{c+1} - i_c) \ge \ell - (k-1) 2^{t+1},
		\]
		using $i_{c+1} - i_c \le 2^{t+1}$ and $i_{c+1} \ge \ell$. By \eqref{eq:slack}, we have $i_c \le \ell - \gamma 2^t$ (using $i_{c+1} - i_c \ge 2^t$).
		The first inequality follows as $i_1 < \dots < i_c$. The inequality for $i_{c+1}, \dots, i_{k_0}$ follows similarly.

		For the second bullet, let $(i_1, \dots, i_{k_0}) \in A_{t_1}(\ell, U)$ and $(j_1, \dots, j_{k_0}) \in A_{t_2}(\ell, U)$ and suppose that $2^{t_1} \geq 2^{t_2+1}\cdot (c+1)/\gamma$. 
		We have $i_c \le \ell - \gamma 2^{t_1}$ and $j_c \ge \ell - 2^{t_2+1}$ (using \eqref{eq:slack} and \eqref{eq:cut-sets}), from which it follows that $j_c > i_c$.
		Similarly, $i_1 < i_c \le \ell - \gamma 2^{t_1}$ and $j_1 \ge \ell - (c - 1) 2^{t_2 + 1}$, implying that $j_1 > i_1$, and $i_{c+1} \ge \ell + \gamma 2^{t_1}$ and $j_{c+1} \le \ell + 2^{t_2+1}$, which implies that $i_{c+1} > j_{c+1}$. The inequality $f(j_{c+1}) < f(i_{c+1})$ follows from the assumption that $T$ satisfies \eqref{itm:property-greedy} from Lemma~\ref{lem:big-split}.	
	\end{proof}

	The proof of Theorem~\ref{thm:main-structure} will follow by considering a random $\bell \sim [n]$ and the sets $A_1(\bell, T), \dots, A_{\lfloor\log n\rfloor}(\bell, T)$. By looking at how the sizes of the sets $A_1(\bell, T), \dots, A_{\log n-1}(\bell, T)$ vary, we will be able to say that $\bell$ is the start of a growing suffix, or identify a splittable interval. Towards this goal, we first establish a simple lemma; here $v(\ell, U)$ is defined to be $\sum_{t=0}^{\lfloor \log n \rfloor} |A_t(\ell, U)| / 2^t$.
	\begin{lemma}\label{lem:expect}
		Let $U \subset T$ be any subset and $\bell \sim [n]$ be sampled uniformly at random. Then
		\[
			\Ex_{\bell \sim [n]} v(\ell, U) \geq \frac{|U|}{3n}.
		\]
	\end{lemma}
	\begin{proof}
		Fix a sequence $i = (i_1, \dots, i_{k_0}) \in U$, and let $t(i)\in \{0, \dots, \lfloor\log n\rfloor\}$ be its scale. 
		Then, the probability (over a uniformly random $\bell$ in $[n]$) that $i$ belongs to $A_{t(i)}(\bell, U)$ is lower bounded as
		\[
			\Prx_{\bell\sim[n]}[ i\in A_{t(i)}(\bell, U) ] \geq \frac{(1-2\gamma)2^{t(i)}}{n} = \frac{2^{t(i)}}{3n}.
		\]
		Therefore,
		$
			\sum_{t=0}^{\log n-1} \sum_{i\in U \colon t(i) = t} \Pr_{\bell\sim[n]}[ i\in A_{t}(\bell, U) ] /2^t \geq |U|/(3n)
		$, or, equivalently, since $\Pr_{\bell\sim[n]}[ i\in A_{t}(\bell, U) ] = 0$ for $t\neq t(i)$,
		\[
			\Ex_{\bell \sim [n]}\left[  \sum_{t=0}^{\log n-1} \frac{ |A_{t}(\bell, U)|  }{2^t} \right]
			= \Ex_{\bell \sim [n]}\left[  \sum_{t=0}^{\log n-1} \sum_{i\in U} \frac{ \mathbbm{1}\{ i \in A_{t}(\bell, U)\} }{2^t} \right] \geq \frac{|U|}{3n},
		\]
		establishing the lemma.
	\end{proof}

	We next establish an auxiliary lemma that we will use in order to find growing suffixes.

	\begin{lemma}\label{lem:move-back}
		Let $\ell \in [n]$ and $U \subset T$ be such that every $t \in \{0, \dots, \lfloor \log n \rfloor\}$ satisfies $|A_t(\ell, U)| / 2^t \leq \beta$. Then, if $\ell' \in [n]$ is any index satisfying
		\begin{align}
			\max \{ i_c : (i_1, \dots, i_{k_0}) \in A_t(\ell, U), t \in \{0, \dots, \lfloor \log n \rfloor\} \leq \ell' \leq \ell,  \label{eq:bound-on-ell}
		\end{align}
	 	then $\ell'$ is the start of an $(4\beta, v(\ell, U) / (12 \log k))$-growing suffix.
	\end{lemma}

	\begin{proof}
		Let $\Delta = 1 + \log(1/\gamma) + \log(c+1) $, and notice that $3 \le \Delta \le 3\log k$. Then, there exists a set $\calT \subseteq \{0,\dots, \lfloor \log n \rfloor\}$ such that
		\begin{enumerate}
			\item 
				All distinct $t, t' \in \calT$ satisfy $|t - t' | \geq \Delta$; and,
			\item 
				$\sum_{t \in \calT} \frac{|A_t(\ell, U)|}{2^t} \geq \frac{1}{\Delta + 1} \sum_{t=0}^{\log n-1} \frac{|A_t(\ell, U)|}{2^t} = \frac{v(\ell, U)}{\Delta + 1}$.
		\end{enumerate}
		(Such a set exists by an averaging argument.) Now, consider the sets 
		\[
			D_t(\ell)= 
			\begin{cases}
				\{ i_{c+1} : (i_1,\dots,i_{k_0})\in A_t(\ell, U) \} &\text{ if } t\in\calT\\
				\emptyset &\text{ if } t\in\{0,\dots,\lfloor \log n \rfloor\}\setminus \calT .
			\end{cases}
		\]
		Considering any $\ell' \in [n]$ satisfying (\ref{eq:bound-on-ell}), we have the following for all $t \in \{0, \dots, \lfloor \log n \rfloor\}$ with $D_{t}(\ell) \neq \emptyset$: $\ell - 2^{t+1} \leq \ell' \leq \ell$; $\min D_t(\ell) \geq \ell + 2^t /3$; and $\max D_t(\ell) \leq \ell' + 2^{t+1}$. Therefore, $D_t(\ell) \subset S_{t-1}(\ell') \cup S_{t}(\ell') \cup S_{t+1}(\ell')$. (Recall that $S_t(a) = [a + 2^{t-1}, a + 2^t)$.)
		For each $t \in \calT$, let $n(t) \in \{ t-1, t, t+1\}$  satisfying $|D_t(\ell) \cap S_{n(t)}(\ell') | \geq |D_t(\ell)|/ 3$, and notice that all $n(t) \in \{0, \dots, \lfloor \log n \rfloor\}$ are distinct since $\Delta \geq 3$.

		The first condition in Definition~\ref{def:growing-suffixes} holds as the densities of $D_t(\ell) \cap S_{n(t)}(\ell')$ in the corresponding intervals $S_{n(t)}(\ell')$ are upper bounded by $|D_t(\ell)| / |S_{n(t)}(\ell')| \le |A_t(\ell, U)|/2^{t-2} \leq 4\beta$, and the sum of these densities satisfies 
		\[ 
			\sum_{t \in \calT} \frac{|D_t(\ell)\cap S_{n(t)}(\ell')|}{|S_{n(t)}(\ell')|} \geq \sum_{t \in \calT} \frac{|D_t(\ell)|}{3 \cdot 2^t} = \sum_{t \in \calT} \frac{|A_t(\ell, U)|}{3 \cdot 2^t} \geq \frac{v(\ell, U)}{3(\Delta + 1)},
		\]
		which is at least $v(\ell, U) / (12 \log k)$.
		The second condition in Definition~\ref{def:growing-suffixes} holds, because for any choice of $b \in D_t(\ell), b' \in D_{t'}(\ell)$ with $t < t'$, we have $t' \ge t + \Delta$ (by the choice of $\calT$), and hence $f(b) < f(b')$ by the second item of Lemma~\ref{lem:At:properties}.
	\end{proof}

	\begin{lemma}\label{lem:easy-growing}
		For every $\eta > 0$, there exists a subset $U \subset T$ such that every $(i_1, \dots, i_{k_0}) \in U$ has $i_{c}$ as the start of an $(1, \eta)$-growing suffix, and every $\ell \in [n]$ satisfies $v(\ell, T \setminus U) \leq 12\eta \log(k)$.
	\end{lemma}
		
	\begin{proof}
		Define sets $U_j$, elements $\ell_j$, and $k_0$-tuples $(i_{j,1}, \ldots, i_{j, k_0})$ recursively as follows. Set $U_0 := \emptyset$, and given a set $U_{j-1}$, if $v(\ell, T \setminus U_{j-1}) \le 12 \eta \log k$ for every $\ell \in [n]$, stop; otherwise, let $\ell_{j} \in [n]$ be such that $v(\ell_j, T \setminus U_j) > 12 \eta \log k$ and define $U_j = U_{j-1} \cup \{(i_{j,1}, \ldots, i_{j,k_0})\}$, where 
		\[
			i_{j,c} = \max\{i_c : (i_1, \ldots, i_{k_0}) \in T \setminus U_j \text{ and $(i_1, \ldots, i_{k_0})$ is $\gamma$-cut by $\ell_j$}\}.
		\]
		Let $j^*$ be the maximum $j$ for which $U_j$ was defined, and set $U := U_{j^*}$.
		Every $k_0$-tuple in $U$ is of the form $(i_{j,1}, \ldots, i_{j, k_0})$ for some $j \le j^*$. By Lemma~\ref{lem:move-back}, applied with $\ell = \ell_j$, $U = T \setminus U_{j-1}$, $i_{j,c}$, it follows that $i_{j,c}$ is the start of an $(1, \eta)$-growing suffix, for every $j$ for which $U_j$ was defined. Lemma~\ref{lem:easy-growing} follows.
	\end{proof}

	We let $C > 0$ be a large enough constant. Let $U \subset T$ be the set obtained from Lemma~\ref{lem:easy-growing} with $\eta = Ck$, and suppose that $|U| \geq |T|/2$. Then, we may let $\alpha = 1$ and $H = \{ i_c : (i_1, \dots, i_{k_0}) \in U\}$. Notice that every index in $H$ is the start of an $(\alpha, Ck\alpha)$-growing suffix, and since $|H| \geq |T|/2$, we obtain the first item in Theorem~\ref{thm:main-structure}. Suppose then, that $|U| < |T|/2$, and consider the set $V = T\setminus U$. By definition of $V$, we now have $v(\ell, V) \leq 12 Ck \log k$ for every $\ell \in [n]$. Let $b_0$ be the largest integer which satisfies $2^{b_0} \leq 12 C k\log k$ and $b_1$ be the smallest integer which satisfies $2^{-b_1} \leq \delta/(12 k^2)$, so $2^{b_0} \lsim 2^{b_1} \asymp k^2/\delta$. For $-b_0 \leq j \leq b_1$, consider the pairwise-disjoint sets
	\begin{align} \label{eq:def-B}
		B_j &= \left\{ \ell \in [n] : 2^{-j} \leq v(\ell, V) \leq 2^{-j + 1}\right\},
	\end{align}
	and note that by Lemma~\ref{lem:expect}, since $|V| \geq |T| / 2 \ge \delta n / 2k^2$,
	\[ 
		\frac{1}{n} \sum_{j =-b_0}^{b_1} |B_j| \cdot 2^{-j+1} 
		\geq \frac{1}{n} \sum_{\ell \in [n]} v(\ell, V) \geq \frac{\delta}{6k^2}.  
	\]
	Thus, denoting
	\[ 
		\mu \eqdef \frac{\delta}{6 k^2(b_1 + b_0 +1)} \asymp \frac{\delta}{k^2\log(k/\delta)},
	\] 
	there is an integer $-b_0 \leq j^\ast \leq b_1$ that satisfies
	\begin{align} 
		|B_{j^\ast}| \cdot 2^{-j^\ast} &\geq \mu n. \label{eq:B-lb}
	\end{align}

	\newcommand{\GreedyDisjointIntervals}{\texttt{GreedyDisjointIntervals}}

	\begin{lemma}
		There exists a deterministic algorithm, $\emph{\GreedyDisjointIntervals}(f, B, j)$, which takes three inputs: a function $f \colon [n] \to \R$, a set $B \subseteq [n]$ of integers, and an integer $j \in [-b_0, b_1]$, and outputs a collection $\calI$ of interval-tuple pairs  
		or a subset $H \subseteq B$. An execution of the algorithm $\emph{\GreedyDisjointIntervals}(f, B_{j^\ast}, j^\ast)$ where $\mu$, $B_{j^\ast}$ and $j^\ast$ are defined in~\eqref{eq:B-lb}, satisfies one of the following two conditions, where $C > 0$ is a large constant.
			\begin{itemize}
			\item 
				The algorithm returns a set $H \subseteq B$ of indices that start a $(4 \cdot 2^{-j^\ast} /(Ck\log k), 2^{-j^\ast}/(12\log k))$-growing suffix, and $|H| \geq 2^{j^\ast-1} \mu n$; or
			\item 
				The algorithm returns a $(c, 1/(6k),2^{-j^\ast}/(8Ck^2\log k))$-splittable collection $(I_1, T_1), \dots, (I_s, T_s)$, where $\sum_{h = 1}^{s} |I_h| \geq 2^{j^\ast-2} \mu n$.
			\end{itemize}
	\end{lemma}

	\begin{figure}[ht!]
		\begin{framed}

			\centering
			\begin{minipage}{.98\textwidth}
				\noindent Subroutine $\GreedyDisjointIntervals\vspace{0.3cm}(f, B, j)$

				\noindent {\bf Input:} A function $f \colon [n] \to \R$, a set $B \subseteq [n]$ and an integer $j$, such that every $\ell \in B$ satisfies $2^{-j} \leq v(\ell, V) \leq 2^{-j+1}$.\\
				{\bf Output:} a set of disjoint intervals-tuple pairs $(I_1, T_1), \dots, (I_s, T_s)$ or a subset $H \subseteq B$. 
				
				\begin{enumerate}
					\item 
						Let $\calI$ be a collection of interval-tuple pairs, which is initially empty.
					\item 
						Consider the map $q \colon B \to \{ 0, \dots, \lfloor \log n \rfloor\} \cup \{ \bot \}$ defined by
						\begin{align*} 
						q(\ell) = \left\{ 
							\begin{array}{lc} 
								\bot  & \forall t \in \{0, \dots, \lfloor \log n \rfloor \}, \frac{|A_t(\ell,V)|}{2^{t}} <  \frac{2^{-j}}{ Ck\log k} \\
								\max\left\{ t :  \frac{|A_t(\ell, V)|}{2^t} \geq \frac{2^{-j}}{Ck \log k} \right\} & \text{otherwise} 
							\end{array} 
							\right.  .
						\end{align*} 
						\item\label{step:return:H} Let $H = \{ \ell \in B : q(\ell) = \bot\}$, and \Return $H$ if $|H| \geq |B| / 2$. 
					\item Otherwise, let $D \gets B \setminus H$ and repeat the following until $D = \emptyset$:
					\begin{itemize}
						\item\label{step:qt} 
							Pick any $\ell \in D$ where $q(\ell) = \max_{\ell' \in D} q(\ell')$, and let $t = q(\ell)$.
						\item\label{step:IT} 
							Let $I \gets [\ell - k 2^{t+1}, \ell + k 2^{t+1}] \cap [n]$ and $T' \gets A_t(\ell, V)$.
						\item 
							Obtain $T''$ from $T'$ as follows: find a value $\nu$ such that at least $|T'|/2$ of tuples $(i_1,\dots, i_{k_0})\in T'$ satisfy $f(i_c) \leq \nu$, and at least $|T'|/2$ of tuples $(i_1,\dots, i_{k_0})\in T'$ satisfy $f(i_{c+1}) > \nu$ ($\nu$ could be taken to be the median of the multiset $\{f(i_c) : (i_1, \dots, i_{k_0}) \in T'\}$).  Recombine these prefixes and suffixes (matching them in one-to-one correspondence) to obtain a set of disjoint $k_0$-tuples $T''$ of size $|T''|\geq |T'|/2$.
						\item\label{step:update:C} 
							Append $(I, T'')$ to $\calI$, and let $D \gets D \setminus [\ell - 2 \cdot k 2^{t+1}, \ell + 2 \cdot k 2^{t+1}]$.
					\end{itemize} 
				\item\label{step:return:I} \Return $\calI$.
				\end{enumerate}
			\end{minipage}
		\end{framed}\vspace{-0.2cm}
		\caption{Description of the $\GreedyDisjointIntervals$ subroutine.} \label{fig:greedy-intervals}
	\end{figure}

	\begin{proof}
		It is clear that the algorithm always terminates, and outputs either a collection $\calI$ of interval-tuple pairs or a subset $H \subseteq B$. 
		Suppose that the input of the algorithm, $(f, B_{j^\ast}, j^\ast)$, satisfies~\eqref{eq:B-lb}, and consider the two possible types of outputs.\medskip

		If the algorithm returns a set $H\subseteq B_{j^\ast}$ (in step~\ref{step:return:H}), then we have $|H| \geq \frac{|B|}{2} \geq \frac{1}{2}\cdot 2^{j^\ast}\mu n$ (the second inequality by~\eqref{eq:B-lb}). (To see why the elements of $H$ start $(4 \cdot 2^{-j^\ast} /(Ck\log k), 2^{-j^\ast}/(12\log k))$-growing suffixes (Definition~\ref{def:growing-suffixes}), notice that we may apply Lemma~\ref{lem:move-back} with $\ell' = \ell$ and $\beta = 2^{-j^\ast} / (Ck \log k)$.) 
		
		If, instead, the algorithm returns a collection $\calI = ((I_h, T_h) : h \in [s])$ in step~\ref{step:return:I}, we have that, by construction, each $T_h$ is obtained from a set $T'_h=A_{t}(\ell, V)$ for some $\ell$ with $q(\ell)\neq \bot$. Consequently, for all $h \in[s]$ we have 
		\begin{equation} \label{eq:density-T-h}
			\frac{|T_h|}{|I_h|} \geq \frac{|T'_h|}{2|I_h|} 
			\geq \frac{|A_{q(\ell)}(\ell, V)|}{4k \cdot 2^{q(\ell)+1}} 
			\geq \frac{1}{8k} \cdot \frac{2^{-j^\ast}}{Ck\log k}.
		\end{equation}
		(from the definition of $q(\ell)$). To argue that $\sum_{h=1}^s |I_h|$ is large, observe that, since we did not output the set $H$, we must have had $|D|>|B_{j^\ast}|/2$. Since, when adding $(I_h, T_h)$ (corresponding to some $\ell_h$) to $\calI$ we remove at most $4k 2^{q(\ell)+1}=2|I_h|$ elements from $D$, in order to obtain an empty set $D$ and reach step~\ref{step:return:I} we must have $\sum_{h=1}^s |I_h| \geq |B_{j^\ast}|/4$, which is at least $2^{j^\ast}\mu n/4$ by~\eqref{eq:B-lb}. Moreover, the sets $I_h$ are disjoint: this is because of our choice of maximal $q(\ell)$ in step~\ref{step:IT}, which ensures that after removing $[\ell - 2k 2^{q(\ell)+1}, \ell + 2k 2^{q(\ell)+1}]$ in step~\ref{step:update:C} there cannot remain any $\ell'\in D$ with $[\ell' - k 2^{q(\ell')+1}, \ell' + k 2^{q(\ell')+1}]\cap I_h \neq \emptyset$.

		Thus, it remains to prove that $\calI$ is a $(c,1/(6k),2^{-j^\ast}/(8Ck^2\log k))$-splittable collection. To do so, consider any $(I_h, T_h)\in\calI$. 
		The first condition in Definition~\ref{def:splittable} of splittable pairs, namely that $|T_h| / |I_h| \ge 2^{-j^\ast}/(8Ck^2 \log k)$ holds due to \eqref{eq:density-T-h}.
		Recalling step~\ref{step:IT}, we have $I_h = [\ell - k 2^{t+1}, \ell + k 2^{t+1}]$ for some $\ell$, where $t = q(\ell)$, and $T_h$ obtained from $T_h' = A_{t}(\ell, V)$. Set
		\[
			L_h\eqdef [\ell - k 2^{t+1}, \ell - \gamma 2^{t}], \quad M_h\eqdef (\ell - \gamma 2^{t}, \ell + \gamma 2^{t}),\quad R_h\eqdef [\ell + \gamma 2^{t}, \ell + k 2^{t+1}].
		\]
		This is a partition of $I_h$ into three adjacent intervals whose size is at least $|I_h| / (6k)$ (recall that $\gamma = 1/3$). Moreover, for every $(i_1,\dots,i_{k_0})\in T_h'$, the $c$-prefix $(i_1,\dots,i_c)$ is in $(L_h)^c$ while the $(k_0-c)$-suffix $(i_{c+1},\dots,i_{k_0})$ is in $(R_h)^{k_0-c}$, by the first item of Lemma~\ref{lem:At:properties}. Since $T_h$ is obtained from a subset of these very prefixes and suffices, the conclusion holds for $T_h$ as well. Moreover, our construction of $T_h$ from $T'_h$ guarantees that the last requirement in Definition~\ref{def:splittable} holds: for every prefix $(i_1, \dots, i_c)$ of a tuple in $T_h$ and suffix $(j_1, \dots, j_{k_0-c})$ of a tuple in $T_h$, we have $f(i_c) < f(j_1)$. This shows that $(I_h,T_h)$ is $(c,1/(6k),2^{-j^\ast}/(8Ck^2 \log k))$-splittable, and overall that $\calI$ is a $(c,1/(6k),2^{-j^\ast}/(8Ck^2\log k))$-splittable collection as claimed.
	\end{proof}

	Theorem~\ref{thm:main-structure} follows by executing $\GreedyDisjointIntervals(f, B_{j^\ast}, j^\ast)$. 
	If the algorithm outputs a set $H \subseteq B_{j^\ast}$, set $\alpha = 4 \cdot 2^{-j^\ast} / (Ck\log k)$, so we have identified a subset $H$ of $(\alpha, C' \alpha k)$-growing suffixes (where $C' = C / 48$) satisfying $\alpha |H| \geq \delta n / \poly(k, \log(1/\delta)) = |T_0| / \poly(k, \log(1/\delta))$ (using the definition of $\mu$ before \eqref{eq:B-lb}). 
	Otherwise, set $\alpha = 2^{-j^\ast} / (8 Ck^2\log k)$, and the algorithm outputs a $(c, 1/(6k),\alpha)$-splittable collection $\{(I_1, T_1), \dots, (I_s, T_s)\}$ of the set $T' := \cup_{h \in [s]} T_h$. Clearly, $E(T') \subseteq E(T)$, and moreover, $\alpha \sum_{h = 1}^s |I_h| \ge \delta n / \poly(k, \log(1/\delta)) = |T_0| / \poly(k, \log(1/\delta))$. 
	In fact, $2^{-j^*} = \Omega(\delta / k^2)$ and so $\alpha \ge \Omega(\delta / (k^4 \log k))$.

\section{The Algorithm}\label{sec:algorithm}

\subsection{High-level plan}

	\newcommand{\Sampler}{\texttt{Sampler}}

	We now present the algorithm for finding monotone subsequences of length $k$. 

	\begin{theorem}\label{thm:ub}
		Consider any fixed value of $k \in \N$. There exists a non-adaptive and randomized algorithm, $\emph{\Sampler}_k(f, \eps)$, which takes two inputs: query access to a function $f \colon [n] \to \R$ and a parameter $\eps > 0$. If $f$ is $\eps$-far from $(12\dots k)$-free, then $\emph{\Sampler}_k(f, \eps)$ finds a $(12\dots k)$-pattern with probability at least $9/10$. The query complexity of $\emph{\Sampler}_k(f, \eps)$ is at most 
		\[ \frac{1}{\eps} \left(\frac{\log n}{\eps}\right)^{\lfloor \log_2 k \rfloor} \cdot \poly(\log(1/\eps))\,. \]
	\end{theorem}

	The particular dependence on $k$ and $\log(1/\eps)$ obtained from Theorem~\ref{thm:ub} is on the order of $(k \log(1/\eps))^{O(k^2)}$. 
	The algorithm is divided into two cases, corresponding to the two outcomes from an application of Theorem~\ref{thm:tree}. Suppose $f \colon [n] \to \R$ is a function which is $\eps$-far from being $(12\dots k)$-free. By Theorem~\ref{thm:tree} one of the followin holds, where $C > 0$ is a large constant.
	\begin{description}
		\item[Case 1:] 
			there exist $\alpha \ge \eps / \polylog(1/\eps)$ and a set $H \subseteq [n]$ of $(\alpha, Ck\alpha)$-growing suffixes where $\alpha |H| \geq \eps n / \polylog(1/\eps)$, or
		\item[Case 2:] 
			there exist a set $T \subseteq [n]^k$ of disjoint, length-$k$ monotone sequences, that satisfies $|T| \geq \eps n / (\polylog(1/\eps))$, and a $k$-tree descriptor $(G, \varrho, \sfI)$ which represents $(f, T, [n])$. 
	\end{description}
	Theorem~\ref{thm:ub} follows from analyzing the two cases independently, and designing an algorithm for each.

	\newcommand{\SampleSuffix}{\texttt{Sample-Suffix}}
	\begin{lemma}[Case 1]\label{lem:case1}
		Consider any fixed value of $k \in \N$, and let $C > 0$ be a large enough constant. There exists a non-adaptive and randomized algorithm, $\emph{\SampleSuffix}_k(f, \eps)$ which takes two inputs: query access to a function $f \colon [n] \to \R$ and a parameter $\eps > 0$. Suppose there exist $\alpha \in (0,1)$ and a set $H \subseteq [n]$ of $(\alpha, Ck \alpha)$-growing suffixes satisfying $\alpha |H| \geq \eps n / \polylog(1/\eps)$,\footnote{Here we think of $k$ as fixed, so $\polylog(1/\eps)$ is allowed to depend on $k$. In this lemma, the expression stands for $(k \log (1/\eps))^k$.} then $\emph{\SampleSuffix}_k(f, \eps)$ finds a length-$k$ monotone subsequence of $f$ with probability at least $9/10$.
		The query complexity of $\emph{\SampleSuffix}_k(f, \eps)$ is at most
		\[ 
			\frac{\log n}{\eps} \cdot \polylog(1/\eps).
		\]
	\end{lemma}

	Lemma~\ref{lem:case1} above, which corresponds to the first case of Theorem~\ref{thm:tree}, is proved in Section~\ref{sec:case1}.

	\newcommand{\SampleSplittable}{\texttt{Sample-Splittable}}

	\begin{lemma}[Case 2]\label{lem:case2}
		Consider any fixed value of $k \in \N$. There exists a non-adaptive, randomized algorithm, $\emph{\SampleSplittable}_k(f, \eps)$ which takes two inputs: query access to a sequence $f \colon [n] \to \R$ and a parameter $\eps > 0$. Suppose there exists a set $T \subseteq [n]^k$ of disjoint, length-$k$ monotone subsequences of $f$ where $|T| \geq \eps n /  \polylog(1/\eps)$,\footnote{in this case the $\polylog(1/\eps)$ term stands for $(k \log(1/\eps))^{O(k^2)}$} as well as a $(k, k, \alpha)$-tree descriptor $(G, \varrho, \sfI)$ that represents $(f, T, [n])$, where $\alpha \ge \eps / \polylog(1/\eps)$, then $\emph{\SampleSplittable}_k(f, \eps)$ finds a length-$k$ monotone subsequence of $f$ with probability at least $9/10$.  The query complexity of $\emph{\SampleSplittable}_k(f, \eps)$ is at most
		\[ \frac{1}{\eps} \left(\frac{\log n}{\eps} \right)^{\lfloor \log_2 k \rfloor} \cdot \polylog(1/\eps). \]
	\end{lemma}

	\begin{proof}[Proof of Theorem~\ref{thm:ub} assuming Lemmas~\ref{lem:case1} and~\ref{lem:case2}]
		The algorithm $\Sampler_k(f, \eps)$ executes both $\SampleSuffix_k(f, \eps)$ and $\SampleSplittable_k(f, \eps)$; if either algorithm finds a length-$k$ monotone subsequence of $f$, output such a subsequence. We note that by Theorem~\ref{thm:tree}, either case 1, or case 2 holds. If case 1 holds, then by Lemma~\ref{lem:case1}, $\SampleSuffix(f, \eps)$ outputs a length-$k$ monotone subsequence with probability at least $9/10$, and if case 2 holds, then by Lemma~\ref{lem:case2}, $\SampleSplittable_k(f, \eps)$ outputs a length-$k$ monotone subsequence with probability at least $9/10$. Thus, regardless of which case holds, a length-$k$ monotone subsequence will be found with probability at least $9/10$. The query complexity then follows from the maximum of the two query complexities.
	\end{proof}

\subsection{Proof of Lemma~\ref{lem:case1}: an algorithm for growing suffixes}\label{sec:case1}

	We now prove Lemma~\ref{lem:case1}. Let $C > 0$ be a large constant, and let $k \in \N$ be fixed. Let $\eps > 0$ and $f \colon [n] \to \R$ be a function which is $\eps$-far from $(12\dots k)$-free. Furthermore, as per the assumption of case 1 of the algorithm, we assume that there exists a parameter $\alpha \in (0, 1)$ as well as a set $H \subseteq [n]$ of $(\alpha, Ck\alpha)$-growing suffixes, where $\alpha |H| \geq \eps n / \polylog(1/\eps)$. 

	\newcommand{\GrowingSuffix}{\texttt{Growing-Suffix}}

	\begin{figure}[ht!]
		\begin{framed}
			\begin{minipage}{.98\textwidth}
				\noindent Subroutine $\GrowingSuffix\hspace{0.05cm}(f, \alpha_0, a)$
				
				\vspace{0.3cm}
				\noindent {\bf Input:} Query access to a function $f \colon [n] \to \R$, a parameter $\alpha_0 \in (0, 1)$, and an index $a \in [n]$.\\
				{\bf Output:} a subset of $k$ indices $i_1 < \dots < i_k$ where $f(i_1) < \dots <f(i_k)$, or \fail. 
				
				\begin{enumerate}
					\item 
						Let $\eta_a = \lceil \log(n-a) \rceil$ and consider the sets $S_j(a) = (a + \ell_{j-1}, a+\ell_j] \cap [n]$ for all $j \in [\eta_a]$ and $\ell_j = 2^j$.
					\item 
						For each $j \in [\eta_a]$, let $\bA_j \subseteq S_j(a)$ be obtained by sampling uniformly at random $T \eqdef 1 / \alpha_0$ times from $S_j(a)$.
					\item 
						For each $j \in [\eta_a]$ and each $b \in \bA_j$, query $f(b)$ .
					\item 
						If there exist indices $i_1, \dots, i_k \in \bA_1 \cup \dots \cup \bA_{\eta_i}$ satisfying $i_1 < \dots < i_k$ and $f(i_1) < \dots < f(i_k)$, \Return such indices $i_1, \dots, i_k$. Otherwise, \Return \fail. 
				\end{enumerate}
			\end{minipage}
		\end{framed}\vspace{-0.2cm}
		\caption{Description of the $\GrowingSuffix$ subroutine.} \label{alg:growing-suffix}
		\label{fig:growing-suffix-alg}
	\end{figure}

	The algorithm, which underlies the result of Lemma~\ref{lem:case1}, proceeds by sampling uniformly at random an index $\ba \sim [n]$, and running a sub-routine which we call $\GrowingSuffix$,  with $\ba$ as input. The sub-routine is designed so that if $\ba$ is the start of an $(\alpha, Ck\alpha)$-growing suffix then the algorithm will find a length-$k$ monotone subsequence of $f$ with probability at least $99/100$. The sub-routine, $\GrowingSuffix$, is presented in Figure~\ref{fig:growing-suffix-alg}.

	\begin{lemma}\label{lem:growing}
		Let $f \colon [n] \to \R$ be a function, let $\alpha, \alpha_0, \beta \in (0, 1)$ be parameters satisfying $\beta \ge Ck \alpha$ and $\alpha_0 \le \alpha$, and suppose that $a \in [n]$ starts a $(\alpha, \beta)$-growing suffix in $f$. Then $\emph{\GrowingSuffix}(f, \alpha_0, a)$ finds a length-$k$ monotone subsequence of $f$ with probability at least $99/100$.
	\end{lemma}

	\begin{proof}
		Recall, from Definition~\ref{def:growing-suffixes}, that if $a \in [n]$ is the start of a $(\alpha, \beta)$-growing suffix of $f$ then there exist a collection of sets, $D_{1}(a), \dots, D_{\eta_a}(a)$ and parameters $\delta_1(a), \dots, \delta_{\eta_a}(a) \in (0, \alpha]$, where every $j \in [\eta_a]$ has
		\[ 
			D_j(a) \subseteq S_j(a), \qquad 
			|D_j(a)| = \delta_j(a) \cdot |S_j(a)|,\qquad 
			\text{and} \qquad \sum_{j=1}^{\eta_a} \delta_j(a) \geq \beta. 
		\]
		Further, if, for some $j_1, \dots, j_k \in [\eta_i]$, we have $j_1 < \dots < j_k$ and for all $\ell \in [k]$, $\bA_{j_{\ell}} \cap D_{j_{\ell}}(a) \neq \emptyset$, then the union $D_{j_1}(a) \cup \ldots \cup D_{j_k}(a)$ contains a length-$k$ monotone subsequence. In view of this, for each $j \in [\eta_a]$, consider the indicator random variable 
		\[ \bE_j \eqdef \ind\{ \bA_{j} \cap D_j(a) \neq \emptyset \}, \]
		and observe that by the foregoing discussion $\GrowingSuffix(f, \alpha_0, a)$ samples a length-$k$ monotone subsequence of $f$ whenever $\sum_{j=1}^{\eta_a} \bE_j \geq k$. 
		We note that the $\bE_j$'s are independent, and that 
		\[ 
			\Prx[\bE_j = 1] = 1 - \left( 1 - \delta_j(a)\right)^T \geq \min\left\{ \frac{T \cdot \delta_j(a)}{10}, \frac{1}{10} \right\}. 
		\]
		Let $J \subseteq [\eta_a]$ be the set of indices satisfying $T \cdot \delta_j(a) \geq 1$ (recall that $T = 1 / \alpha_0$). Then, if $|J| \geq C k$ we have 
		\[ 
			\Ex\!\left[ \sum_{j=1}^{\eta_a} \bE_j \right] \geq \frac{C k}{10},
		\]
		since every variable $j \in J$ contributes at least $1/10$.
		On the other hand, if $|J| \leq C k / 2$, then, since $\delta_j(a) \leq \alpha$ for every $j$, we have $\sum_{j \in [\eta_a] \setminus J} \delta_j(a) \geq \beta - |J| \cdot \alpha \geq \beta / 2$ (using $\beta \ge Ck \alpha$) so that
		\[ 
			\Ex\!\left[ \sum_{j=1}^{\eta_a} \bE_j\right] \geq \Ex\!\left[ \sum_{j\in [\eta_a] \setminus J} \bE_j \right] \geq \frac{T}{10} \cdot \frac{\beta}{2} \geq \frac{Ck}{20}. 
		\]
		In either case, $\Ex[\sum_{j\in[\eta_a]} \bE_j] \geq C k / 20$, and since the events $\bE_i$ are independent, via a Chernoff bound we obtain that $\sum_j \bE_j$ is larger than $k$ with probability at least $99/100$.
	\end{proof}

	\begin{figure}[ht!]
		\begin{framed}
			\begin{minipage}{.98\textwidth}
				\noindent Subroutine $\SampleSuffix_k\hspace{0.05cm}(f, \eps)$

				\vspace{.3cm}
				
				\noindent {\bf Input:} Query access to a function $f \colon [n] \to \R$, and a parameter $\eps \in (0, 1)$.\\
				{\bf Output:} a subset of $k$ indices $i_1 < \dots < i_k$ where $f(i_1) < \dots <f(i_k)$, or \fail. 
				
				\begin{enumerate}
				\item Repeat the following for all $j = 1, \dots, O(\log(1/\eps))$, letting $\alpha_j =  2^{-j}$:
				\begin{itemize}
					\item For $t_j = \alpha_j \cdot \polylog(1/\eps) /\eps$ iterations, sample $\ba \sim [n]$ uniformly at random and run $\GrowingSuffix(f, \alpha_j, \ba)$, and if it returns a length-$k$ monotone subsequence of $f$, \Return that subsequence. 
				\end{itemize}
				\item If the algorithm has not already output a monotone subsequence, \Return \fail.
				\end{enumerate}
			\end{minipage}
		\end{framed}\vspace{-0.2cm}
		\caption{Description of the $\SampleSuffix$ subroutine.} \label{fig:sample-suffix}
	\end{figure}

	\noindent With this in hand, we can now establish~Lemma~\ref{lem:case1}.
	\begin{proof}[Proof of Lemma~\ref{lem:case1}]
		First, note that the query complexity of $\SampleSuffix_k(f, \eps)$ is
		\[ 	
			\sum_{j=1}^{O(\log(1/\eps))} t_j \cdot O( \log n / \alpha_j) = \frac{\log n \cdot \polylog(1/\eps)}{\eps}. 
		\]
		Consider the iteration of $j$ where $\alpha_j \leq \alpha \leq 2 \alpha_j$ (note that since $\alpha \ge \eps / \polylog(1 / \eps)$, there exists such $j$). Then, since $|H| \ge \eps / (\alpha \cdot \polylog(1/\eps))$, we have that $t_j \ge Cn / |H|$ (for a sufficiently large constant $C$). Thus, with probability at least $99/100$, some iteration satisfies $\ba \in H$. When this occurs, $\GrowingSuffix(f, \alpha_j, \ba)$ will output a length-$k$ monotone subsequence with probability at least $99/100$, by Lemma~\ref{lem:growing}, and thus by a union bound we obtain the desired result.
	\end{proof}

\subsection{Proof of Lemma~\ref{lem:case2}: an algorithm for splittable intervals}\label{sec:case2}

	\newcommand{\SampleH}{\texttt{Sample-Helper}}

	We now prove Lemma~\ref{lem:case2}. We consider a fixed setting of $k \in \N$ and $\eps > 0$, and let $f \colon [n] \to \R$ be any sequence which is $\eps$-far from being $(12\dots k)$-free. Furthermore, as per case 2 of the algorithm, we assume that there exists a set $T \subseteq [n]^k$ of disjoint length-$k$ monotone subsequences of $f$ where
	\[ 
		|T| \geq \frac{\eps n}{\polylog(1/\eps)},
	\]
	and $(G, \varrho, \sfI)$ is a $(k, k, \alpha)$-tree descriptor which represents $(f, T, [n])$, where $\alpha \ge \eps / \polylog(1/\eps)$. In what follows, we describe a sub-routine, $\SampleSplittable_k(f, \eps)$ in terms of two parameters $\rho, q \in \R$. The parameter $\rho > 0$ is set to be sufficiently large and independent of $n$, satisfying
	\begin{equation}\label{eq:case2:setting:rho}
		\rho \geq \frac{\eps}{\polylog(1/\eps)}.
	\end{equation}
	One property which we will want to satisfy is that if we take a random subset of $[n]$ by including each element independently with probability $1/(\rho n)$, we will include an element belonging to $E(T)$ with probability at least $1 - 1/(C k)$, for a large constant $C > 0$.
	The parameter $q$ will be an upper bound on the query complexity of the algorithm, which we set to a high enough value satisfying:
	\[ q = O\left(\frac{1}{\rho} \left( \frac{\log n}{\rho}\right)^{\lfloor \log_2 k \rfloor}\right) \leq \frac{1}{\eps} \cdot \left( \frac{\log n}{\eps} \right)^{\lfloor \log_2 k \rfloor} \cdot \polylog(1/\eps).  \]

	\begin{figure}[ht!]
		\begin{framed}
			\begin{minipage}{.98\textwidth}
				\noindent Subroutine $\SampleSplittable_k\hspace{0.05cm}(f, \eps)$
				\vspace{.3cm}
				
				\noindent {\bf Input:} Query access to a sequence $f \colon [n] \to \R$, and a parameter $\eps \in (0, 1)$.\\
				{\bf Output:} a subset of $k$ indices $i_1 < \dots < i_k$ where $f(i_1) < \dots <f(i_k)$, or \fail. 
				
				\begin{enumerate}
				\item Let $r = \lfloor \log_2 k \rfloor$ and run $\SampleH(r,[n],\rho)$, to obtain a set $\bA \subseteq [n]$.
				\item If $|\bA| > q$, \Return \fail; otherwise, for each $a \in \bA$, query $f(a)$. If there exists a monotone sequence of $f$ of length $k$, then \Return that subsequence. If not, \Return \fail.  
				\end{enumerate}
			\end{minipage}
		\end{framed}\vspace{-0.2cm}
		\caption{Description of the $\SampleSplittable$ subroutine.} \label{fig:sample-splittable}
	\end{figure}

	\begin{figure}[ht!]
		\begin{framed}
			\begin{minipage}{.98\textwidth}
				\noindent Subroutine $\SampleH\hspace{0.05cm}(r,I, \rho)$
				\vspace{.3cm}

				\noindent {\bf Input:} An integer $r \in \N$, an interval $I \subseteq [n]$, and a parameter $\rho \in (0, 1)$.\\
				{\bf Output:} a subset of $A \subseteq I$. 
				
				\begin{enumerate}
					\item Let $\bA_0 = \emptyset$. For every index $a \in I$, let $\bA_0 \gets \bA_0 \cup \{ a \}$ with probability $1 / (\rho |I|)$.
					\item If $r = 0$, \Return $\bA_0$.
					\item If $r > 0$, proceed with the following:
					\begin{itemize}
						\item 
							For every index $a \in \bA_0$, consider the $O(\log n)$ intervals given by $B_{a, j} = [a - \ell_j, a + \ell_j]$, for $j = 1, \dots,O(\log n)$ and $\ell_j = 2^j$, and let $\bR_{a, j} \gets \SampleH(r-1, B_{a, j}, \rho)$.
						\item 
							Let $\bA$ be the set
							\[ 
								\bA \gets \bigcup_{\substack{a \in \bA_0, \, j = O(\log n)}} \bR_{a,j}. 
							\]
						\item \Return the set $(\bA_0 \cup \bA) \cap I$.
					\end{itemize}
				\end{enumerate}
			\end{minipage}
		\end{framed}\vspace{-0.2cm}
		\caption{Description of the $\SampleH$ subroutine.} \label{fig:sample-splittable-2}
	\end{figure}

	The descriptions of the main algorithm $\SampleSplittable_k$ and the sub-routine $\SampleH$, are given in Figure~\ref{fig:sample-splittable} and Figure~\ref{fig:sample-splittable-2}. Note that, for any $r \in \N$, if we let $\calD_r$ be the distribution of $|\bA|$, where $\bA$ is the output of a call to $\SampleH(r, [n], \rho)$. Then, we have that $\calD_0 = \Bin(n, \rho)$, and for $r > 0$, $\calD_r$ is stochastically dominated by the random variable
	\[ \sum_{i=1}^{\by_0} \sum_{j=1}^{O(\log n)} \bx_{r-1}^{(i, j)}, \]
	where $\by_0 \sim \Bin(n, 1/(\rho n))$ and $\bx_{r-1}^{(i,j)} \sim \calD_{r-1}$ for all $i \in \N$ and $j \in [O(\log n)]$ are all mutually independent. As a result, for $r \geq 1$,
	\[ \Ex\left[ |\bA| \right] \leq \frac{1}{\rho} \cdot \log n \cdot \Ex_{\bx \sim \calD_{r-1}}[\bx], \]
	and since $\Ex_{\bx \sim\calD_0}[\bx] = 1/\rho$, we have:
	\[ \Ex\left[ |\bA| \right] \leq \frac{1}{\rho} \left( \frac{\log n}{\rho} \right)^{r}. \]
	We may then apply Markov's inequality to conclude that $|\bA| \leq q$ with probability at least $99/100$. As a result, we focus on proving that the probability that the set $\bA$ contains a monotone subsequence of $f$ of length $k$ is at least $99/100$. This would imply the desired result by taking a union bound. 

	In addition to the above, we define another algorithm, $\SampleH^*$, in Figure~\ref{fig:sample-help}, which will be a \emph{helper} sub-routine. We emphasize that $\SampleH^*$ is not executed in the algorithm itself, but will be useful in order to analyze $\SampleH$.  

	\begin{figure}[H]
		\begin{framed}
			\begin{minipage}{.98\textwidth}
				\noindent Subroutine $\SampleH^*\hspace{0.05cm}(r,I, \rho, \calI)$
				\vspace{.3cm}

				\noindent {\bf Input:} An integer $r \in \N$, an interval $I \subseteq [n]$, a parameter $\rho \in (0, 1)$, and a collection of disjoint intervals $\calI$ of $[n]$.\\
				{\bf Output:} two subsets $\bA, \bA_0 \subseteq I$. 
				
				\begin{enumerate}
					\item\label{sample:helper*:step:1} Let $\bA_0 = \emptyset$. For every index $a \in I$ which lies inside an interval in $\calI$, let $\bA_0 \gets \bA_0 \cup \{ a \}$ with probability $1 / (\rho |I|)$.
					\item\label{sample:helper*:step:2} If $r = 0$, \Return $\bA_0$.
					\item\label{sample:helper*:step:3} If $r > 0$, proceed with the following:
					\begin{itemize}
						\item 
							For every index $a \in \bA_0$, consider the $O(\log n)$ intervals given by $B_{a, j} = [a - \ell_j, a + \ell_j]$, for $j = 1, \dots O(\log n)$, and $\ell_j = 2^j$, and let $(\bR_{a, j}, \bR_{a, j, 0}) \gets \SampleH^*(r-1, B_{a, j}, \rho, \calI)$.
						\item 
							Let $\bA$ to be the set
							\[ 
								\bA \gets \bigcup_{\substack{a \in \bA_0, \,\, j = O(\log n)}} \bR_{a,j}. 
							\]
						\item 
							\Return the set $(\bA \cap I, \bA_0 \cap I)$.
					\end{itemize}
				\end{enumerate}
			\end{minipage}
		\end{framed}\vspace{-0.2cm}
		\caption{Description of the $\SampleH^*$ subroutine.} \label{fig:sample-help}
	\end{figure}

	Before proceeding, we require a ``coupling lemma.'' Its main purpose is to prove the intuitive fact that if $\calI_0, \calI_1$ are collections of disjoint intervals, and the latter is a refinement of the former (namely, each intervals in $\calI_1$ is contained in an interval of $\calI_0$), then $\SampleH^*(r, [n], \rho, \calI_0)$ is more likely to find a length-$k$ monotone subsequence than $\SampleH^*(r, [n], \rho, \calI_1)$ does.

	\begin{lemma}\label{lem:coupling}
		Let $r \in \N$ be an integer, $f \colon [n] \to \R$ a function,  $\rho \in (0,1)$ a parameter, and $\calI_0$ and $\calI_1$ collections of disjoint intervals in $[n]$, such that each interval in $\calI_1$ lies inside an interval from $\calI_0$. 
		Denote by $(\bA^{(i)}, \bA_0^{(i)})$ the random pair of sets given by the output of $\emph{\SampleH}^*(r, [n], \rho, \calI_i)$, for $i = 0,1$. 
		Lastly, let $\emph{\Event} \colon \calP([n]) \times \calP([n]) \to \{0,1\}$ be any \emph{monotone} function; that is, it satisfies $\emph{\Event}(S_1, S_2) \leq \emph{\Event}(S_1', S_2')$ for any $S_1 \subseteq S_1' \subseteq [n]$ and $S_2 \subseteq S_2' \subseteq [n]$. Then,
		\[ \Prx[ \emph{\Event}(\bA^{(0)}, \bA^{(0)}_0) = 1] \geq \Prx[ \emph{\Event}(\bA^{(1)}, \bA^{(1)}_0) = 1]. \]
	\end{lemma}

	\begin{proof}
		Consider an execution of $\SampleH^*(r, [n], \rho, \calI_0)$ which outputs a pair $(\bA^{(0)}, \bA_0^{(0)})$. Let $\bA^{(1)}$ and $\bA^{(1)}$ be the subsets of $\bA^{(0)}$ and $\bA^{(0)}$, respectively, obtained by running a parallel execution of $\SampleH^*(r, [n], \rho, \calI_1)$, which follows the execution of $\SampleH^*(r, [n], \rho, \calI_0)$, but whenever an element which is not in an interval of $\calI_1$ is considered, it is simply ignored (i.e., it is not included in $\bA^{(0)}$ or in $\bA_0^{(0)}$ and no recursive calls based on such elements are made). It is easy to see that this coupling yields a pair $(\bA^{(1)}, \bA_0^{(1)})$ with the same distribution as that given by running $\SampleH^*(r, [n], \rho, \calI_1)$.
		As $\emph{\Event}(\cdot, \cdot)$ is increasing, if $\emph{\Event}(\bA^{(0)}, \bA^{(0)}_0)$ holds then so does $\emph{\Event}(\bA^{(1)}, \bA^{(1)}_0)$. The lemma follows.
	 \end{proof}

	 The following corollary is a direct consequence of Lemma~\ref{lem:coupling}. Specifically, we use the facts that $\SampleSplittable_k(f, \eps)$ calls $\SampleH(\lfloor \log_2 k \rfloor, [n], \rho)$, which is equivalent to calling $\SampleH(\lfloor \log_2 k \rfloor, [n], \rho, \{ [n] \})$, and that finding a $(12\dots k)$-pattern in $\calI$ is a monotone event. 

	\begin{corollary}\label{cor:helper}
		Let $\calI$ be any collection of disjoint intervals in $[n]$. Suppose $(\bA, \bA_0)$ is the random pair of sets given by the output of $\emph{\SampleH}^*(\lfloor \log_2 k \rfloor, n, \rho, \calI)$, then,
		\begin{align*} 
			&\Prx[\emph{\SampleSplittable}_k(f, \eps) \text{ finds a $(12\dots k)$-pattern of $f$}] \geq \\
			&\qquad\qquad\qquad\qquad\qquad\qquad\Prx[\bA \text{ contains a $(12\dots k)$-pattern in $f_{|\calI}$} ].
		\end{align*}
	\end{corollary}

	\begin{definition}
		Let $k_0 \in \N$ be a positive integer, and let $(G, \varrho)$ be a $k_0$-tree descriptor (for this definition we do not care about the third component of the descriptor, $\sfI$). We say that $p \in [k_0]$ is the \emph{primary index} of $(G, \varrho)$ if the leaf with rank $p$ under $\leq_{G}$ is the unique leaf whose root-to-leaf path $(u_1, \dots, u_d)$ satisfies the following: for each $d' \in [d-1]$, denoting the left  and right children of $u_{d'}$ by $v_l$ and $v_r$, respectively, $u_{d' + 1}$ is $v_l$ if the number of leaves in the subtree rooted at $v_l$ is at least the number of leaves in the subtree rooted at $v_r$, and otherwise, $u_{d' + 1}$ is $v_r$.
	\end{definition}

	From Corollary~\ref{cor:helper}, we note that Lemma~\ref{lem:case2} follows from the following lemma.
	\begin{lemma} \label{lem:case2-alg}
		Let $k, k_0, n \in \N$ satisfy $1 \leq k_0 \leq k$, let $C$ be a large enough constant, and let $\alpha, \rho \in (0,1)$ be such that $\rho \ge C \alpha$ and $\alpha \ge \rho / \polylog(1/\rho)$.
		Let $f \colon [n] \to \R$ be a function, let $\calI$ be a collection of disjoint intervals in $[n]$, for each $I \in \calI$ let $T_I \subseteq I^{k_0}$ be a set of disjoint, length-$k_0$ monotone subsequence of $f$, and suppose that
		\[ 
			\sum_{I \in \calI} |T_I| \ge \alpha n / 4.
		\] 
		Suppose that $(G, \varrho)$ is a $(k, k_0, \alpha)$-weighted-tree such that for every $I \in \calI$ there exists a function $\sfI_I :V(G) \to \calS(I)$, such that $(G, \varrho, \sfI_I)$ is a tree descriptor that represents $(f, T_I, I)$. 
		Given any $r \in \N$ satisfying $\lfloor \log_2 k_0 \rfloor \leq r$, let $(\bA, \bA_0)$ be the pair of sets output by the sub-routine $\emph{\SampleH}^*(r, [n], \rho, \calI)$.
		With probability at least $1 - k_0 / (100 k)$, there exist indices $i_1, \dots, i_{k_0} \in [n]$ with the following properties. 
		\begin{enumerate} 
			\item
				$(i_1, \ldots, i_{k_0})$ is a length-$k_0$ monotone subsequence of $f$.
			\item
				There is an interval $I \in \calI$ such that $i_1, \ldots, i_{k_0} \in I \cap E(T_I)$.
			\item
				$i_1, \ldots, i_{k_0} \in \bA$ and $i_p \in \bA_0$, where $p$ is the primary index of $(G, \varrho)$.
		\end{enumerate} 
	\end{lemma}

	\begin{proof}	

		The proof proceeds by induction on $k_0$. Consider the base case, when $k_0 = 1$. In this case, $\lfloor \log_2 k_0 \rfloor = 0$, so for any $r \geq 0$, $\SampleH^*(r, [n],\rho, \calI)$ runs step~\ref{sample:helper*:step:1}. As a result, $\SampleH^*(r, [n], \rho,\calI)$ samples each element inside an interval of $\calI$ independently with probability $1/(\rho n)$. In order to satisfy the requirements of the lemma in this case, we need $\bA_0$ to contain an element of $\cup_{I \in \calI} T_I$. By the assumption on the size of this union, and because each of the elements of the union lives inside some interval from $\calI$, such an element will exist with sufficiently high probability via a Chernoff bound.

		For the inductive step, assume that Lemma~\ref{lem:case2-alg} is fulfilled whenever $k_0 < K$, for $K \in \N$ satisfying $1 < K \leq k$, and we will prove, assuming this inductive hypothesis, that Lemma~\ref{lem:case2-alg} holds for $k_0 = K$. So consider a setting $k_0 = K$. Let $\calI$, $(G, \varrho)$ and $\sfI_I$ be as in the statement of the lemma. Denote the root of $(G, \varrho)$ by $v_{\sf root}$, and its left and right children by $v_{\sf left}$ and $v_{\sf right}$. Let $c$ be the number of leaves in the subtree $(G_{\sf left}, \varrho_{\sf left})$ rooted at $v_{\sf left}$, so $k_0 - c$ is the number of leaves in the subtree $(G_{\sf right}, \varrho_{\sf right})$ rooted at $v_{\sf right}$. We shall assume that $c \ge k_0 - c$; the other case follows by an analogous argument.

		For each $I \in \calI$, the collection of pairs $(J, T_{I,J})$, where $J \in \sfI_I(v_{\sf root})$ and $T_{J} = T_I \cap J^{k_0}$ is the restriction of $T_I$ to $J$, is a $(c, 1/(6k), \alpha)$-splittable collection of $I$. Let $\calJ$ be the collection of all such intervals $J$ (note that they are pairwise disjoint and that $\calJ$ is a refinement of $\calI$).
		Let $(L_J, M_J, R_J)$ be the partition of $J$ into left, middle and right intervals, respectively, and let $T_J^{(L)}$ and $T_J^{(R)}$ be sets of $c$-prefixes and $(k_0 - c)$-suffixes of $k_0$-tuples from $T_{I, J}$, as given by Definition~\ref{def:splittable}. Set
		\[
			\calL = \{L_J : J \in \calJ\}, \qquad \calR = \{R_J : J \in \calJ\}, \qquad T^{(L)} = \bigcup_{J \in \calJ} T_J^{(L)}, \qquad T^{(R)} = \bigcup_{J \in \calJ} T_J^{(R)}.
		\]
		Note that $(G_{\sf left}, \varrho_{\sf left}, \sfI_{J, {\sf left}})$ is a $(k, c, \alpha)$-tree descriptor for $(f, T_{J}, J)$, with appropriate $\sfI_{J, {\sf left}}$. Similarly, $(G_{\sf right}, \varrho_{\sf right}, \sfI_{J, {\sf right}})$ is a $(k, k_0 - c, \alpha)$-tree descriptor for $(f, T_{J}, J)$, with appropriate $\sfI_{J, {\sf right}}$.

		\newcommand{\sfa}{ {\sf a}}

		We consider an execution of $\SampleH^*(r, [n], \rho, \calI)$ which outputs a random pair of sets $(\bA, \bA_0)$. Let $\bA^{(L)}$ and $\bA_0^{(L)}$ be the subsets of $\bA$ and $\bA_0$, respectively, obtained by running a parallel execution of $\SampleH^*(r, [n], \rho, \calL)$, where, as in the proof of Lemma~\ref{lem:coupling}, we follow the execution of $\SampleH^*(r, [n], \rho, \calI)$, but whenever an element which is not in $\calL$ is considered, we ignore it. As stated above, this coupling yields a pair $(\bA^{(L)}, \bA_0^{(L)})$ with the distribution given by running $\SampleH^*(r, [n], \rho, \calL)$.

		For $a \in \bA_0^{(L)}$, and any $j \in [O(\log n)]$, let $(\bA^{(a,j)}, \bA_0^{(a,j)})$ be the output of the recursive call (inside the execution of $\SampleH^*(r, [n], \rho, \calI)$) of $\SampleH^*(r-1, B_{a, j}, \rho, \calR)$. 
		
		We define the collection:
		\[ 
			\calS = \left\{ (S_0, S) :  \begin{array}{l} 
				S_0 \subseteq S \subseteq E(T^{(L)}) \\
				\text{there exist $i_1, \dots, i_c \in S$ forming a $(12\dots c)$-pattern such that $i_p \in S_0$} \\
				\text{there exist $J \in \calJ$ such that $i_1, \dots, i_c \in L_J$} 
			\end{array} \right\}. 
		\] 
		For each $(S_0, S) \in \calS$, we let $\sfa(S_0, S) \in E(T^{(L)})$ be some $i_{p} \in S$ such that there exist $c-1$ indices $i_1, \dots, i_{p-1}, i_{p+1}, i_{c}$, such that $(i_1, \ldots, i_c)$ forms a $(12\dots c)$-pattern in $S$, and $i_1, \ldots, i_p \in L_J$ for some $J \in \calJ$. Let $\seg(S_0, S)$ be this interval $J$, and let $\len(S_0, S) \in [O(\log n)]$ be the smallest $j$ for which $R_J \subseteq B_{a,j}$, where $a = \sfa(S_0, S)$. 
		
		Let $\bE_L$ be the event that 
		\[
			\left(\bA^{(L)} \cap E(T^{(L)}), \bA_0^{(L)} \cap E(T^{(L)}) \right) \in \calS,
		\]
		and let $\bE_L(S_0, S)$ be the event that
		\[
			\bA^{(L)} \cap E(T^{(L)}) = S_0 \qquad  \qquad \bA_0^{(L)} \cap E(T^{(L)}) = S,
		\]
		so $\bE_L = \cup_{(S_0, S) \in \calS} \bE_L(S_0, S)$, and the events $\bE_L(S_0, S)$ are pairwise disjoint.

		By the induction hypothesis, applied with the family $\{L_J : J \in \calJ\}$ and the corresponding sets $T_J^{(L)}$ (using $\sum_{J \in \calJ} |T_J^{(L)}| = \sum_{J \in \calJ} |T_J| \ge \alpha n / 4$), we have
		\[
			\Pr[\bE_L] \ge 1 - c/(100k).
		\]
		Let $\bE_R(a, j)$ be the event that $a \in \bA_0$, and in the recursive run of $\SampleH^*(r-1, B_{a, j}, \rho, \calR)$ inside $\SampleH^*(r, [n], \rho, \calI)$, there exist indices $i_1', \ldots, i_{k_0-c}'$ such that
		\begin{itemize}
			\item
				$(i_1', \ldots, i_{k_0-c}')$ form a length $(k_0-c)$-monotone subsequence.
			\item
				$i_1', \ldots, i_{k_0-c}' \in E(T_J^{(R)})$, where $J$ is the interval in $\calJ$ with $i \in J$.
			\item
				$i_1', \ldots, i_{k_0-c}' \in \bA^{(a,j)}$ and $i_q' \in \bA_0^{(a,j)}$, where $q$ is the primary index of $(G_{\sf right}, \varrho_{\sf right})$.
		\end{itemize}
		Let $\bF_R(a, j)$ be the event that in a run of $\SampleH^*(r-1, B_{a, j}, \rho, \calR)$, there exist $i_1', \ldots, i_{k_0-c}'$ as above.
		Fix some $(S_0, S) \in \calS$, and let $a = \sfa(S_0, S)$, $J = \seg(S_0, S)$ and $j = \len(S_0, S)$. We claim that 
		\[
			\Pr[\bE_R(a, j) \,\,|\,\, \bE_L(S_0, S)] = \Pr[\bF_R(a, j)].
		\]
		Indeed, by conditioning on $\bE_L(S_0, S)$ we know that $a \in \bA_0$, so there will be a recursive run of $\SampleH^*(r-1, B_{a, j}, \rho, \calR)$, and moreover the event $\bE_L(S_0, S)$ will have no influence on the outcomes of this run.
		
		Note that $|T_J^{(R)}|  \ge \alpha |R_J| \ge \alpha |B_{a, J}| / 4$.
		By the induction hypothesis, applied with the interval $B_{a, J}$ in place of $[n]$, the family $\{R_J\}$ and the corresponding set $T_J^{(R)}$, and the tree $(G_{\sf right}, \varrho_{\sf right})$, we find that  $\Pr[\bF_R(a, j)] \ge 1 - (k_0-c)/(100k)$.
		We note that if both $\bE_L(S_0, S)$ and $\bE_R(a, j)$ hold, then there are indices $i_1, \ldots, i_c, i_1', \ldots, i_{k_0 - c}'$ such that
		\begin{itemize}
			\item
				$(i_1, \ldots, i_c)$ is a length-$c$ monotone subsequence in $E(T_J^{(L)})$, and $(i_1', \ldots, i_{k_0 - c}')$ is a length-$c$ monotone subsequence in $E(T_J^{(R)})$. In particular, $(i_1, \ldots, i_c, i_1', \ldots, i_{k_0 - c}')$ is a length-$k_0$ monotone subsequence that lies in $E(T_j)$.
			\item
				$i_1, \ldots, i_c, i_1', \ldots, i_{k_0 - c}' \in \bA$ and $i_p \in \bA_0$ (recall that $p$ is the primary index of both $G$ and $G_{\sf left}$).
		\end{itemize}
		I.e.\ if these two events hold, then the requirements ot the lemma are satisfied.
		It follows that the requirements ot the lemma are satisfied with at least the following probability, using the fact that the events $\bE_L(S_0, S)$ are disjoint.
		\begin{align*}
			& \sum_{(S_0, S) \in \calS} \Pr[\bE_R(\sfa(S_0,S), \len(S_0,S)) \text{ and } \bE_L(S_0, S)] \\
			= & \sum_{(S_0, S) \in \calS} \Pr[\bE_R(\sfa(S_0,S), \len(S_0,S)) \,\, | \,\, \bE_L(S_0, S)] \times \Pr[\bE_L(S_0, S)] \\
			\ge & \sum_{(S_0, S) \in \calS} \Pr[\bF_R(\sfa(S_0,S), \len(S_0,S))] \times \Pr[\bE_L(S_0, S) ]\\
			\ge & \left(1 - \frac{k_0-c}{100k}\right) \cdot \sum_{(S_0, S) \in \calS} \Pr[\bE_L(S_0, S)] \\
			\ge & \left(1 - \frac{k_0-c}{100k}\right) \cdot \Pr[\bE_L] \\
			\ge & \left(1 - \frac{k_0-c}{100k}\right) \cdot \left(1 - \frac{c}{100k}\right) \ge 1 - \frac{k_0}{100k}.
		\end{align*}
		This completes the proof of Lemma~\ref{lem:case2-alg}.
	\end{proof}
 
\section{Lower Bounds}\label{sec:lowerbounds}

In this section, we prove our lower bound for non-adaptive testing of $(12\dots k)$-freeness with one-sided error, Theorem~\ref{thm:intro-lb}. Below we give a precise quantitative version of our lower bound statement for the case where $k$ and $n$ are both a power of $2$, from which one can derive the general case, as we shall explain soon. 

\begin{theorem}
	\label{thm:lower_bound}
	Let $k \leq n \in \N$ be powers of $2$ and let $0 < p < 1$. There exists a constant $\eps_0 > 0$ such that any non-adaptive algorithm which, given query access to a function $f\colon [n] \to \R$ that is $\eps_0$-far from $(12\dots k)$-free, outputs a length-$k$ monotone subsequence with probability at least $p$, must make at least $ p \binom{\log_2 n}{ \log_2 k }$ queries. Moreover, one can take $\eps_0 = 1/k$.
\end{theorem}

As is usual for arguments of this type, to prove Theorem~\ref{thm:lower_bound} we follow Yao's minimax principle~\cite{Y77}. We construct a distribution $\calD_{n,k}$ over sequences that are $(1/k)$-far from $(12\ldots k)$-free, such that any deterministic algorithm, that makes fewer than $p \binom{\log_2 n}{ \log_2 k }$ queries, fails to find a $(12\ldots k)$-copy in a sequence drawn from this distribution, with probability larger than $1-p$. Here, a deterministic non-adaptive algorithm that makes $q$ queries amounts to deterministically picking a $q$-element subset $Q$ of $[n]$ in advance (without seeing any values in the sequence), and querying all elements of $Q$.

\paragraph{Handling general $k$ and $n$.}
	We first explain how to prove our general lower bound, Theorem~\ref{thm:intro-lb}, using the lower bound distribution $\calD_{n,k}$ of the case where $n$ and $k$ are powers of $2$, given in Theorem~\ref{thm:lower_bound}, as a black box. The reduction relies on standard ``padding'' techniques. 
	Given integers $k, n$ with $k \le n$, write $k = 2^h + t$ for $h,t \in \N$ with $t < 2^h$, and let $k' = 2^h$. Let $n'$ be the largest power of $2$ which is not larger than $nk' / k$, and note that $n' \geq n/4$ and $k' \leq n'$.
	We construct our lower bound distribution $\calD_{n,k}$ as follows. Given any $f' \colon [n'] \to \R$ in $\calD_{n', k'}$, we partition the set $\{n'+1, n'+2, \ldots, n\}$ into $t$ consecutive intervals $I_1, \ldots, I_t$, each of size at least $n' / k'$, and extend $f'$ to a sequence $f \colon [n] \to \R$ satisfying the following conditions.
	\begin{itemize}
		\item $f(x) = f'(x)$ for any $x \in [n']$.
		\item $f$ is decreasing within any $I_i$, that is, $f(x) > f(y)$ for $x < y \in I_i$. 
		\item $f(x) < f(y)$ for any $x \in [n']$ and $y \in I_1$, and for any $x \in I_i$ and $y \in I_j$ where $i < j$.
	\end{itemize}
	Clearly, we can construct such a sequence $f$ from any given sequence $f'$. Moreover, it is possible to make sure that the values $f(x)$ with $x \in [n]$ are distinct, and thus by relabeling $f$ can be taken to be a permutation.
	Furthermore, any $(12\dots k')$-copy in $f'$ can be extended to a $(12\dots k)$-copy in $f$ by appending exactly one arbitrary element from each $I_i$ to it, for a total of $t = k - k'$ additional elements.

	Building on the fact that $f'$ is $(1/k')$-far from $(12\dots k')$-free and that $n' \geq n/4$ and $n-n' \ge n(k-k')/k$, we conclude that $f$ is $(1/4k')$-far from $(12\dots k)$-free. 
	Form a distribution $\calD_{n,k}$ by picking a random $\boldf$ according to the distribution $\sim \calD_{n',k'}$ and extending it to a sequence $f'$ as above.
 
	\noindent The rest of this section is devoted to the proof of Theorem~\ref{thm:lower_bound}.

\subsection{Basic binary profiles and monotonicity testing}
\label{subsec:lower_bound_monotonicity}

	In a sense, the proof of our lower bound, Theorem~\ref{thm:lower_bound}, is a (substantial) generalization of the non-adaptive lower bound for testing monotonicity. In order to introduce the machinery required for the proof, we present, in this subsection, a simple proof of the classical $\Omega(\log n)$ non-adaptive one-sided lower bound for monotonicity testing \cite{EKKRV00} using basic versions of the tools we shall use for the full proof. Then, in Subsection~\ref{subsec:lower_bound_full_proof} we proceed to present our tools in their full generality, and provide the proof of Theorem~\ref{thm:lower_bound}. 

	Intuitively, one way to explain why monotonicity testing requires $\Omega(\log n)$ queries relies on the following reasoning. There exist $\Omega(\log n)$ different distance ``profiles'' our queries should capture; and it can be shown that in general, a small set of queries cannot capture many types of different profiles all at once. At a high level, our new lower bound is an extension of this argument, which uses a more general type of profiles. We start, then, with a formal definition of the basic profiles required for the case of monotonicity testing. Below we restate the required definitions related to the binary representation of numbers in $[n]$.

	\begin{definition}[Binary representation]
		For any $n \in \N$ which is a power of $2$ and $t \in [n]$, the \emph{binary representation} $B_n(t)$ of $t$ is the unique tuple $(b^t_1, b^t_2, \ldots, b^t_{\log_2 n}) \in \{0,1\}^{\log_2 n}$ satisfying $t = b^t_1 \cdot 2^{0} + b^t_2 \cdot 2^{1} + \dots + b^t_{\log_2 n} \cdot 2^{\log_2 n - 1}$. 
		For $i \in [\log_2 n]$, the \emph{bit-flip operator}, $F_i \colon [n] \to [n]$, is defined as follows. Given $t \in [n]$ with $B_n(t) = (b^t_1, \ldots, b^t_{\log n})$, we set $F_n(t) = t'$ where $t' \in [n]$ is the unique integer satisfying $B_n(t') = (b^t_1, \ldots, b^{t}_{i-1}, 1-b^t_i, b^t_{i+1}, \ldots, b^t_{\log n})$.  
		Finally, for any two distinct elements $x,y \in [n]$, let $M(x,y) \in [\log_2 n]$ denote the index of the most significant bit in which they differ, i.e., the largest $i$ with $b_i^x \neq b_i^y$.
	\end{definition}
	\noindent Note that the bit-flip operator $F_i$ is a permutation on $[n]$.

\paragraph{The construction.}
	We start by providing our lower bound construction $\calD_{n,2}$, supported on sequences that are far from $(12)$-free.

	Let $f^{\downarrow} \colon [n] \to [n]$ denote the (unique) decreasing permutation on $[n]$, i.e., the function $f^{\downarrow}(x) = n+1-x$ for any $x \in [n]$.
	For any $i \in [\log n]$, define $f_i \colon [n] \to [n]$ to be the composition of $f^{\downarrow}$ with the bit-flip operator $F_i$, that is, $f_i(x) = f^{\downarrow}(F_i(x))$ for any $x \in [n]$. Note that $f_i$ is a permutation, as a composition of permutations. See Figure~\ref{fig:mon-construct} for a visualization of the construction. Finally, define $\calD_{n,2}$ as the uniform distribution over the sequences $f_1, f_2, \ldots, f_{\log n}$.

	The next lemma characterizes the set of all $(1,2)$-patterns in $f_i$. 
	\begin{lemma}
		\label{lem:monotone_copies}
		Let $i \in [\log n]$. A pair $x < y \in [n]$ forms a $(1,2)$-copy in $f_i$ if and only if $M(x,y) = i$.	
	\end{lemma}
	\begin{proof}
		Let $x < y \in [n]$. If $M(x,y) > i$, then $F_i(x) < F_i(y)$ holds and so $f_i(x) = f^{\downarrow}(F_i(x)) > f^{\downarrow}(F_i(y)) = f_i(y)$, implying that $(x,y)$ is not a $(1,2)$-copy. If $M(x,y) < i$ then $x$ and $y$ share the bit in index $i$ of the binary representation, and thus flipping it either adds $2^{i-1}$ to both $x$ and $y$ or decreases $2^{i-1}$ from both of them. In both cases, $F_i(x) < F_i(y)$, and like the previous case we get $f_i(x) > f_i(y)$. 
		Finally, if $M(x,y)=i$ then one can write $x = z + 0 \cdot 2^{i-1} + x'$ and $y = z + 1 \cdot 2^{i-1} + y'$, where $z$ corresponds to the $\log n - i$ most significant bits in the binary representation (which are the same in $x$ and $y$), and $x', y' < 2^{i-1}$ correspond to the $i-1$ least significant bits. Therefore, $F_i(x) = z + 1 \cdot 2^{i-1} + x' > z + 0 \cdot 2^{i-1} + y' = F_i(y)$ and thus $f_i(x) = f^{\downarrow}(F_i(x)) < f^{\downarrow}(F_i(y)) = f_i(y)$, as desired.
	\end{proof}

	We conclude that each of the sequences $f_i$ is $(1/2)$-far from $(12)$-free.
	\begin{lemma}
		\label{lem:monotonicity_construction}
		For any $i \in [\log n]$, the sequence $f_i$ contains a collection $\calC$ of $n/2$ disjoint $(1,2)$-copies.
	\end{lemma}
	\begin{proof}
		For any $x \in [n]$ whose binary representation $B_n(x) = (b^x_1, \ldots, b^x_{\log n})$ satisfies $b^x_i = 0$, we have $M(x, F_i(x)) = i$. By Lemma~\ref{lem:monotone_copies}, $(x, F_i(x))$ is thus a $(1,2)$-copy. Picking
		$$
			\calC = \{ (x, F_i(x)) \ : \ x \in [n],\  b^x_i = 0 \},
		$$
		and noting that the pairs in $\calC$ are disjoint, the proof follows.
	\end{proof}

\paragraph{Binary Profiles.}
	We now formally define our notion of \emph{binary profiles}, and describe why they are useful for proving lower bounds for problems of this type.

	\begin{definition}[Binary profiles captured]
		Let $n \in \N$ be a power of $2$ and let $Q \subseteq [n]$.
		The set of \emph{binary profiles} captured by $Q$ is defined as
		$$
			\bprof(Q) = \{i \in [\log n] \  : \ \text{there exist $x,y \in Q$ satisfying $M(x,y)=i$} \}.
		$$
	\end{definition}

	The next lemma asserts that the number of binary profiles that set captures does not exceed (or even match) the size of the set.

	\begin{lemma}
		\label{lem:binary_dist_prof}
		Let $Q \subseteq [n]$ be a subset of size $q > 0$. Then $|\bprof(Q)| \leq q-1$.
	\end{lemma}
	\begin{proof}
		We proceed by induction on $q$. For $q \leq 2$, the statement clearly holds. Otherwise, let $i_{\text{max}} = \max \bprof(Q)$ be the maximum index of a bit in which two elements $x,y \in Q$ differ.
		For $j=0,1$, define 
		$$
			Q_j = \{x \in Q \ :\  \text{the binary representation of $x$ is $B_n(x) = (b^x_1, \ldots, b^x_{\log n})$, and $b^x_{i_{\text{max}}} = j$}\}.
		$$
		\noindent
		Clearly, for any $x \in Q_0$ and $y \in Q_1$, we have $M(x,y) = i_{\text{max}}$. We can therefore write $\bprof(Q)$ as  
		$$
			\bprof(Q) = \bprof(Q_0) \cup \bprof(Q_1) \cup \{i_{\text{max}}\},
		$$
		from which we conclude that
		$$
			|\bprof(Q)| \leq |\bprof(Q_0)| + |\bprof(Q_1)| + 1 \leq |Q_0| - 1 + |Q_1| - 1 + 1 = |Q| - 1,
		$$
		where the second inequality follows from the induction hypothesis.
	\end{proof}

\paragraph{Proof for the case $k=2$ using binary profiles.}
	After collecting all the ingredients required to prove the case $k=2$ of Theorem~\ref{thm:lower_bound}, we now conclude the proof.
	Fix $0 < p < 1$, let $n$ be a power of two, and consider the distribution $\calD_{n,2}$ defined above, supported on sequences that are $(1/2)$-far from $(12)$-free (see Lemma~\ref{lem:monotonicity_construction}).
	Let $Q \subseteq [n]$ be any subset of size at most $p \log n$. It suffices to show that, for $\boldf \sim \calD_{n,2}$, the probability that $Q$ contains a $(12)$-copy in $\boldf$ is less than $p$. By Lemma~\ref{lem:monotone_copies}, $Q$ contains a $(12)$-copy with respect to $f_i$ if and only if $i \in \bprof(Q)$. Thus, the above probability is equal to $|\bprof(Q)| / \log n$, which, by Lemma~\ref{lem:binary_dist_prof}, is at most $(|Q|-1) / \log n < p$, as desired.

\subsection{Hierarchical binary profiles and the lower bound}
	\label{subsec:lower_bound_full_proof}
	To prove Theorem~\ref{thm:lower_bound} in its full generality, we significantly extend the proof presented in Subsection~\ref{subsec:lower_bound_monotonicity} for the case $k=2$, relying on a generalized hierarchical (and more involved) notion of a binary profile. 
	Let $n > k \geq 2$ be powers of $2$, and write $k = 2^h$ (so $h \in \N$). 
	We show that there exist $\binom{\log_2 n}{h} = \binom{\log_2 n}{\log_2 k}$ different types of \emph{binary $h$-profiles} (see Definition~\ref{def:h-profile}) with the following properties. First, a subset $Q \subseteq [n]$ can capture at most $|Q|-1$ such profiles (Lemma~\ref{lem:lower_bound_general_construction} below, generalizing Lemma~\ref{lem:monotonicity_construction}); and second, for each such profile there exists a sequence  (in fact, a permutation) that is $(1/k)$-far from $(12\ldots k)$-free, such that any set of queries $Q$ that finds $(12\ldots k)$-pattern with respect to this sequence must capture the given profile (Lemma~\ref{lem:hierarchical_copies} below, generalizing Lemma~\ref{lem:monotone_copies}).

	\paragraph{Hierarchical binary profiles.}
		While the proof for the case $k=2$ relied on a rather basic variant of a binary profile, our lower bound for general $k$ requires a more sophisticated, hierarchical type of profile, described below.
		\begin{definition}[binary $h$-profiles] \label{def:h-profile}
			Let $(x_1, \dots, x_k) \in [n]^{k}$ be a $k$-tuple of indices satisfying $x_1 < \dots < x_k$.  For an $h$-tuple $(i_1, \dots, i_h) \in [\log_2 n]^h$ satisfying $i_1 < \dots < i_k$, 
			we say that $(x_1, \ldots, x_k)$ has \emph{$h$-profile of type $(i_1,\ldots, i_{h})$} if, 
			\[ 
				M(x_j, x_{j+1}) = i_{M(j-1, j)} \qquad \text{for every $j \in [k-1]$}.
			\]
		\end{definition}
		For example, when $h=3$ (and $k=8$), a tuple $(x_1, \ldots, x_8) \in [n]^8$ with $x_1 < \ldots < x_8$ has binary $3$-profile of type $(i_1, i_2, i_3)$ if the sequence $(M(x_j, x_{j+1}))_{j=1}^7$ is $(i_1, i_2, i_1, i_3, i_1, i_2, i_1)$.
		See Figure~\ref{fig:profiles} for a visual depiction of such a binary $3$-profile.

		Similarly to the case $k=2$, given a set of queries $Q \subseteq [n]$, we shall be interested in the collection of $h$-profiles captured by $Q$. 
		\begin{definition}[Binary $h$-profiles captured]
			Let $n \geq k \geq 2$ be powers of $2$ where $k = 2^h$.
			For any $Q \subseteq [n]$, we denote the set of all $h$-profiles captured by $Q$ by 
			\[
				\bprof_h(Q) = \left\{(i_1, \ldots, i_{h})  :\  
					\begin{array}{l} \text{there exist $x_1, \dots,x_k \in Q$ where 	}  x_1 < \dots < x_k \\
						\text{and $(x_1, \dots, x_k)$ has $h$-profile of type $(i_1, \ldots, i_{h})$} 
					\end{array} 
				\right\}.
			\]
		\end{definition}
		The next lemma is one of the main ingredients of our proof, generalizing Lemma~\ref{lem:binary_dist_prof}. It shows that a set $Q$ of queries cannot capture $|Q|$ or more different $h$-profiles.

		\begin{lemma}
			\label{lem:tree_prof_lower_bound}
			Let $h, n \in \N$ where $n \geq 2^h$ is a power of $2$. For any $\emptyset \neq Q \subseteq [n]$, we have $|\bprof_h(Q)| \leq |Q|-1$.
		\end{lemma}
		\begin{proof}
			We proceed by induction on $h$. The case $h=1$ was settled in Lemma~\ref{lem:binary_dist_prof}. Suppose now that $h > 1$, and define 
			$$\emptyset = B_{\log n + 1} \subseteq B_{\log n} \subseteq \ldots \subseteq B_{1} = Q$$
			 as follows.  Set $B_{\log n + 1} = \emptyset$, and given $B_{i+1}$, define the set $B_i \supseteq B_{i+1}$ as an arbitrary maximal subset of $Q$ containing $B_{i+1}$ which does not have two elements with $M(x,y) < i$. 
			
			Additionally, for each $j \in [\log_2 n]$, define
			\[ N_j = \left\{ (i_2, \dots, i_h) : 1 \leq j < i_2 \dots < i_h \leq \log_2 n \text{ and } (j, i_2, \dots, i_h) \in \bprof_h(Q) \right\}. \]
			
			\begin{claim}
				Let $j < i_2 < \ldots < i_h \in [\log n]$, and suppose that $(j, i_2, \ldots, i_h) \in \bprof_h(Q)$. Then $(j, i_2, \ldots, i_h) \in \bprof_h(B_j)$.
			\end{claim}
			\begin{proof}
				Suppose that a tuple $(x_1, \ldots, x_k)$ with $x_1 < \dots < x_k \in Q$ has $h$-profile $(j, i_2, \ldots, i_{h})$. By the maximality of $B_j$, we know that for every $1 \leq \ell \leq k$ there exists $y_\ell \in B_j$ such that either $x_{\ell} = y_{\ell}$ or $M(x_{\ell}, y_{\ell}) < j$. Indeed, if this were not the case, then $B'_j \eqdef B_j \cup \{x_{\ell}\}$ would be a set that strictly contains $B_{j}$ and does contain two elements $x \neq y$ with $M(x,y) = j$, a contradiction to the maximality of $B_{j}$. 
				By definition of a profile, we conclude that $\{y_1, \ldots, y_{k}\} \subseteq B_j$ has $h$-profile $(j, i_2, \ldots, i_{h})$.
			\end{proof}
			\begin{claim}
				For any $j \in [\log n]$, we have $N_j \subseteq \bprof_{h-1}(B_j \setminus B_{j+1})$.
			\end{claim}
			\begin{proof}
				Suppose that $(i_2, \ldots, i_h) \in N_j$, then $(j, i_2, \ldots, i_h) \in \bprof_h(Q)$. By the previous lemma, we know that $(j, i_2, \ldots, i_h) \in \bprof_h(B_j)$. 
				Therefore, there exists a tuple $(y_1, \dots, y_k)$ where $y_1 < \ldots< y_{k} \in B_j$, that has $h$-profile of type $(j, i_2, \ldots, i_{h})$. 
				
				For any $t \in k/2$, it holds that $M(y_{2t-1}, y_{2t}) = j$. Therefore, at most one of $y_{2t-1}, y_{2t}$ is in $B_{j+1}$, and hence, for any such $t$ there exists $z_t \in \{y_{2t-1}, y_{2t}\} \setminus B_{j+1} \subseteq B_{j} \setminus B_{j+1}$. Consider the tuple $(z_1, \ldots, z_{k/2})$, whose elements are contained in $B_j \setminus B_{j+1}$. It follows from our choice of $z_t$ that $M(z_t, z_{t+1}) = M(y_{2t}, y_{2t+2})$ for any $t \in [k/2]$, from which we conclude that $(z_1, \ldots, z_{k/2})$
				has $(h-1)$-profile $(i_2, \ldots, i_{h})$. In other words, $(i_2, \ldots, i_h) \in \bprof_h(B_j \setminus B_{j+1})$, as desired.
			\end{proof}
			
			We are now ready to finish the proof of Lemma~\ref{lem:tree_prof_lower_bound}. Observe that $\bprof_h(Q)$ and $Q$ can be written as the following disjoint unions:
		\[
			\bprof_h(Q) = \bigcup_{j=1}^{\log_2 n} \{(j, i_2, \ldots, i_h) : (i_2, \ldots, i_h) \in N_j \}
			\qquad \text{and}  \qquad 
			Q = \bigcup_{j=1}^{\log_2 n} (B_j \setminus B_{j+1}).
		\]
			It follows from the last claim and the induction assumption that 
			\begin{equation}
				\label{eqn:N_j}
				|N_j| \leq |\bprof_{h-1}(B_j \setminus B_{j+1})| \le |B_j \setminus B_{j+1}|,
			\end{equation}
			where for $j$ with $N_j \neq \emptyset$ there is a strict inequality.
			Now,
			if $N_j$ is empty for all $j$ then, trivially, $|\bprof_h(Q)|=  0 \leq |Q|-1$. Otherwise, there exists some non-empty $N_j$, for which \eqref{eqn:N_j} yields a strict inequality, and we get
			\[
			|Q| = \sum_{j=1}^{\log_2 n} |B_j \setminus B_{j+1}| >\sum_{j=1}^{\log_2 n} |N_j| = |\bprof_h(Q)|, 
			\]
			establishing the proof of the Lemma~\ref{lem:tree_prof_lower_bound}.
		\end{proof} 

	\paragraph{The construction.}
		For any $i_1 < i_2 < \ldots < i_h \in [\log n]$, we define $f_{i_1, \ldots, i_h} \colon [n] \to [n]$ as
		\[ 
			f_{i_1, \ldots, i_{h}} \eqdef f^{\downarrow} \circ F_{i_h} \circ \ldots \circ F_{i_1},
		\]
		where, as before, $\circ$ denotes function composition. In other words, for any $x \in [n]$ we have $f_{i_1, \ldots, i_h}(x) = f^{\downarrow}(F_{i_{h}}(F_{i_{h-1}}(\dots(F_{i_1}(x)\dots))))$. Note that $f_{i_1, \ldots, i_h}$ is indeed a permutation, as a composition of permutations.
		(See Figure~\ref{fig:construct-recurse}, which visually describes the construction of $f_{i_1, \ldots, i_h}$ recursively, as a composition of $F_{i_h}$ with $f_{i_1, \ldots, i_{h-1}}$.)
		We take $\calD_{n,k}$ to be the uniform distribution over all sequences of the form $f_{i_1, \ldots, i_h}$ with $i_1 < i_2 < \ldots < i_k$. The size of the support of $\calD_{n,k}$ is $\binom{\log_2 n}{h} = \binom{\log_2 n}{\log_2 k}$.

	\paragraph{Structural properties of the construction.}
		Recall that our lower bound distribution $\calD_{n,k}$ is supported on the family of permutations $f_{i_1, \ldots, i_{h}}$, where $i_1 < \ldots < i_h \in [\log n]$, described above.
		We now turn to show that these $f_{i_1, \ldots, i_h}$ satisfy two desirable properties. First, to capture a $(12\dots k)$-copy in $(f_{i_1, \ldots, i_h})$, our set of queries $Q$ must satisfy $(i_1, \ldots, i_h) \in \bprof_h(Q)$ (Lemma~\ref{lem:hierarchical_copies}). And second, each such $f_{i_1, \ldots, i_k}$ is $(1/k)$-far from $(12\dots k)$-free (Lemma~\ref{lem:lower_bound_general_construction}). 

		\begin{lemma}
			\label{lem:hierarchical_copies}
			Let $(x_1, \ldots, x_k) \in [n]^k$ be a $k$-tuple where $x_1 < \ldots < x_k$, and let $f = f_{i_1, \ldots, i_{h}}$ be defined as above. Then $f(x_1) < f(x_2) < \ldots < f(x_k)$ (i.e., $(x_1, \ldots, x_k)$ is a $(12\dots k)$-copy with respect to $f_{i_1, \ldots, i_h}$) if and only if  $(x_1, \ldots, x_k)$ has binary $h$-profile of type $(i_1, i_2 \ldots, i_{h})$. Furthermore, $f_{i_1, \ldots, i_{h}}$ does not contain increasing subsequences of length $k+1$ or more.
		\end{lemma}

		\begin{proof}
			The proof is by induction on $h$, with the base case $h=1$ covered by Lemma~\ref{lem:monotone_copies}; in particular, it follows from Lemma~\ref{lem:monotone_copies} that $f_{i}$ has no increasing subsequence of length $3$, since there exist no $x < y < z \in [n]$ with $M(x,y) = M(y,z) = i$. 
			
			For the inductive step, we need the following claim, which generalizes Lemma~\ref{lem:monotone_copies}. 
			
			\begin{claim}
				\label{claim:lower_bigger_general}
				A pair $x < y \in [n]$ satisfies $f_{i_1, \ldots, i_h}(x) < f_{i_1, \ldots, i_h}(y)$ if and only if $M(x,y) \in \{i_1, \ldots, i_h\}$.
			\end{claim}
			\begin{proof}
				Let $F_{i_1, \ldots, i_h} = F_{i_h} \circ \ldots \circ F_{i_1}$. 	
				Since $f_{i_1, \ldots, i_h} = f^{\downarrow} \circ F_{i_1, \ldots, i_h}$, it suffices to show that $F_{i_1, \ldots, i_h}(x) > F_{i_1, \ldots, i_h}(y)$ if any only if $M(x, y) \in \{i_1, \ldots, i_h\}$. To do so, we prove the following two statements.
				\begin{itemize}
					\item For any $x < y \in [n]$, $F_i(x) > F_i(y)$ if and only if $M(x,y) = i$.
					\item For any $x < y \in [n]$, $M(F_i(x), F_i(y)) = M(x,y)$. 
				\end{itemize}
				Indeed, using these two statements, the proof easily follows by induction: the value of $M(x,y)$ never changes regardless of which bit-flips we simultaneously apply to $x$ and $y$. Now, applying any of the bit-flips $F_i$ to $x$ and $y$, where $i \neq M(x,y)$, does not change the relative order between them, while applying $F_{M(x,y)}$ does change their relative order. This means that a change of relative order occurs if and only if $M(x,y) \in \{i_1, \ldots, i_h\}$, which settles the claim.
				
				The proof of the first statement was essentially given, word for word, in the proof of Lemma~\ref{lem:monotone_copies}. The second statement follows 
				by a simple case analysis of the cases where $i$ is bigger than, equal to, or smaller than $M(x,y)$, showing that in any of these cases, $M(F_i(x), F_i(y)) = M(x,y)$.
			\end{proof}
			Suppose now that $(x_1, \ldots, x_k) \in [n]^k$ is a tuple with $x_1 < \ldots < x_k$ and a binary $h$-profile of type $(i_1, \ldots, i_h)$ is a $(12 \dots k)$-copy in $f_{i_1, \ldots, i_h}$. By definition of a binary $h$-profile, we have that $M(x_j, x_{j+1}) \in \{i_1, \ldots, i_h\}$ for any $j \in [k-1]$, which, by the claim, implies that $f_{i_1, \ldots, i_h}(x_j) < f_{i_1, \ldots, i_h}(x_{j+1})$. It thus follows that $(x_1, \ldots, x_k)$ is a $(12 \dots k)$-copy in $f_{i_1, \ldots, i_h}$, as desired.
			
			Conversely, suppose that a tuple $(x_1, \ldots, x_k) \in [n]^k$ with $x_1 < \ldots < x_k$ is a $(12 \dots k)$-copy in $f_{i_1, \ldots, i_k}$. 
			We need to show that $(x_1, \ldots, x_{k})$ has binary $h$-profile of type $(i_1, \ldots, i_h)$, that is, $M(x_j, x_{j+1}) = i_{M(j-1, j)}$ for every $j \in [k-1]$.
			Define $r = \argmax_j \{M(x_j, x_{j+1})\}$, and note that $r$ is unique; otherwise, we would have $x < y < z \in [n]$ so that $M(x,y) = M(y,z)$, a contradiction. 
			
			\begin{claim}
				$M(x_r, x_{r+1}) = i_h$.
			\end{claim}
			\begin{proof}
				By Claim~\ref{claim:lower_bigger_general}, we know that $M(x_r, x_{r+1}) \in \{i_1, \ldots, i_h\}$. Suppose to the contrary that $M(x_r, x_{r+1}) \leq i_{h-1}$. Then, $M(x_j, x_{j+1}) \leq M(x_r, x_{r+1}) \leq i_{h-1}$ for every $j \in [k-1]$, and by Claim~\ref{claim:lower_bigger_general}, for any $j \in [k-1]$ we have $f_{i_1, \ldots, i_{h-1}}(x_j) \leq f_{i_1, \ldots, i_{h-1}}(x_{j+1})$, that is, $(x_1, x_2, \ldots, x_k)$ is a $(12 \dots k)$-copy in $f_{i_1, \ldots, i_{h-1}}$. This contradicts the last part of the inductive hypothesis.
			\end{proof}
			
			\begin{claim}
				$r = k/2$.
			\end{claim}
			\begin{proof}
				Without loss of generality, suppose to the contrary that $r > k/2$ (the case where $r< k/2$ is symmetric). As the tuple $(x_1, \ldots, x_r)$ is an increasing subsequence for $f_{i_1, \ldots, i_h}$, we have $M(x_j, x_{j+1}) \in \{i_1, \ldots, i_h\}$ for any $j \in [r-1]$. By the maximality and uniqueness of $r$, $M(x_j, x_{j+1}) < i_h$ for any $j \in [r-1]$. Thus, it follows from Claim~\ref{claim:lower_bigger_general} that $(x_1, \ldots, x_r)$ is a $(12\dots r)$-copy in $f_{i_1, \ldots, i_{h-1}}$, contradicting the last part of the inductive hypothesis.
			\end{proof}
			It thus follows from the two claims that $M(x_{k/2}, x_{k/2+1}) = i_h$. Since $M(x_j, x_{j+1}) \in \{i_1, \ldots, i_{h-1}\}$ for any $j \in [k-1] \setminus \{k/2\}$, we conclude, again from Claim~\ref{claim:lower_bigger_general}, that $(x_1, \ldots, x_{k/2})$ and $(x_{k/2+1}, \ldots, x_k)$ both induce length-$(k/2)$ increasing subsequences in $f_{i_1, \ldots, i_{h-1}}$. By the inductive hypothesis, they both have binary $(h-1)$-profile $(i_1, \ldots, i_{h-1})$. Combined with the last two claims, we conclude that $(x_1, \ldots, x_k)$ has binary $h$-profile $(i_1, \ldots, i_h)$, as desired.
			
			It remains to verify that $f_{i_1, \ldots, i_h}$ does not contain an increasing subsequence of length $k+1$. If, to the contrary, it does contain one, induced on some tuple $(x_1, \ldots, x_{k+1}) \in [n]^{k+1}$ where $x_1 < \ldots < x_{k+1}$, then, applying the last two claims to the length-$k$ two tuples $(x_1, \ldots, x_k)$ and $(x_2, \ldots, x_{k+1})$, we conclude that $M(x_{k/2}, x_{k/2+1}) = M(x_{k/2+1}, x_{k/2 +2}) = i_h$. However, as discussed above, there cannot exist $x<y<z \in [n]$ with $M(x,y) = M(y,z)$ -- a contradiction.
		\end{proof}

		It remains to prove that each $f_{i_1, \ldots, i_h}$ is indeed $(1/k)$-far from $(12\dots k)$-free. After we spent quite some effort to characterize \emph{all} $(12 \dots k)$-copies in $f_{i_1, \ldots, i_h}$, this upcoming task is much simpler.
		\begin{lemma}
			\label{lem:lower_bound_general_construction}
			Let $n \geq k \geq 2$ be powers of two and write $k = 2^h $.
			The sequence $f_{i_1, \ldots, i_h} \colon [n] \to [n]$, defined above, contains $n / k$ disjoint $(12 \dots k)$-copies.
		\end{lemma}
		\begin{proof}
			Fix $i_1 < \ldots < i_h$ as in the statement of the lemma. We say that $x,y \in [n]$ with binary representations $B_n(x) = (b^x_1, \ldots, b^x_{\log n})$ and $B_n(y) = (b^y_1, \ldots, b^y_{\log n})$ are \emph{$(i_1, \ldots, i_h)$-equivalent} if $b^x_i = b^y_i$ for any $i \in [\log n] \setminus \{i_1, \ldots, i_{h}\}$. Clearly, this is an equivalence relation, partitioning $[n]$ into $n/k$ equivalence classes, each of size exactly $k = 2^h$. Moreover, it is straightforward to verify that the elements $x_1 < x_2 < \ldots < x_k$ of any equivalence class satisfy $M(x_j, x_{j+1}) \in \{i_1, \ldots, i_h\}$ for any $j \in [k-1]$, and thus, by Claim~\ref{claim:lower_bigger_general}, $(x_1, \ldots, x_k)$ constitutes a $(12 \dots k)$-copy in $f_{i_1, \ldots, i_k}$.  
		\end{proof}

	\paragraph{Proof of Theorem~\ref{thm:lower_bound}.}
		It now remains to connect all the dots for the proof of Theorem~\ref{thm:lower_bound}. 
	\begin{proof}
		Fix $0 < p < 1$, let $n \geq k$ be powers of $2$, and write $k = 2^h$.
		As before, we follow Yao's minimax principle \cite{Y77}, letting $\calD_{n,k}$ be the uniform distribution over all $\binom{\log_2 n}{h} = \binom{\log_2 n}{\log_2 k}$ sequences (in fact permutations) $f_{i_1, \ldots, i_h} \colon [n] \to [n]$, where $i_1 < \ldots < i_h \in [\log n]$. Recall that, by Lemma~\ref{lem:lower_bound_general_construction}, this distribution is supported on sequences that are $(1/k)$-far from $(12\dots k)$-free. 

		It suffices to show that, for $\boldf \sim \calD_{n,k}$, the probability for any subset $Q \subseteq [n]$ of size at most $p \binom{\log_2 n}{h}$ to capture a $(12 \dots k)$-copy in $\boldf$ is less than $p$.
		Indeed, by Lemma~\ref{lem:hierarchical_copies}, $Q$ captures a copy in $f_{i_1, \ldots, i_h}$ if and only if $(i_1, \ldots, i_h) \in \bprof_h(Q)$, so the success probability for any given $Q$ is exactly $|\bprof_h(Q)| / \binom{\log_2 n}{h} < |Q| / \binom{\log_2 n}{h} \leq p$ for any $Q \subseteq [n]$ with $|Q| \leq p  \binom{\log_2 n}{h}$, where the first inequality follows from Lemma~\ref{lem:tree_prof_lower_bound}. The proof of Theorem~\ref{thm:lower_bound} follows.
	\end{proof}
\begin{flushleft}
\bibliographystyle{alpha}
\bibliography{waingarten}
\end{flushleft}

\end{document}